\documentclass{article}

\usepackage{titling}
\usepackage{authblk}
\usepackage{titlesec}

\usepackage{subcaption}
\usepackage{graphicx}

\usepackage{amsmath}               
\usepackage{amsfonts}              
\usepackage{amsthm}                
\usepackage{bm}
\usepackage{stmaryrd}
\usepackage{float}
\usepackage{changes}
\usepackage{enumitem}
\usepackage{orcidlink}

\newtheorem{theorem}{Theorem}[section]
\newtheorem{definition}{Definition}[section]

\newtheorem{proposition}{Proposition}[section]

\newtheorem{lemma}{Lemma}[section]
\newtheorem{remark}{Remark}[section]

\titleformat*{\section}{\Large\bfseries}
\titleformat*{\subsection}{\large\bfseries}
\titleformat*{\subsubsection}{\large\bfseries}
\titleformat*{\paragraph}{\large\bfseries}
\titleformat*{\subparagraph}{\large\bfseries}

\newcommand{\R}{\mathbb{R}}

\allowdisplaybreaks

\title{ Exact moment methods for conservation laws in phase space}
\author[1,*]{Tileuzhan Mukhamet \orcidlink{0000-0003-3811-1040}}
\author[1]{Katharina Kormann \orcidlink{0000-0003-1956-2073}}
\affil[1]{Faculty of Mathematics, Ruhr-University Bochum, Germany}
\affil[*]{Corresponding author: Tileuzhan Mukhamet, tileuzhan.mukhamet@rub.de}
\date{\today}

\begin{document}
\maketitle	

\noindent \textbf{Abstract:} 
We construct exact moment models for a class of conservation laws based on a parametrization of the distribution functions by its moments as proposed in [Burby, Sci. Rep. 13 (2023)]. In particular, we derive moment models for the Vlasov--Maxwell and the relativistic Vlasov--Maxwell systems. The new models can be viewed as a higher degree extension of the Direct Quadrature Method of Moments (DQMOM). We propose slicing moment method for the representation of the distribution function. For simulations with strong kinetic effects, we utilize a hybrid formulation that transforms adaptively a fraction of fluid nodes into particles in order to capture multi-streaming. We demonstrate convergence and performance of our strategies on the weak and strong Landau damping and two-stream instability experiments in non-relativistic and relativistic limits. 



\section{Introduction} 
One of the hardest problems in computational plasma physics is the efficient modeling of plasma at large Knudsen numbers. As the particle mean free path approaches or exceeds the system size, kinetic effects become important. Kinetic effects play important role in space physics and astrophysics \cite{palmroth2025vlasov} and fusion energy studies \cite{Rinderknecht_2018}. 

The distribution of charged particles in phase space is described by kinetic equations such as the Boltzmann and the Vlasov equations. Due to the high-dimensionality of the phase space solving these models directly however is extremely expensive. 

Moment-based models are commonly used to reduce the dimensionality of the problem to the spatial coordinates only. Moment based models can capture various temporal and spatial scales depending on the number of moments. Every additional moment provides information about a state of the parent model. Such moment models however suffer from the closure problem that necessitates an assumption about the distribution function in order to truncate the hierarchy of moments and close the system of equations. 

Classical moment models, such as Grad's models, use closures that are valid near the thermodynamic equilibrium  \cite{issan}  where collisions are frequent and the Knudsen number is small ($Kn<1$). As the non-Maxwellian features arise, the underlying assumptions of the classical models are violated and important kinetic effects are missing. There are certain classical closures that approximate specific kinetic effects, such as the linear Landau damping \cite{Hammett,hunana2019introductory}. Despite their wide-spread use, these closures have a limited range of applicability as they are not designed to capture fine-scale physics in general settings.

Recently a new closure was proposed by Burby \cite{burby2023}. The closure assumes distribution function parametrized in terms of the Dirac-delta and its derivatives. 
\begin{align} \label{eq:burby_ansatz}
&g(\bm u, \bm x, t) =  C^0(\bm x, t) \delta(\bm u -  \bm v(\bm x, t)) - C^1_{i_1} (\bm x, t) \partial_{u_{i_1}} \delta (\bm u -  \bm v(\bm x, t)) \\ \notag
& + \frac{1}{2} C^2_{i_1 i_2} (\bm x, t) \partial_{u_{i_1}} \partial_{u_{i_2}} \delta (\bm u -  \bm v(\bm x, t)) + ... + \frac{(-1)^{k}}{k!} C^{k}_{i_1 ... i_{k}} (\bm x, t) \partial_{u_{i_1}} ... \partial_{u_{i_{k}}} \delta (\bm u - \bm v(\bm x, t)),
\end{align}
where
\begin{align} \label{eq:centered-moments}
\bm C^k(\bm x, t) = \int  (\bm u - \bm v(\bm x, t))^{\otimes k}  g(\bm u, \bm x, t) du,
\end{align}
are the centered moments, $\bm v(\bm x, t)$ is the abscissa variable that determines the center of the moments, and $\otimes k$ denotes the $k$-fold outer product, i.e., $\bm u^{\otimes k} = \underbrace{\bm u \otimes ... \otimes \bm u}_{k-times}$.
It can be checked that the velocity moments up to the specified degree are reproduced by the ansatz. Thus for example the first $k$-moments of arbitrary non-Maxwellian distributions, such as the two-stream instability, can be parametrized by the ansatz. 

The parametrization \eqref{eq:burby_ansatz} can be viewed as an extension of the Direct Quadrature Method of Moments (DQMOM) developed by Marchisio and Fox \cite{MARCHISIO200543}. Since its introduction, DQMOM and its variants have been applied to the population balance equation and predict various industrial processes in chemical engineering such as aerosol engineering, gas-liquid dispersion, liquid-liquid dispersion, and many more (see \cite{LI2019483} and references therein). The method theoretically can model arbitrary distributions, non-equilibrium processes at arbitrary Knudsen numbers using only a fraction of the cost of the fully kinetic solver. DQMOM assumes the distribution function in the following form:
\begin{align} \label{eq:DQMOM}
g(\bm u, \bm x, t) = \sum_{p=1}^N n_p(\bm x, t) \delta(\bm u - \bm \xi_p(\bm x, t)),
\end{align}
where $p$ is the node index, $n_p(\bm x,t)$ is the weight of the node, and $\bm \xi_p(\bm x,t)$ is the abscissa of the node. The DQMOM then determines the evolution equations for the weights and the abscissas such that a certain number of moment equations of the parent model are satisfied. 

The DQMOM ansatz \eqref{eq:DQMOM} for a single node $N=1$, is clearly identical to the ansatz \eqref{eq:burby_ansatz} of degree $k=1$. That is the weight $n_p(\bm x, t)$ can be viewed as the zero-th order moment $C^0(\bm x, t)$. The additional terms in \eqref{eq:burby_ansatz} allow to consider higher order moments that are absent in the DQMOM. In analogy to the DQMOM ansatz \eqref{eq:DQMOM}, one can therefore introduce the multi-node version based on the ansatz \eqref{eq:burby_ansatz}. For this first introduce $g_{p}$ for the node $p$.
\begin{align} \label{eq:burby_ansatz_node_p}
&g_{p} =  \left[ C_p^0(\bm x, t) \delta(\bm u -  \bm v_p(\bm x, t)) - C^1_{p,i_1} (\bm x, t) \partial_{u_{i_1}} \delta (\bm u -  \bm v_p(\bm x, t)) \right. \\ \notag
& \left. + \frac{1}{2} C^2_{p,i_1 i_2} (\bm x, t) \partial_{u_{i_1}} \partial_{u_{i_2}} \delta (\bm u -  \bm v_p(\bm x, t)) + ... + \frac{(-1)^{k}}{k!} C^{k}_{p,i_1 ... i_{k}} (\bm x, t) \partial_{u_{i_1}} ... \partial_{u_{i_{k}}} \delta (\bm u - \bm v_p(\bm x, t)) \right].
\end{align}
then the multi-node ansatz reads
\begin{align} \label{eq:burby_ansatz_multinode}
&g(\bm u, \bm x, t) = \sum_{p=1}^N g_{p}
\end{align}
The multi-node version considers the distribution as a sum of $N$ sub-distributions each centered around the abscissas $\bm v_p$ in the velocity space. The moments $\bm C_p^k$ account for the shape of each sub-distribution with respect its corresponding abscissa $\bm v_p$. 
\begin{align}
\bm C_p^k = \int  (\bm u - \bm v_p(\bm x, t))^{\otimes k}  g_{p}(\bm u, \bm x, t) du.
\end{align}
From the theory of distributions it is known that arbitrary distributions can be represented uniquely by the finite linear combination of the Dirac $\delta$ and its derivatives provided that the distribution has its support concentrated at a single point \cite[Theorem 3.2.1]{Friedlander1982}. One can imagine that as the number of nodes goes to infinity ($N \rightarrow \infty$) each sub-distribution shrinks to a single point in velocity space and can be captured accurately by the $\delta$ and its derivatives. 

Both DQMOM and the method in \cite{burby2023} lead to evolution equations for the moments and abscissas. The difference is that DQMOM seeks the weights and abscissas such that the evolution equations for a certain number of moments of the parent kinetic model are satisfied. The method in \cite{burby2023} constructs the equations for the moments and abscissas such that the ansatz \eqref{eq:burby_ansatz} satisfies the Vlasov--Poisson system exactly. This exactness of the models along with the higher degree approximation of the ansatz \eqref{eq:burby_ansatz_node_p} can potentially provide an accurate approximation of the parent model.  

There are several major challenges associated with the method in \cite{burby2023}. First, the method was presented for the Vlasov--Poisson using highly technical Lie-theoretic analysis. The procedure assumes background in geometric mechanics which may not be known to the applied physics and engineering communities. So a simpler method to construct this type of moment models for the broader class of conservation laws is needed. In particular, such exact moment models for the relativistic Vlasov--Maxwell system are not known. Thus, the \textit{first contribution} of this work is the simple procedure to construct exact moment models based on the ansatz \eqref{eq:burby_ansatz_multinode} for a general class of conservation laws. This procedure is applied to the Vlasov--Maxwell as well as the relativistic Vlasov--Maxwell systems to obtain exact moment models.

The practical usage of the models is also associated with several challenges. The moment models based on the ansatz \eqref{eq:burby_ansatz} require reduction to the center manifold in order to obtain physical solutions \cite{burby2023}. This reduction can be carried out analytically \cite{BURBY2020105289, burby2017magnetohydrodynamic} or numerically \cite{gonzalez2022finite}. The lowest degree model that corresponds to the ansatz \eqref{eq:DQMOM} however does not require reduction and for practical application we develop ideas for this case leaving the higher-degree models for future studies.

Thus \textit{the second contribution} of this work concerns the practical usage of the lowest degree method. The major challenge is to determine the initial conditions for the dynamic variables $n_p$ and $\bm v_p$ in the ansatz \eqref{eq:DQMOM}. The initial conditions for $n_p$ and $\bm v_p$ can be computed using the moment inversion algorithms \cite[Chapter 3]{crestetto},\cite[Appendix A]{McGraw01011997}. However for $N>15$ the moment inversion algorithms become unstable, and as $N$ increases the problem quickly becomes ill-conditioned, since higher-degree moments of the distribution function are highly sensitive to the shape of $f$. 
To overcome these limitations, we introduce slicing method that first slices the distribution $f$ in the velocity domain and then applies the moment inversion method to each slice separately. The distribution function in this case is described by the following ansatz
\begin{align} \label{eq:slicing_ansatz}
g(\bm u, \bm x, t) = \sum_{s=1}^L \sum_{p=1}^{N_s} n_{s,p}(\bm x, t) \delta(\bm u - \bm v_{s,p}(\bm x, t)),
\end{align}
where $L$ denotes the number of slices, $N_s$ is the number of nodes per slice, and $n_{s,p}$ and $\bm v_{s,p}$ are the $p$-th weight and abscissa, respectively, of the $s$-th slice. The accuracy of the method is determined by the number of slices $L$ and the number of nodes $N_s$ per slice. The latter one is kept fixed and small to avoid the ill-conditioning of the moment inversion problem. Furthermore, for the case $N_s=2$ the moment inversion problem has analytical solutions \cite{DESJARDINS20082514} which is sufficient to approximate the moments of degree zero to three of each slice. In this way the slicing method avoids the ill-conditioning of DQMOM and has good approximation properties. A similar method was proposed for the population balance equation in \cite{ATTARAKIH2009742}. This slicing method can also be applied to the more general ansatz \eqref{eq:burby_ansatz_node_p}. 

The models based on the general ansatz \eqref{eq:burby_ansatz_node_p} however have a serious limitation. In highly kinetic simulations, the abscissa $\bm v_{s,p}$ ceases to be a one-to-one mapping due to the multi-streaming \cite{Uhlemann_2018}. Alternatively, reduced-order models can be constructed using phase space moments of the distribution function that uses a phase-space moment parametrization, see \eqref{eq:particle_ansatz}. This particular parametrization was studied in \cite{burby2025} for the Vlasov--Poisson system based on the Scovel and Weinstein \cite{scovel}. The phase-space moments are only functions of time and can capture multi-streaming without issues. Our \textit{final contribution} is that we propose a simple procedure to construct such phase-space moments based models for conservation laws. In addition, for the lowest-degree case we develop a transformation technique from the fluid moment representation to the phase space moment representation and verify the method numerically. 

The paper is organized as follows. Section \ref{sec:fluids} presents exact fluid models for a general class of conservation laws and the conservation properties of the models. Theorem \ref{prop:main_theorem} states the exactness for the reduced fluid models. In Section \ref{sec:particles}, we derive in a similar way phase-space-moment models and state their exactness in Theorem \ref{prop:main_theorem_particles}. In Section \ref{sec:application}, we apply these models to the non-relativistic and relativistic Vlasov--Maxwell systems. We show that the models preserve invariants including total mass, total momentum, and total energy. The reader interested in practical aspects relevant to the lowest-degree method can read directly from Section \ref{sec:slicing}, where we introduce the slicing method. In Section \ref{sec:fluid_pic_transf}, we describe our transformation from the fluid moment representation to the phase space moment representation. In Section \ref{sec:numerics}, we provide numerical results for the 1D-1U Vlasov--Amp\`ere system, demonstrating that weak and strong Landau damping and two-stream instability can be accurately simulated.

\section{Setup} \label{sec:preliminaries}
Let $g(\bm u, \bm x, t)$ be the unknown function of position, velocity, and time, defined on a domain $\Omega := \Omega_u \times \Omega_x \times (0,T)$. We consider a conservation law of the generic form
\begin{align} \label{eq:general_advection_non_conserv}
\partial_t g(\bm u, \bm x, t) +  \nabla \cdot (\bm G(\bm u, \bm x, \bm t) g(\bm u, \bm x, t)) +  \nabla_u \cdot \left( \bm F(\bm u, \bm x, t) g(\bm u, \bm x, t) \right) + R(\bm u, \bm x, t) g(\bm u, \bm x, t)= 0,
\end{align}
here, $\bm x \in \Omega_x \subset \mathbb{R}^{d}$ denotes the position, $\bm u \in \Omega_u \subset \mathbb{R}^{d}$ the momentum,  and $t \in (0, T) \subset \mathbb{R}^+$ the time. The operators  $\nabla$ and $\nabla_u$ represent the gradients with respect to $\bm x$ and $\bm u$, respectively.  

Let $C^k(\Omega)^d$ be the space of $k$ times continuously differentiable functions on $\Omega$ with an image in $\mathbb{R}^d$; and in particular let $C^{\infty}(\Omega)^d$ be the space of smooth functions with any number of derivatives. We assume that the advection coefficients $\bm G(\bm u, \bm x, t), \bm F(\bm u, \bm x ,t)$ and the source $R(\bm u, \bm x, t)$ are smooth functions, i.e., $\bm G(\bm u, \bm x, t), \bm F(\bm u, \bm x ,t) \in C^{\infty}(\Omega)^d$ and $R(\bm u, \bm x, t) \in C^{\infty}(\Omega)$.  These coefficients may depend on the unknown $g(\bm u, \bm x, t)$ itself. To relate the coefficients to $g(\bm u, \bm x, t)$,  the dependence is expressed by a system of differential-algebraic equations $\bm {\mathcal A}$:
\begin{align}
\bm {\mathcal A}[\bm G(\bm u, \bm x, t), \bm F(\bm u, \bm x, t), R(\bm u, \bm x, t), g(\bm u, \bm x, t)]=0.
\end{align}
For example, for the Vlasov equation the system $\bm {\mathcal A}$ is given by the Maxwell equations that evolves the electric and magnetic fields.

In what follows, we perform analysis for the case of one node only, i.e. $N=1$ in \eqref{eq:burby_ansatz_multinode}, and later show how it applies to multiple nodes. Following \cite{burby2023}, we define the degree $k$ fluid moment of $g(\bm u, \bm x, t)$ by
\begin{align} \label{eq:non-centered-moments}
\bm M^k (\bm x, t) = \int \bm u^{\otimes k} \, g(\bm x, \bm u, t) du,
\end{align} 
here $\otimes k$ denotes the $k$-fold outer product, i.e., $\bm u^{\otimes k} = \underbrace{\bm u \otimes ... \otimes \bm u}_{k-times}$.  
The centered moments of degree $k$ with respect to the  abscissa $\bm v(\bm x, t) \in C^{\infty}(\Omega_x \times [0,T])^d$ were defined earlier in \eqref{eq:centered-moments}.

To ease the notation, we always omit the arguments of functions. For the advection coefficients, we retain only the first argument and so write, for example, $\bm G(\bm u)$ instead of $\bm G(\bm u, \bm x, t)$. 
 We also introduce a special gradient $\nabla_v$ that is defined with the help of the Fr\'echet derivative $\frac{d}{d \varepsilon} \big|_{0} \bm G( \bm v + \varepsilon \delta \bm v) = \delta \bm v  \cdot \nabla_{\bm v} \bm G( \bm v)$.

Let $\bm e_i$ denote the $i$-th coordinate vector. We define gradients from the left and from the right as follows: $\nabla \bm v := \partial_i v_j \bm e_i \otimes \bm e_j$ and $(\bm v \nabla) = \partial_j v_i \bm e_i \otimes \bm e_j$. Similarly, for tensors we define $\nabla \bm S := \partial_i S_{mn} \bm e_i \otimes \bm e_m \otimes \bm e_n$ and $(\bm S \nabla) := \partial_i S_{mn} \bm e_m \otimes \bm e_n \otimes \bm e_i$. We also define the left and right divergence as $\nabla \cdot \bm S =  \partial_i S_{ij} \bm e_j$ and $(\bm S \cdot \nabla) =  \partial_j S_{ij} \bm e_i$. For order two tensors, we define cross products: $\bm S \times^1 \bm B=  \epsilon_{imn} S_{mj} B_n \bm e_i \otimes \bm e_j$ and $\bm S \times^2 \bm B=  \epsilon_{jmn} S_{im} B_n \bm e_i \otimes \bm e_j$, where $\epsilon_{ijk}$ is the permutation symbol.

With the ansatz \eqref{eq:burby_ansatz}, we work with distributions rather than functions.  The distributions $a(\bm u, \bm x, t)$ and $b(\bm u, \bm x, t)$ are considered equivalent if, for all $\phi(\bm x, \bm u, t) \in C^{\infty}_0(\Omega)$ (the space of smooth test functions with compact support), the following holds: 
$ \left<a, \phi \right> = \left< b, \phi \right>$ where $\left<a, \phi  \right>:=\int_{\Omega} a \, \phi \, dV $ is the $L_2$ inner product and $dV =dt \, dx \, du$ is the volume element.

\section{Fluid models for conservation laws} \label{sec:fluids}
To construct a reduced-order model of degree $k$, we take centered moments \eqref{eq:centered-moments} of the conservation law \eqref{eq:general_advection_non_conserv}. To this end, we multiply the advection equation \eqref{eq:general_advection_non_conserv} with $(\bm u - \bm v)^{\otimes k}$ and integrate over the momentum coordinate
\begin{align} 
\int_{\Omega_u} \partial_t g \, (\bm u-  \bm v)^{\otimes k} \, du &= - \int_{\Omega_u}  \nabla \cdot \left( \bm G(\bm u) \, g \right) \,  (\bm u- \bm v)^{\otimes k} \, du \\ \notag
&- \int_{\Omega_u} \nabla_u \cdot \left( \bm F(\bm u) g\right) \,  (\bm u- \bm v)^{\otimes k}   \, du - \int_{\Omega_u} \bm R(\bm u) \, g \,  (\bm u- \bm v)^{\otimes k}   \, du \, .
\end{align}
Next, we use the product rule $\partial_t g \, (\bm u- \bm v)^{\otimes k} = \partial_t \left[ g \, (\bm u- \bm v)^{\otimes k}  \right] - g \, \partial_t \left[(\bm u- \bm v)^{\otimes k} \right] $ to obtain evolution equations for the centered moments
\begin{align} \notag
\partial_t  C^0 &= - \int_{\Omega_u}  \nabla \cdot \left( \bm G(\bm u) \, g \right) \, du 
- \int_{\Omega_u} \nabla_u \cdot \left( \bm F(\bm u) g\right)    \, du - \int_{\Omega_u} \bm R(\bm u) \, g   \, du,  \\  \label{eq:centered_hierarchy}
\vdots \\ \notag
\partial_t \bm C^k &=\int_{\Omega_u} g \, \partial_t \left[(\bm u- \bm v)^{\otimes k} \right] \, du  - \int_{\Omega_u}  \nabla \cdot \left( \bm G(\bm u) \, g \right) \,  (\bm u- \bm v)^{\otimes k} \, du \\ \notag
&- \int_{\Omega_u} \nabla_u \cdot \left( \bm F(\bm u) g\right) \,  (\bm u- \bm v)^{\otimes k}   \, du - \int_{\Omega_u} \bm R(\bm u) \, g \,  (\bm u- \bm v)^{\otimes k}   \, du .
\end{align} 
The system \eqref{eq:centered_hierarchy} for the truncated $k$-number of moments is not closed without an assumption on $g$, since each equation cannot be represented entirely in terms of a finite number of moments.

To close the system, we assume now that $g$ is given by the ansatz \eqref{eq:burby_ansatz} of degree $k$. Then, we also need an equation for the abscissa $\bm v$. We return to this question in Theorem \ref{prop:main_theorem}.

The following definition introduces the notation of exactness for a moment system which will be the central property in the remainder of the paper:
\begin{definition} \label{def:exactness}
A moment system is called \emph{exact} if the distribution $g$ parametrized by the moments satisfies the distributional version of the conservation law.
\begin{align} 
\int \partial_t g \, \phi dV  +  \int \nabla \cdot (\bm G \, g) \, \phi \, dV +  \int \nabla_u \cdot \left( \bm F \, g \right) \, \phi \, dV + \int R\, g \, \phi \, dV = 0 , \hspace{0.5cm} \forall \phi \in C^{\infty}_0(\Omega),
\end{align}
where the coefficients are evaluated from $\mathcal{A}[\bm G,\bm F, R, g] = 0$.
\end{definition}

Before stating the main results, we examine the conditions for the exactness of the moment system.

\begin{definition} \label{def:e-tube}

	Let  $ \bm u \in C^k(\Omega_x \times [0,T],\R^{d})$.  Then, we define the \emph{$\varepsilon$-tube} around the codomain of $\bm u$ as
	$
	T_{\varepsilon}(\bm u) := \{ y \in  \R^{d} \mid \exists \bm x \in \Omega_x \text{ and } t \in [0,T] \text{ such that } \| \bm y-\bm u(\bm x, t)\| \leq \varepsilon\}.
	$ 
\end{definition}

\begin{lemma} \label{prop:lemma_delta_fluid}
Let $f(z,t)$ be a smooth real-valued function with compact support. Let $\phi \in C^{\infty}_0(\Omega_u)$ s.t. $\phi(\bm u)=1$ for all $\bm u \in T_{\varepsilon} (\bm v(\bm x, t))$ for some $\varepsilon>0$. Then, the following identity holds
\begin{align}
\int \partial_{u_i} \delta(\bm u-\bm v(\bm x, t)) f(\bm u, \bm x, t) \, \phi(\bm u) \, du = \int \partial_{u_i} \delta(\bm u- \bm v(\bm x, t)) f(\bm u, \bm x, t) \, du. 
\end{align}
\end{lemma}
\begin{proof}
Since $\phi(\bm u)=1$ for $\bm u$ in $T_{\varepsilon} (\bm v(\bm x, t))$, we have that $\partial_{u_i} \phi(\bm u) |_{\bm u=  \bm v} = 0$. Using integration by parts $\int \partial_{u_i} \delta(\bm u-  \bm v) f(\bm u, \bm x, t) \, \phi(\bm u) \, du = -\int \delta(\bm u-  \bm v) \partial_{u_i}  f(\bm u, \bm x, t)  \phi(\bm u) \, du$ $ - \underbrace{ \int \delta(\bm u-  \bm v)   f(\bm u, \bm x, t)  \partial_{u_i} \phi(\bm u) \, du}_{=0} = $ \\ $  -\partial_{u_i}  f(\bm u, \bm x, t) |_{u=  \bm v} = -\int \delta(\bm u-  \bm v) \partial_{u_i}  f(\bm u, \bm x, t) \, du = \int \partial_{u_i} \delta(\bm u -   \bm v) f(\bm u, \bm x, t) \, du$.
\end{proof}
In the case of multiple derivatives, we obtain similarly
\begin{align} \label{eq:fluid_multiple_derivatives}
\int \partial_{u_i}...\partial_{u_j} \delta(\bm u- \bm v) f(\bm u, \bm x, t) \, \phi(\bm u) \, du = \int \partial_{u_i}...\partial_{u_j} \delta(\bm u -   \bm v) f(\bm u, \bm x, t) \, du.
\end{align}

Let us first look at the necessary condition for the moment system to be exact.

\begin{proposition} \label{prop:fluid_necessary}
Consider a moment system ($C^0, \bm C^1,...,\bm C^k$) and its corresponding distribution \eqref{eq:burby_ansatz}. If the moment system is exact, then it satisfies the equations \eqref{eq:centered_hierarchy}.
\end{proposition}

\begin{proof}
For readability we only show the proposition for the case $\bm F=0$ and $R=0$, i.e.
\begin{align} \label{eq:prop_fluid_1}
\int \partial_t g \, \phi \, dV + \int \partial_{x_i} \left( G_i g \right) \, \phi \, dV = 0.
\end{align}
The additional case for non-zero coefficients can be treated in a similar way.
We consider a test function with compact support $\phi = (\bm u -   \bm v)^{\otimes k} \, \phi_x(\bm x) \, \phi_u(\bm u) \, \phi_t(t)$ where $\phi_x(\bm x) \in C^{\infty}_0(\Omega_x)$,$\phi_u(\bm u) \in C^{\infty}_0(\Omega_u)$, $\phi_t(t) \in C^{\infty}_0(\Omega_t)$. We further take $\phi_u$ that is $1$ in $T_{\varepsilon}(  \bm v)$ as in Lemma \ref{prop:lemma_delta_fluid} so that the gradients of $\phi_u$ are vanishing at $  \bm v$. Then equation \eqref{eq:prop_fluid_1} is written as
\begin{align} \label{eq:expanded_equation}
\int \partial_t & \left[ g \, (\bm u -   \bm v) ^ {\otimes k} \right]\, \phi_u \, \phi_x \, \phi_t \, dV \\ \notag 
& = \int g \, \partial_t  \left[ (\bm u -   \bm v) ^{\otimes k} \right]\, \phi_u \, \phi_x \, \phi_t \, dV - \int \partial_{x_i} \left( G_i g \right) \, (\bm u -   \bm v) ^{ \otimes k}  \phi_u \, \phi_x \, \phi_t \, dV .
\end{align}
Where we used product rule on $\partial_t$. 
For the term on the left-hand side we perform the integration by parts
\begin{align}
\int \partial_t & \left[ g \, (\bm u -   \bm v) ^{\otimes k} \right]\, \phi_u \, \phi_x \, \phi_t \, dV   = - \int g \, (\bm u -   \bm v) ^{\otimes k} \, \phi_u \, \phi_x \, \partial_t  \phi_t \, dV.
\end{align}
Since $g$ is given in terms of $\delta$ and its derivatives centered at $  \bm v$, we use Lemma \ref{prop:lemma_delta_fluid} to write
\begin{align} \notag
\int \partial_t & \left[ g \, (\bm u -   \bm v) ^{\otimes k} \right]\, \phi_u \, \phi_x \, \phi_t \, dV  = - \int g \, (\bm u -   \bm v) ^{ \otimes k}  \, \phi_x \, \partial_t  \phi_t \, dV 
= - \int \bm C^k \, \phi_x \, \partial_t \phi_t \, dx \, dt \\ \label{eq:lhs}
& \hspace{9cm} = \int \partial_t \bm C^k \, \phi_x \, \phi_t \, dx \, dt.
\end{align}
With the help of Lemma \ref{prop:lemma_delta_fluid}, the first term on the right-hand side of \eqref{eq:expanded_equation} can be written as
\begin{align} \label{eq:rhs_1}
\int g \, \partial_t  \left[ (\bm u -   \bm v) ^{\otimes k} \right]\, \phi_u \, \phi_x \, \phi_t \, dV = \int g \, \partial_t  \left[ (\bm u -   \bm v) ^ {\otimes k} \right]\, \phi_x \, \phi_t \, dx \, dt, \\ \notag
\end{align}
and the second term we rewrite using integration by parts
\begin{align} \notag
-\int \partial_{x_i} & \left( G_i g \right) \, (\bm u -   \bm v) ^{\otimes k}  \phi_u \, \phi_x \, \phi_t \, dV \\ \notag
 = &  \int G_i g \, \partial_{x_i} \left[ (\bm u -   \bm v) ^ {\otimes k} \right]  \phi_u \, \phi_x \, \phi_t \, dV 
+  \int  G_i g \, (\bm u -   \bm v) ^ {\otimes k}\phi_u \, \partial_{x_i} \phi_x \, \phi_t \, dV ,
\end{align}
then applying the Lemma to eliminate $\phi_u$ and performing integration by parts on $\partial_{x_i}$ we get
\begin{align} \label{eq:rhs_2}
-\int \partial_{x_i} & \left( G_i g \right) \, (\bm u -   \bm v) ^{\otimes k}  \phi_u \, \phi_x \, \phi_t \, dV  = -\int \partial_{x_i} \left( G_i g \right) \, (\bm u -   \bm v) ^ {\otimes k}  \phi_x \, \phi_t \, dx \, dt, \\ \notag
\end{align}
with this, we see that the terms in \eqref{eq:lhs}, \eqref{eq:rhs_1}, \eqref{eq:rhs_2} under $dx \, dt$ integrals are the same as in \eqref{eq:centered_hierarchy}.
\end{proof}

Proposition \ref{prop:fluid_necessary} states a necessary condition for the moment system to be exact. To obtain a sufficient conditions, we return to the question on the condition on the center $\bm v$ of the system.
To close the system, an evolution equation for the center of the moments is needed. Since the center of the moments is a redundant variable, there is some freedom in its choice. It turns out that there is an expression for $\bm v$ that keeps the moment system \eqref{eq:centered_hierarchy} exact.

The following result holds for a finite number of moments in arbitrary dimensions.

\begin{theorem} \label{prop:main_theorem}
Let the moment hierarchy \eqref{eq:centered_hierarchy} of degree $k$ be closed with ansatz $g$ of degree $k$ given in \eqref{eq:burby_ansatz}. Let the center satisfy 
\begin{align} \label{eq:velocity}
\partial_t \bm v =& - \bm G(  \bm v) \cdot (\nabla \bm v) +  \bm F(  \bm v).
\end{align}
Then the reduced model coupled to $\bm {\mathcal A}[\bm G, \bm F, R, g]=0$, is an exact solution of the conservation law \eqref{eq:general_advection_non_conserv}.
\end{theorem}
The first term in equation \eqref{eq:velocity} indicates that the center is advected by the coefficient $\bm G$, while the second term indicates that the coefficient $\bm F$ acts as a forcing term.

\begin{remark} \label{remark:main_remark}
Note that $\bm {\mathcal A}$ does not appear explicitly in the proof of Theorem \ref{prop:main_theorem}. The proof applies as long as $\bm {\mathcal A}[\bm G, \bm F, R, g_N]$ is well-defined for the distribution $g_N$ and the coefficients $\bm G$, $\bm F$, $R$ can be determined.   
\end{remark}

The proof is given in Appendix \ref{sec:proof_main_theorem_1u1x} for 1u1x case and in Appendix \ref{sec:proof_main_theorem} for the general case. We close this section by providing the explicit formulas for the models of degree zero, one and two.

For convenience, the centered moments of degree zero, one, and two can be denoted by
\begin{align} \label{eq:moments}
n(\bm x, t) &= \int g(\bm u, \bm x, t) du , \\
\bm P(\bm x, t) &= \int (\bm u -  \bm v(\bm x, t)) g(\bm u, \bm x, t) du , \\
\bm S(\bm x, t) &= \int (\bm u -  \bm v(\bm x, t)) \otimes  (\bm u - \bm v(\bm x, t)) g(\bm u, \bm x, t) du.
\end{align}

The degree zero model reads

\begin{align} \label{eq:deg_zero_v}
\partial_t \bm v &= - \bm G(  \bm v) \cdot (\nabla \bm v) +  \bm F(  \bm v), \\ \label{eq:deg_zero_n}
\partial_t n &= - \nabla \cdot \left( \bm G(  \bm v) \, n \right)  - R(  \bm v) \, n. 
\end{align}

The degree one model reads

\begin{align}
\partial_t \bm v =& - \bm G(  \bm v) \cdot (\nabla \bm v) +  \bm F(  \bm v), \\
\partial_t n =& - \nabla \cdot \left( \bm G(  \bm v) \, n + \bm P \cdot \nabla_v \bm G(  \bm v) \right)  - R(  \bm v) \, n - \bm P \cdot \nabla_v R(  \bm v), \\
\partial_t \bm P =& - (\bm P \otimes \bm G(  \bm v)) \cdot \nabla - (\bm P \cdot  \nabla_v \bm G(  \bm v))\cdot (\nabla \bm v) +\bm P \cdot \nabla_v \bm F(  \bm v) - \bm P \, R(  \bm v). 
\end{align}

The degree two model reads

\begin{align}
\partial_t \bm v =& - \bm G(  \bm v) \cdot (\nabla \bm v) + \bm F(  \bm v), \\ 
\partial_t n =& - \nabla \cdot \left( \bm G(  \bm v) \, n + \bm P \cdot \nabla_v \bm G(  \bm v) + \frac{1}{2} \bm S : \nabla_v \nabla_v \bm G(  \bm v) \right) \\ \notag
 & - R(  \bm v) \, n - \bm P \cdot \nabla_v R(  \bm v) - \frac{1}{2} \bm S : \nabla_v \nabla_v R(  \bm v), \\ 
\partial_t \bm P =& - (\bm P \otimes \bm G(  \bm v)) \cdot \nabla -   (\bm P \cdot \nabla_v \bm G(  \bm v))\cdot (\nabla \bm v) \\ \notag & - (\bm S \cdot \nabla_v \bm G(  \bm v)) \cdot \nabla -   \frac{1}{2} (\bm S: \nabla_v \nabla_v \bm G(  \bm v)) \cdot (\nabla \bm v) \\ \notag
& +\bm P \cdot \nabla_v \bm F(  \bm v) +  \frac{1}{2} \bm S : \nabla_v \nabla_v \bm F(  \bm v)  - \bm P \, R(  \bm v) - \bm S \cdot \nabla_v R(  \bm v), \\ 
\partial_t \bm S =& - \left( \bm S \otimes \bm G(  \bm v) \right) \cdot \nabla -   \bm S \cdot (\nabla_v \bm G(  \bm v)) \cdot (\nabla \bm v) -   (\bm v \nabla)\cdot (\bm G(  \bm v) \nabla_v) \cdot \bm S \\ \notag
 &+ (\bm F(  \bm v) \nabla_v) \cdot \bm S + \bm S \cdot (\nabla_v \bm F(  \bm v))   -\bm S R(  \bm v).
\end{align}

\noindent \textbf{Multi-node system} \\
To extend these models to the case of multiple nodes, we assume that the equations 

To extend these models to the case of multiple nodes, we assume that the equations for the moments and the abscissa for each node $p$ are the same as for the case of single node discussed above as outlined in \eqref{eq:burby_ansatz_node_p} and \eqref{eq:burby_ansatz_multinode}. For example, the degree zero model (\ref{eq:deg_zero_v}--\ref{eq:deg_zero_n}) for each node $p$ read
\begin{align} \label{eq:deg_zero_v_p}
\partial_t \bm v_p &= - \bm G(  \bm v_p) \cdot (\nabla \bm v_p) +  \bm F(  \bm v_p), \\ \label{eq:deg_zero_n_p}
\partial_t n_p &= - \nabla \cdot \left( \bm G(  \bm v_p) \, n_p \right)  - R(  \bm v_p) \, n_p .
\end{align}
and similarly for the higher degree models.
We assume that the coefficients  $\bm G$, $\bm F$, $R$ are computed from $\bm{\mathcal{A}}[\bm G, \bm F, R, \sum_{p=1}^N g_{p}]$ where $g_{p}$ was defined in \eqref{eq:burby_ansatz_node_p}.
By Theorem \ref{prop:main_theorem} and Remark \ref{remark:main_remark}, this implies that for each node $p$, the distribution $g_{p}$ satisfies the conservation law
\begin{align}
\int \partial_t g_{p} \, \phi dV  +  \int \nabla \cdot (\bm G \, g_{p}) \, \phi \, dV +  \int \nabla_u \cdot \left( \bm F \, g_{p} \right) \, \phi \, dV + \int R\, g_{p} \, \phi \, dV = 0, \hspace{0.5cm} \forall \phi \in C^{\infty}_0(\Omega).
\end{align}
Thus, summing over $p$ we see that $g = \sum_{p=1}^N g_{p}$ also satisfies the conservation law.

\section{Particle models for conservation laws} \label{sec:particles}

In order to derive a particle method, we start with a function $Z_p \in C^1([0,T]; \R^{d_{\bm u} + d_{\bm x}})$ and define the phase space moment of degree $s$ as
\begin{align}
\bm W^s(t) = \int_{\Omega} (\bm z-\bm Z_p)^{\otimes s} \, g \, dz .
\end{align}
We assume the distribution function $g=\sum_p^N g_{p}$, where $p$ is the particle index. Assume a special ansatz for each $g_{p}$ that we take from \cite{burby2025},
\begin{align} \label{eq:particle_ansatz}
g_{p}(\bm u, \bm x, t) &= W^0_p \delta(\bm z- \bm Z_p) - \bm W^1_{p,i_1} \partial_{z_{i_1}} \delta(\bm z-\bm Z_p) \\ \notag 
&+ \frac{1}{2}\bm W^2_{p,i_1, i_2} \partial_{z_{i_1}} \partial_{z_{i_2}} \delta(\bm z-\bm Z_p) + ...+\frac{(-1)^{k}}{k!} \bm W^{k}_{p,i_1...i_{k}}\partial_{z_{i_1}}... \partial_{z_{i_{k}}} \delta(\bm z-\bm Z_p),
\end{align}
where $\bm z:=(\bm u, \bm x)^\top$ is the phase space coordinate, and $k$ is denotes the degree of the ansatz. Here, our goal is to use this ansatz to construct exact particle models for conservation laws.
This ansatz function of degree $k$ reproduces the phase space moments up to degree $k$ exactly. The index $k$ can also be different for each particle $p$. The idea is to derive evolution equations for the variables $\{\bm Z_p, W_p^0, \bm W_p^1, ..., \bm W_p^k\}$ such that the ansatz  $g=\sum_p^N g_{p}$ is an exact solution of the conservation law. 

Let us rewrite the conservation law \eqref{eq:general_advection_non_conserv} in the following form
\begin{align} \label{eq:conservation_law_phase_space}
\partial_t g(\bm z, t) + \nabla_z \cdot (\bm A(\bm z, t) g(\bm z, t)) = 0,
\end{align}
where we introduced the phase space advection coefficient $\bm A(\bm z, t)= \left[\bm F(\bm u, \bm x, t), \bm G(\bm u, \bm x, t) \right]^\top$. Let us consider the case of a single particle, i.e., $N=1$.
To construct reduced models we proceed by taking the phase space moments of the conservation law with respect to the center $\bm Z_p$ 
\begin{align} \notag
& \int \partial_t g(\bm z, t) \, dz + \int \nabla_z \cdot (\bm A(\bm z, t) g(\bm z, t)) \, dz = 0, \\ \notag
& \int \partial_t g(\bm z, t) \, (\bm z-\bm Z_p(t)) \, dz + \int \nabla_z \cdot (\bm A(\bm z, t) g(\bm z, t)) \, (\bm z-\bm Z_p(t)) \, dz = 0, \\ \notag
& \int \partial_t g(\bm z, t) \, (\bm z-\bm Z_p(t)) \otimes (\bm z-\bm Z_p(t)) \, dz + \int \nabla_z \cdot (\bm A(\bm z, t) g(\bm z, t)) \, (\bm z-\bm Z_p(t)) \otimes (\bm z-\bm Z_p(t)) \, dz = 0, \\ 
\vdots
\end{align}
We then use the product rule with respect to the time variable to obtain
\begin{align} \notag
& \partial_t W^0_p =  0, \\ \notag
& \partial_t \bm W^1_p = \int g(\bm z, t)  \partial_t (\bm z-\bm Z_p(t)) \, dz -\int \nabla_z \cdot (\bm A(\bm z, t) g(\bm z, t)) \, (\bm z-\bm Z_p(t)) \, dz ,\\ \notag
& \partial_t \bm W^2_p = \int g(\bm z, t) \,  \partial_t ((\bm z-\bm Z_p(t)) \otimes (\bm z-\bm Z_p(t))) \, dz  \\ \notag
& \hspace{3.5cm}-\int \nabla_z \cdot (\bm A(\bm z, t) g(\bm z, t)) \, (\bm z-\bm Z_p(t)) \otimes (\bm z-\bm Z_p(t)) \, dz ,\\ \label{eq:particle_necessary_equations_1}
\vdots 
\end{align}
The moment system \eqref{eq:particle_necessary_equations_1} depends on the function $g$ and cannot be solved in the current form. To eliminate this dependence, we assume that $g$ is defined by the ansatz \eqref{eq:particle_ansatz}. In addition we need an expression for the center $\bm Z_p$ that we state in Theorem \ref{prop:main_theorem_particles}.

\begin{theorem} \label{prop:main_theorem_particles}
Let the moment hierarchy of degree $k$ be closed with ansatz $g$ of degree $k$ given in \eqref{eq:particle_ansatz}. Let the center satisfy 
\begin{align} \label{eq:velocity_particles}
\partial_t \bm Z_p =& \bm A(\bm Z_p, t),
\end{align}
then the reduced model coupled to $\bm {\mathcal A}[\bm A, g]=0$ is an exact solution of the conservation law \eqref{eq:conservation_law_phase_space}.
\end{theorem}

The proof is similar to the proof of Theorem \ref{prop:main_theorem} and is given in Appendix \ref{sec:proof_main_theorem_particles}.

\begin{remark} \label{remark:second_remark}
Again $\bm {\mathcal A}$ does not appear explicitly in the proof. The proof applies as long as $\bm {\mathcal A}[\bm A, g]=0$ allows to determine the coefficients $\bm A$ from the distribution $g$.   
\end{remark}

\textbf{Multi-node case (many particles).}
The extension of the theorem to a finite number of particles and species is straightforward. For this one can assume that each particle represents a portion of the probability density function, along with the corresponding moments of that portion with respect to the center of the given particle $\bm Z_p$. The total probability density function is then given by the sum $g = \sum_p^N g_{p}$ where $g_p$ are of degree $k$. Since the moments of the whole distribution can be split additively into the moments of sub-distributions and the moments of the sub-distributions are reproduced exactly, it follows that the ansatz $g$ approximates moments up to degree $k$ exactly. The coefficients $\bm A(\bm z,t)$,  are computed from the algebraic differential system $\mathcal A[\bm A(\bm z,t), \sum_p g_{p}]$. In view of Remark \ref{remark:second_remark}, the exactness of the moment system for each particle is not affected, i.e., Eq. \eqref{eq:conservation_law_phase_space} still holds for each $g_{p}$. Since $\mathcal A[\bm A(\bm z,t), \sum_p g_{p}]$ is used for all particles, the advection coefficient $\bm A(\bm z, t)$ are the same. Therefore, Eq. \eqref{eq:conservation_law_phase_space} written out for individual particles sums up and holds for $g=\sum_p^N g_{p}$. This shows that the model with multiple particles remains exact. 

Let us give explicit equations for the degree zero, degree one, degree two models. The degree zero model is the classical particle-in-cell model with constant weights:
\begin{align} \label{eq:general_degree_zero_particle_model}
\partial_t \bm Z_p &= \bm A(\bm Z_p, t), \\ \notag
\partial_t W^0 &= 0 .
\end{align}
The degree one model reads
\begin{align} \notag
\partial_t \bm Z_p &= \bm A(\bm Z_p, t), \\ \label{eq:general_degree_one_particle_model}
\partial_t W^0 &= 0, \\ \notag
\partial_t \bm W^1 &= \bm W^1 \cdot \nabla_{z} \bm A(\bm Z_p, t).
\end{align}
The degree two model reads
\begin{align} \notag
\partial_t \bm Z_p &= \bm A(\bm Z_p, t) ,\\ \notag
\partial_t W^0 &= 0 ,\\ \label{eq:general_degree_two_particle_model}
\partial_t \bm W^1 &= \bm W^1 \cdot \nabla_{z} \bm A(\bm Z_p, t) + \frac{1}{2!} \bm W^2 : \nabla_z \nabla_z \bm A(\bm Z_p, t) ,\\ \notag
\partial_t \bm W^2 &=  \bm W^2 \cdot \nabla_z \bm A(\bm Z_p,t) + \left( \bm W^2 \cdot \nabla_z \bm A(\bm Z_p, t) \right)^\top .
\end{align}

\subsection{Hybrid-model}
To construct hybrid model, we note that Theorem \ref{prop:main_theorem} and Theorem \ref{prop:main_theorem_particles} hold without explicit assumptions about $\mathcal A$. We can construct the hybrid model as follows (recall we assumed $R=0$ in the particles case): 
\begin{proposition} \label{prop:superposition}
Let $g_f$ and $g_p$ be the fluid \eqref{eq:burby_ansatz} and particle ansatz \eqref{eq:particle_ansatz}, respectively. Let the moment variables and the particle variables be evolved according to moment systems in Theorem \ref{prop:main_theorem} and \ref{prop:main_theorem_particles}. Furthermore, assume the coefficients are solved from $\bm{\mathcal A}[\bm G, \bm F, R, g]=0$ with $g=g_f + g_p$. Then, the hybrid ansatz  $g=g_f+g_p$ satisfies the conservation law \eqref{eq:general_advection_non_conserv}  with $R=0$. 
\end{proposition}
\begin{proof}
According to Theorem \ref{prop:main_theorem} and Remark \ref{remark:main_remark}, if $g_f$ is parametrized in terms of moments then it holds: $\partial_t g_f + \nabla \cdot (\bm G \, g_f) + \nabla_u \cdot(\bm F \, g_f) = 0$. Similarly according to Proposition \ref{prop:main_theorem_particles} and Remark \ref{remark:second_remark} it holds $\partial_t g_p + \nabla \cdot (\bm G \, g_p) + \nabla_u \cdot(\bm F \, g_p)  = 0$ 
summing up the two $\partial_t g + \nabla \cdot (\bm G \, g) + \nabla_u \cdot(\bm F \, g) = 0$.
\end{proof}

It is simple to see that the hybrid model remains exact for the multi-node moment and particle equations.

\section{Application to Plasma} \label{sec:application}
In this section, we obtain hybrid models for the Vlasov--Maxwell and the relativistic Vlasov--Maxwell systems. For the sake of convenience, we will change slightly the variables: instead of considering the center $\bm v$ we will consider the center given by $m \bm v$, where $m$ is the species' mass.

\subsection{Vlasov--Maxwell}
The Vlasov--Maxwell system for a single plasma species in self-consistent electro-magnetic field is given is Gauss units by:
\begin{align} \label{eq:maxwell_f}
& \partial_t f +  \nabla \cdot \left(\frac{\bm u}{m}  f \right) +  \nabla_u \cdot \left( e \left(\bm E + \frac{\bm u}{cm} \times \bm B \right) f \right) = 0 ,
\end{align}
where the fields $\bm E(\bm x, t), \bm B(\bm x, t)$ are calculated using Maxwell equations
\begin{align} \label{eq:maxwell_e}
& -\frac{1}{c} \partial_t \bm E + \nabla \times \bm B = \frac{4 \pi}{c} \bm J, \\ \label{eq:maxwell_b}
& \frac{1}{c} \partial_t \bm B + \nabla \times \bm E = 0 ,\\ \label{eq:maxwell_rho}
& \nabla \cdot \bm E = 4 \pi e \tilde n, \\ \label{eq:maxwell_divb}
& \nabla \cdot \bm B = 0,
\end{align}
with 
\begin{align}
\tilde n = \int f du - n_0,  \hspace{1cm} \bm J = \frac{e}{m} \int f \, \bm u \, du.
\end{align}
here $\bm x \in \Omega_x \subset \mathbb{R}^3$ is position, $\bm u \in \Omega_u \subset \mathbb{R}^3$ - momentum,  $f(\bm u, \bm x, t)$ - probability density function, $\bm E(\bm x, t)$ - electric field, $\bm B(\bm x, t)$ - magnetic field, $n(\bm x, t)$ - charge density, $n_0$ - constant neutralizing background density, $m$ - mass, and $e$ - charge.
 
We consider the case of multiple nodes $N$. To obtain the fluid model we set $\bm G(\bm u, \bm x, t)=\frac{\bm u}{m}$, $\bm F(\bm u, \bm x, t)= e \left( \bm E(\bm x, t) + \right. $ $ \left. \frac{\bm u \times \bm B(\bm x, t)}{cm} \right)$, $R(\bm u, \bm x, t) = 0$. 
For example, the degree zero fluid system reads
\begin{align} \label{eq:deg_zero_moment_system_v}
\partial_t \bm v_p &= -\bm v_p \cdot \nabla \bm v_p + \frac{e}{m} \left( \bm E(x,t) + \frac{\bm v_p}{c} \times \bm B \right) ,\\ \label{eq:deg_zero_moment_system_n}
\partial_t n_p &= - \nabla \cdot (n_p \bm v_p)
\end{align} 

And the degree two fluid system is given by
\begin{align} \label{eq:ten_moment_system_v}
\partial_t \bm v_p &= -\bm v_p \cdot \nabla \bm v_p + \frac{e}{m} \left( \bm E(x,t) + \frac{\bm v_p}{c} \times \bm B \right) ,\\ \label{eq:ten_moment_system_n}
\partial_t n_p &= -m^{-1} \nabla \cdot (\bm P_p+ m n_p \bm v_p), \\ \label{eq:ten_moment_system_p}
\partial_t \bm P_p &= - \bm P_p \cdot \nabla \bm v_p -  \bm v_p \cdot \nabla \bm P_p -  (\nabla \cdot \bm v_p) \bm P_p +  \frac{e}{cm} \bm P_p \times \bm B - m^{-1} \nabla \cdot \bm S_p , \\ \label{eq:ten_moment_system_s}
\partial_t \bm S_p &= -\bm v_p \cdot \nabla \bm S_p - (\bm v_p \nabla) \cdot \bm S_p - \bm S_p \cdot \nabla \bm v_p -(\nabla \cdot \bm v_p) \bm S_p  + \frac{e}{cm} \left( \bm S_p \times^1 \bm B \right) + \frac{e}{cm} \left( \bm S_p \times^2 \bm B \right) .
\end{align}

The electric and magnetic fields $\bm E, \bm B$ are calculated from the Maxwell equations. The fluid models of degree $k=0$ give rise to the following current on the right-hand side of the Amp\`ere's equation \eqref{eq:maxwell_e}
\begin{align} \label{eq:current_deg_zero}
\bm J=  \frac{e}{m} \int \bm u \left(\sum_p^N g_p \right) du  = \frac{e}{m} \sum_p^N \left( m \, n_p \, \bm v_p \right).
\end{align}
and for degree $k>0$
\begin{align}\label{eq:current_deg_greater_zero}
\bm J = \frac{e}{m} \sum_p^N \left(\bm P_p + m \, n_p \, \bm v_p \right).
\end{align}

The expression for the particle model is a bit more involved. We recall from Section \ref{sec:particles} that $\bm A = (\bm F, \bm G)$, so that the degree zero model reads
\begin{align} \label{eq:ten_moment_system_xk}
\partial_t \bm X_p &= \frac{\bm U_p}{m} ,\\ \label{eq:ten_moment_system_uk}
\partial_t \bm U_p &= e \left( \bm E(\bm X_p, t) + \frac{\bm U_p \times \bm B(\bm X_p, t)}{cm} \right),
\end{align}
that matches the classical particle-in-cell discretization.

To obtain the degree one particle model we denote $\bm A = (\bm A_u, \bm A_x)=(\bm F, \bm G) $ and its derivatives
\begin{align}
\nabla_z \bm A = \begin{bmatrix}
\nabla_u \bm A_u & \nabla_u \bm A_x \\
\nabla_x \bm A_u & \bm 0
\end{bmatrix},
\end{align}
where the gradients explicitly are given by
\begin{align}
& \nabla_u \bm A_u = \nabla_{u_i} \bm A_{u_j} \, \bm e_i \otimes \bm e_j = \frac{e}{cm} \varepsilon_{jmn} \delta_{im} B_n \, \bm e_i \otimes \bm e_j, \\
& \nabla_u \bm A_x = \nabla_{u_i} \bm A_{x_j} \, \bm e_i \otimes \bm e_j =\frac{1}{m} \delta_{ij} \, \bm e_i \otimes \bm e_j ,\\
& \nabla_x \bm A_u = \nabla_{x_i} \bm A_{u_j} \, \bm e_i \otimes \bm e_j= e \left( \partial_i E_j + \varepsilon_{jmn}\frac{u_m \partial_i B_n}{cm} \right) \, \bm e_i \otimes \bm e_j ,
\end{align}
and where all quantities are evaluated at the particle positions $\bm Z_p = (\bm U_p, \bm X_p)$. We also split the degree one moment
\begin{align}
\bm W^1_p = \begin{bmatrix}
\bm W^1_{p,u} \\
\bm W^1_{p,x}
\end{bmatrix},
\end{align}
where $\bm W^1_p \in \mathbb{R}^{d_u+d_x}$,$\bm W^1_{p,u} \in \mathbb{R}^{d_u}$, $\bm W^1_{p,x} \in \mathbb{R}^{d_x}$. Then the evolution equation for the moment of degree one reads
\begin{align} \notag
\partial_t \bm W^1_p &= \begin{bmatrix}
\bm W^1_{p,u} \cdot \nabla_u \bm A_u + \bm W^1_{p,x} \cdot \nabla_x \bm A_u \\
\bm W^1_{p,u} \cdot \nabla_u \bm A_x + \bm W^1_{p,x} \cdot \nabla_x \bm A_x
\end{bmatrix},
\end{align}
and the particle model of degree one then reads
\begin{align}
\partial_t \bm X_p &= \frac{\bm U_p}{m} ,\\ 
\partial_t \bm U_p &= e \left( \bm E(\bm X_p, t) + \frac{\bm U_p \times \bm B(\bm X_p, t)}{cm} \right) ,\\ \
\partial_t \bm W^1_p & =\begin{bmatrix}
\frac{e}{cm} \bm W^1_{p,u} \times \bm B(\bm X_p, t) + e \left( \bm W^1_{p,x} \cdot \nabla \bm E(\bm X_p, t) - \frac{( \bm W^1_{p,x} \cdot \nabla \bm B(\bm X_p, t) ) \times \bm U_p }{cm}\right) \\
\frac{1}{m} \bm W^1_{p,u}
\end{bmatrix},
\end{align}

The expression for the degree two system \eqref{eq:general_degree_two_particle_model} is convoluted, so we write only its constitutive terms. The terms of interest are: $\bm W^2_p \cdot \nabla_z \bm A$ and $\bm W^2_p : \nabla_z \nabla_z \bm A$. First, we split $\bm W^2_p$
\begin{align}
\bm W^2_p = \begin{bmatrix}
\bm W^2_{p,uu} & \bm W^2_{p,ux} \\
\bm W^2_{p,xu} & \bm W^2_{p,xx}
\end{bmatrix},
\end{align}
then the first term reads
\begin{align}
\bm W^2_p \cdot \nabla_z \bm A &=
\begin{bmatrix}
\bm W^2_{p,uu} & \bm W^2_{p,ux} \\
\bm W^2_{p,xu} & \bm W^2_{p,xx}
\end{bmatrix} \begin{bmatrix}
\nabla_u \bm A_u & \nabla_u \bm A_x \\
\nabla_x \bm A_u & \nabla_x \bm A_x
\end{bmatrix} ,
\end{align}
and the second term reads
\begin{align} \notag
\bm W^2_p : \nabla_z \nabla_z \bm A & = \begin{bmatrix}
\bm W^2_p : \nabla_z \nabla_z \bm A_u \\
\bm W^2_p : \nabla_z \nabla_z \bm A_x 
\end{bmatrix} \\ \notag 
&= 
\begin{bmatrix}
\bm W^2_{p,uu} : \nabla_u \nabla_u \bm A_u + \bm W^2_{p,ux} : \nabla_u \nabla_x \bm A_u + \bm W^2_{p,xu} : \nabla_x \nabla_u \bm A_u + \bm W^2_{p,xx} : \nabla_x \nabla_x \bm A_u \\
\bm W^2_{p,uu} : \nabla_u \nabla_u \bm A_x + \bm W^2_{p,ux} : \nabla_u \nabla_x \bm A_x + \bm W^2_{p,xu} : \nabla_x \nabla_u \bm A_x + \bm W^2_{p,xx} : \nabla_x \nabla_x \bm A_x 
\end{bmatrix} \\
&=\begin{bmatrix}
 \bm W^2_{p,ux} : \nabla_u \nabla_x \bm A_u + \bm W^2_{p,xu} : \nabla_x \nabla_u \bm A_u + \bm W^2_{p,xx} : \nabla_x \nabla_x \bm A_u \\
\bm 0
\end{bmatrix},
\end{align}
where the non-zero gradient components are defined by
\begin{align}
& \nabla_{x_k} \nabla_{u_i} A_{u_j} \, \bm e_k \otimes \bm e_i \otimes \bm e_j = \frac{e}{cm} \varepsilon_{jmn} \delta_{im} \partial_{k} B_n \bm e_k \otimes \bm e_i \otimes \bm e_j, \\
& \nabla_{u_k} \nabla_{x_i} A_{u_j} \, \bm e_k \otimes \bm e_i \otimes \bm e_j = \frac{e}{cm} \left( \varepsilon_{jmn} \delta_{km} \partial_i B_n \right) \,  \bm e_k \otimes \bm e_i \otimes \bm e_j,  \\
& \nabla_{x_k} \nabla_{x_i} A_{u_j} \, \bm e_k \otimes \bm e_i \otimes \bm e_j = e \left( \partial_{k} \partial_{i} E_j + \varepsilon_{jmn} \frac{u_m \partial_{k} \partial_{i} B_n}{cm} \right) \bm e_k \otimes \bm e_i \otimes \bm e_j ,
\end{align}
and all quantities are evaluated at the particle positions $\bm Z_p = (\bm U_p, \bm X_p)$.

The current due to particles depends on the degree of the model. A naive computation of the current of the degree one gives
\begin{align} \label{eq:wrong_current}
\bm J = \frac{e}{m} \int \bm u \left(\sum_p^N g_p \right) du  =\frac{e}{m} \sum_p^N \left[ \bm U_p \, W^{0}_p \, \delta(\bm x-\bm X_p) + \bm W_{p,u}^1 \, \delta(\bm x-\bm X_p) - \bm W_{p,x}^1 \cdot \nabla_x \delta(\bm x-\bm X_p) \, \bm U_p \right],
\end{align}
The current \eqref{eq:wrong_current} includes Dirac delta functions and as such the Amp\`ere's equation \eqref{eq:maxwell_e} is not well-defined. To overcome this we propose to regularize the current by convolution. Let $J_{i,p}$ denote the $i$-th component of the current $\bm J$ due to a particle with index $p$.
\begin{align} 
J_i(\bm x) &=  \sum_p^{N_p} \int_{\Omega}  J_{i,p}(\bm x - \bm \tau) K(\bm \tau) d \bm \tau,
\end{align}
where $K(\bm \tau)$ is a smoothing kernel such as the Gaussian kernel. Or equivalently replacing $\delta(\bm x - \bm X_p)$ in the current \eqref{eq:wrong_current} with $K(\bm x - \bm X_p)$ and changing the sign where integration by parts is used.
\begin{align}  \label{eq:regul_current}
\bm J = \frac{e}{m} \sum_p^N \left[ \bm U_p \, W^{0}_p \, K(\bm x-\bm X_p) + \bm W_{p,u}^1 \, K(\bm x-\bm X_p) + \bm W_{p,x}^1 \cdot \nabla_x K(\bm x-\bm X_p) \, \bm U_p \right].
\end{align}
Due to the remark \ref{remark:second_remark}, the model with the modified current is exact solution of the modified Vlasov--Maxwell system. 
 
Since the models are exact solutions in the sense of distributions, it can be shown that they conserve mass, energy, and momentum in the sense of distributions. We study these properties in Section \ref{sec:relativistic} for the relativistic models.

\subsection{Relativistic Vlasov--Maxwell system} \label{sec:relativistic}
The relativistic Vlasov--Maxwell equation reads:
\begin{align} \label{eq:rel_maxwell_f}
& \partial_t f + \nabla \cdot \left( \frac{\bm u}{m \, \gamma(\bm u)}  f \right) + \nabla_u \cdot \left( e \left( \bm E + \frac{\bm u}{c \,m \, \gamma(\bm u) } \times \bm B \right)  f \right) = 0,
\end{align}
and is coupled to Maxwell equations (\ref{eq:maxwell_e}-\ref{eq:maxwell_divb})
with 
\begin{align}
\gamma(\bm u) = \sqrt{1+\frac{|\bm u|^2}{m^2 c^2} }, \hspace{1cm} \tilde n = \int f \, du - n_0, \hspace{1cm} \bm J = \frac{e}{m} \int \frac{\bm u}{\gamma(\bm u)} f \, du.
\end{align}
To construct hybrid model we set $\bm G(\bm u, \bm x, t)=\frac{\bm u}{m \gamma(\bm u)}$, $\bm F(\bm u, \bm x, t) = e \left( \bm E(\bm x, t) + \frac{\bm u \times \bm B(\bm x, t)}{cm \gamma(\bm u)} \right)$, \\ $R(\bm u, \bm x, t)=0$. 
To obtain fluid equations, we first calculate derivatives of the coefficients
\begin{align} 
&\nabla_v \bm G(m\bm v) = \frac{\bm I}{\gamma(m\bm v)} - \frac{\bm v \otimes \bm v}{\ c^2 \ \gamma(m\bm v)^3}, \\
&\nabla_v \nabla_v \bm G(m\bm v) = - \frac{\bm v \otimes \bm I}{ c^2 \, \gamma(m\bm v)^3} - \left( \frac{\delta_{ik} v_j + \delta_{ij} v_k}{c^2 \, \gamma(m\bm v)^3} \bm e_i \otimes \bm e_j \otimes \bm e_k \right) + \frac{3 \bm v \otimes \bm v \otimes \bm v}{c^4 \, \gamma(m\bm v)^5} ,\\
& \nabla_v \cdot \bm F(m\bm v) = 0, \\
& \nabla_v \bm F(m\bm v) = \frac{e}{c} \left[ \frac{\bm I}{\gamma(m\bm v)} - \frac{\bm v \otimes \bm v}{c^2 \ \gamma(m\bm v)^3} \right] \times^2 \bm B ,\\
& (\nabla_v \nabla_v \bm F(m\bm v))_{ijk} = \frac{e}{c} \epsilon_{kmn} \left[ -\frac{v_i \delta_{jm}}{\gamma(m\bm v)^3 c^2} - \frac{\delta_{im} v_j}{\gamma(m\bm v)^3 c^2} - \frac{\delta_{ij} v_m}{\gamma(m\bm v)^3 c^2} + \frac{3 v_i v_j v_m}{\gamma(m\bm v)^5 c^4}  \right] B_n.
\end{align}
With these, the degree zero fluid system reads
\begin{align} \label{eq:deg_zero_rel_v}
&\partial_t \bm v_p = -\frac{\bm v_p}{\gamma(m \bm v_p)} \cdot (\nabla \bm v_p) + \frac{e}{m} \left( \bm E + \frac{\bm v_p}{c \gamma(m \bm v_p)  } \times \bm B \right), \\ \label{eq:deg_zero_rel_n}
&\partial_t n_p = - \nabla \cdot \left( \frac{\bm v_p}{\gamma(m \bm v_p) } n_p \right). \\ \notag
\end{align}
The above fluid equations together with the Maxwell equations  (\ref{eq:maxwell_e}-\ref{eq:maxwell_divb}) are known as cold plasma relativistic fluid model in the literature \cite{Tronci_2010, spencer, mukhamet2025, warpx}.

Next, we use the coefficients to obtain the degree one system
\begin{align}
&\partial_t \bm v_p = -\frac{\bm v_p}{\gamma(m \bm v_p)} \cdot (\nabla \bm v_p) + \frac{e}{m} \left( \bm E + \frac{\bm v_p}{c \gamma(m \bm v_p)  } \times \bm B \right), \\
&\partial_t n_p = - \nabla \cdot \left( \frac{\bm v_p}{\gamma(m \bm v_p) } n_p + \frac{\bm P_p}{m \gamma(m \bm v_p) } - \frac{(\bm P_p \cdot \bm v_p) \bm v_p}{\gamma(m \bm v_p)^ 3 m c^2} \right), \\ \notag
&\partial_t \bm P_p = - \left( \bm P_p \otimes \frac{\bm v_p}{\gamma(m \bm v_p)} \right) \cdot \nabla - m \left( \frac{\bm P_p}{m \gamma(m \bm v_p)} - \frac{(\bm P_p \cdot \bm v_p) \bm v_p}{\gamma(m \bm v_p)^3 m c^2} \right) \cdot \nabla \bm v_p  \\ \notag
& \hspace{1cm} + \frac{e}{cm} \left( \frac{\bm P_p}{\gamma(m \bm v_p)} - \frac{(\bm P_p \cdot \bm v_p) \bm v_p}{\gamma(m \bm v_p)^3 c^2} \right) \times \bm B.
\end{align}

And similarly the degree two system
\begin{align} \label{eq:deg_two_rel_v}
&\partial_t \bm v_p = -\frac{\bm v_p}{\gamma(m \bm v_p)} \cdot (\nabla \bm v_p) + \frac{e}{m} \left( \bm E + \frac{\bm v_p}{c \gamma(m \bm v_p)  } \times \bm B \right), \\
&\partial_t n_p = - \nabla \cdot \left( \frac{\bm v_p}{\gamma(m \bm v_p) } n_p + \frac{\bm P_p}{m \gamma(m \bm v_p) } - \frac{(\bm P_p \cdot \bm v_p) \bm v_p}{\gamma(m \bm v_p)^3 m c^2} - \frac{\bm v_p \cdot \bm S_p}{\gamma(m \bm v_p)^3 m^2 c^2} - \frac{\bm S_p \cdot \bm v_p}{\gamma(m \bm v_p)^3 m^2 c^2} \right. \\ \notag 
& \hspace{2cm} \left. \frac{3 \bm v_p \ (\bm v_p \cdot \bm S_p \cdot \bm v_p)}{\gamma(m \bm v_p)^5 m^2 c^4} - \frac{(\bm S_p : \bm I) \bm v_p}{\gamma(m \bm v_p)^3 m^2 c^2} \right) ,\\ \notag
&\partial_t \bm P_p = - \left( \bm P_p \otimes \frac{\bm v_p}{\gamma(m \bm v_p)} \right) \cdot \nabla - m \left( \frac{\bm P_p}{m \gamma(m \bm v_p)} - \frac{(\bm P_p \cdot \bm v_p) \bm v_p}{\gamma(m \bm v_p)^3 m c^2} \right) \cdot \nabla \bm v_p - \left( \frac{\bm S_p}{m \gamma(m \bm v_p)} - \frac{(\bm S_p \cdot \bm v_p) \bm v_p}{\gamma(m \bm v_p)^3 m c^2} \right) \cdot \nabla \\ \notag
& \hspace{1.cm} + \frac{m}{2} \left( \frac{\bm v_p \cdot \bm S_p}{\gamma(m \bm v_p)^3 m^2 c^2} +\frac{\bm S_p \cdot \bm v_p}{\gamma(m \bm v_p)^3 m^2 c^2}  + \frac{(\bm S_p: \bm I) \bm v_p}{\gamma(m \bm v_p)^3 m^2 c^2} - \frac{ 3 \bm v_p (\bm v_p \cdot \bm S_p \cdot \bm v_p) }{\gamma(m \bm v_p)^5 m^2 c^4 } \right) \cdot \nabla \bm v_p \\ \label{eq:deg_two_rel_S}
& \hspace{1cm} + \frac{e}{cm} \left( \frac{\bm P_p}{\gamma(m \bm v_p)} - \frac{(\bm P_p \cdot \bm v_p) \bm v_p}{\gamma(m \bm v_p)^3 c^2} \right) \times \bm B \\ \notag
& \hspace{1cm} - \frac{e}{2 cm} \left[ \frac{\bm v_p \cdot \bm S_p}{\gamma(m \bm v_p)^3 m c^2} +\frac{\bm S_p \cdot \bm v_p}{\gamma(m \bm v_p)^3 m c^2} + \frac{(\bm S_p : \bm I) \bm v_p}{\gamma(m \bm v_p)^3 m c^2} - \frac{3 \bm v_p(\bm v_p \cdot \bm S_p \cdot \bm v_p)}{\gamma(m \bm v_p)^5 m c^4} \right] \times \bm B, \\
&\partial_t \bm S_p = - \left( \bm S_p \otimes \frac{\bm v_p}{\gamma(m \bm v_p)} \right) \cdot \nabla \\ \notag
& \hspace{1cm}  - \bm S_p \cdot \left[ \frac{\bm I}{\gamma(m \bm v_p)} - \frac{\bm v_p \otimes \bm v_p}{\gamma(m \bm v_p)^3 c^2} \right] \cdot (\nabla \bm v_p)  - (\bm v_p \nabla) \cdot \left( \frac{\bm I}{\gamma(m \bm v_p)} - \frac{\bm v_p \otimes \bm v_p}{\gamma(m \bm v_p)^3 c^2} \right) \cdot \bm S_p \\ \notag
& \hspace{1cm}   + \frac{e}{cm} \left[ \frac{\bm S_p \times^1 \bm B}{\gamma(m \bm v_p)} - \frac{(\bm v_p \times \bm B) \otimes (\bm v_p \cdot \bm S_p)}{\gamma(m \bm v_p)^3 c^2} \right]  + \frac{e}{cm} \left( \frac{\bm S_p \times^2 \bm B}{\gamma(m \bm v_p)} - \frac{(\bm S_p \cdot \bm v_p) \otimes (\bm v_p \times \bm B) }{\gamma(m \bm v_p)^3 c^2} \right).
\end{align}
These models are coupled to the Maxwell equations via the Amp\`ere equation \eqref{eq:maxwell_e} using current $\bm J$. Explicitly for the zero fluid model the current $\bm J$ is given by 
\begin{align} \label{eq:current_degree_zero}
\bm J &= \frac{e}{m} \int \frac{\bm u}{\gamma(\bm u)} \left(\sum_p^N g_p \right) \, du = \frac{e}{m} \sum_{p}^N \left[ mn_p \frac{\bm v_p}{\gamma(m\bm v_p)} \right]
\end{align}
and in full analogy one computes the current for the degree one model
\begin{align} \label{eq:current_degree_one}
\bm J & = \frac{e}{m} \sum_{p}^N \left[ mn_p \frac{\bm v_p}{\gamma(m\bm v_p)} + \frac{\bm P_p}{\gamma(m\bm v_p)}-\frac{(\bm v_p \cdot \bm P_p)}{c^2 \gamma(m\bm v_p)^3} \bm v_p \right]
\end{align}
and the degree two model
\begin{align} \label{eq:current_degree_two}
\bm J &=  \frac{e}{m} \sum_{p}^N \left[ mn_p \frac{\bm v_p}{\gamma(m \bm v_p)} + \frac{\bm P_p}{\gamma(m \bm v_p)}-\frac{(\bm v_p \cdot \bm P_p)}{c^2 \gamma(m \bm v_p)^3} \bm v_p - \frac{\bm S_p \cdot \bm v_p}{2 \gamma(m \bm v_p)^3 m c^2} \right. \\ \notag
& \hspace{4.2cm} \left. - \frac{\bm v_p \cdot \bm S_p}{2 \gamma(m \bm v_p)^3 m c^2} - \frac{\bm v_p (\bm I : \bm S_p)}{2 \gamma(m \bm v_p)^3 m c^2} + \frac{3 \bm v_p (\bm v_p \cdot \bm S_p \cdot \bm v_p)}{2 \gamma(m \bm v_p)^5 m c^4} \right]
\end{align}

Similarly we can obtain the phase space models (particle models). The degree zero model reads 
\begin{align} \label{eq:particle_1}
\partial_t \bm X_p(t) &= \frac{\bm U_p}{m \gamma(\bm U_p)} ,\\ \label{eq:particle_2}
\partial_t \bm U_p(t) &= e \left( \bm E(\bm X_p, t)+ \frac{\bm U_p \times \bm B(\bm X_p, t)}{cm \gamma(\bm U_p)} \right).
\end{align}

To derive the degree one model, in addition to (\ref{eq:particle_1}-\ref{eq:particle_2}) we need an equation for $\partial_t \bm W^1_p$:
\begin{align}
\partial_t \bm W^1_p = \begin{bmatrix}
\bm W_{p,u}^1 \\
\bm W_{p,x}^1
\end{bmatrix}^\top
\begin{bmatrix}
\nabla_u \bm A_u & \nabla_u \bm A_x \\
\nabla_x \bm A_u & \nabla_x \bm A_x
\end{bmatrix},
\end{align}
after calculating gradients, the degree one model reads
\begin{align}
\partial_t \bm X_p(t) &= \frac{\bm U_p}{m \gamma(\bm U_p)}, \\ 
\partial_t \bm U_p(t) &= e \left( \bm E(\bm X_p, t)+ \frac{\bm U_p \times \bm B(\bm X_p, t)}{cm \gamma(\bm U_p)} \right) ,\\
\partial_t \bm W^1_p &= \begin{bmatrix}
\frac{e}{cm} \left[ \frac{\bm W_{p,u}^1 \times \bm B(\bm X_p, t)}{\gamma(\bm U_p)} - \frac{(\bm W_{p,u}^1 \cdot \bm U_p) (\bm U_p \times \bm B(\bm X_p, t))}{m^2 c^2 \gamma(\bm U_p)^3} \right] + e \left[\bm W_{p,x}^1 \cdot \nabla \bm E(\bm X_p,t) - \frac{(\bm W_{p,x}^1 \cdot \nabla \bm B(\bm X_p,t)) \times \bm U_p }{m \, c \, \gamma(\bm U_p)} \right] \\
\frac{1}{m} \left[ \frac{\bm W_{p,u}^1}{\gamma(\bm U_p)} - \frac{\bm U_p \, (\bm W_{p,u}^1 \cdot \bm U_p)}{m^2 \, c^2 \, \gamma(\bm U_p)^3} \right]
\end{bmatrix},
\end{align}

We omit the equations for the particle model of degree two owing to its complexity.

The model is coupled to the Maxwell equations via the current $\bm J$ in the distribution version of the Amp\`ere's equation \eqref{eq:maxwell_e}. Explicitly for the degree-one particle model the current $\bm J$ is given by
\begin{align} 
\bm J = \frac{e}{m} \int \frac{\bm u}{\gamma(\bm u)} \left( \sum_p^N g_p \right) \, du &= \frac{e}{m} \sum_p^N \left\{ W^0_p \frac{\bm U_p}{\gamma(\bm U_p)} \delta(\bm x-\bm X_p) \right. \\ \notag
& + \left. \left[ \frac{\bm W^1_{p,u}}{\gamma(\bm U_p)} - \frac{\bm U_p \left( \bm W^1_{p,u} \cdot \bm U_p \right)}{m^2 \, c^2 \, \gamma(\bm U_p)^3} \right] \delta(\bm x-\bm X_p) - \frac{e}{m} \bm W^1_{p,x} \cdot \nabla_x \delta(\bm x- \bm X_p) \frac{\bm U_p}{\gamma(\bm U_p)} \right\}.
\end{align}
This needs to be regularized just like the current \eqref{eq:regul_current} was regularized using convolution over the Gaussian kernel. The current for the hybrid particle-fluid model is obtained by summing together the fluid and particle components of the current. Numerical discretization of the hybrid model of degree zero was considered in \cite{mukhamet2025, warpx}.

These fluid and particle models preserve total mass $\int f \, dx \, du$, total energy $\int f \left( \gamma(\bm u)- 1 \right) m c^2 \, dx \, du + \frac{1}{8 \pi} \int \left( \bm E^2 + \bm B^2 \right) \, dx $, and momentum $\int f \, \bm u \, dx \, du + \frac{1}{4 \, \pi \, c} \int \left( \bm E \times \bm B \right) dx$. Let us show for example conservation of total energy.

\begin{proposition}
Let $f$ has a compact support and be a distributional solution of the relativistic Vlasov--Maxwell system. Then $f$ conserves total energy.
\end{proposition}
\begin{proof}
By the assumption, $f$ satisfies
\begin{align} \label{eq:rel_vlasov}
\int_{\Omega} \partial_t f \phi \, dV = -\int_{\Omega} \nabla \cdot \left( \frac{\bm u}{m \, \gamma(\bm u)} f \right) \phi \, dV - \int_{\Omega} \nabla_u \cdot \left( e \left( \bm E + \frac{\bm u}{cm \, \gamma(\bm u)} \times \bm B \right) f \right) \phi \, dV, \hspace{1cm} \forall \phi \in C^{\infty}_0(\Omega).
\end{align}
In \eqref{eq:rel_vlasov}, setting $\phi = (\gamma(\bm u) - 1) m \, c^2 \, \phi_x \, \phi_u \, \phi_t$, where $\phi_u \in C^\infty_0(\Omega_u)$, $\phi_x \in C^\infty_0(\Omega_x)$, $\phi_t \in C^\infty_0([0,T])$, and applying integration by parts we find
\begin{align}
\int \partial_t f (\gamma(\bm u)-1) mc^2 \phi_x \, \phi_u \, \phi_t \, dV &= \int \left( \frac{\bm u}{m \gamma(\bm u)} f \right) (\gamma(\bm u)-1) mc^2 \, \cdot  \nabla_x \phi_x \, \phi_u \, \phi_t \, dV \\ \notag
& + \int e \left( \bm E + \frac{\bm u}{cm \gamma(\bm u)} \times \bm B \right)f \cdot \frac{\bm u}{m \gamma(\bm u)} \phi_x \, \phi_u \, \phi_t \, dV  \\ \notag
& + \int e \left( \bm E + \frac{\bm u}{cm \gamma(\bm u)} \right) f (\gamma(\bm u)-1) mc^2 \cdot  \phi_x \nabla_u \phi_u \phi_t \, dV.
\end{align}
Choosing $\phi_x$ and $\phi_u$ that are $1$ in the support of $f$, the terms containing gradients vanish and we obtain
\begin{align} \label{eq:prop_rel_vm}
\int \partial_t f (\gamma(\bm u)-1) mc^2 \phi_t \, dV &= \int e \left( \bm E + \frac{\bm u}{cm \gamma(\bm u)} \times \bm B \right)f \cdot \frac{\bm u}{m \gamma(\bm u)} \, \phi_t \, dV  .
\end{align}
Now we consider the energy rate
\begin{align}
\frac{dH}{dt}  = \int \partial_t f \left( \gamma(\bm u)- 1 \right) m c^2 \, dx \, du + \frac{1}{8 \pi} \int \left( \partial_t \bm E^2 + \partial_t \bm B^2 \right) \, dx ,
\end{align}
from the Maxwell equations, we have the identity
\begin{align} \label{eq:prop_maxwell_simpl}
\frac{1}{8 \pi} \int \left( \partial_t \bm E^2 + \partial_t \bm B^2 \right) dx = \frac{1}{4 \pi} \int \left[ c \nabla \cdot \left( \bm B \times \bm E\right) - 4 \pi \bm J \cdot \bm E \right] \, dx , 
\end{align}
the term  $c \nabla \cdot \left( \bm B \times \bm E\right)$ gives boundary term (Green's theorem) and vanishes. 
Using \eqref{eq:prop_maxwell_simpl} we have
\begin{align}
\frac{dH}{dt}  = \int \partial_t f \left( \gamma(\bm u)- 1 \right) m c^2 \, dx \, du - \int \bm J \cdot \bm E \, dx,
\end{align}
Using \eqref{eq:prop_rel_vm}
\begin{align}
\int \frac{dH}{dt} \phi_t \, dt  = \int e \left( \bm E + \frac{\bm u}{cm \gamma(\bm u)} \times \bm B \right)f \cdot \frac{\bm u}{m \gamma(\bm u)} \, \phi_t \, dV - \frac{e}{m} \int \frac{\bm u}{\gamma(\bm u)}f \cdot \bm E \phi_t dV = 0,
\end{align}
for all $\phi_t \in C^\infty_0([0,T])$.
\end{proof}

\section{The slicing moment method} \label{sec:slicing}
In this section, we describe the details of our slicing method.  The slicing method finds the initial conditions for $n_{s,p}$ and $\bm v_{s,p}$ of the ansatz \eqref{eq:slicing_ansatz}. For the illustration purposes, consider an arbitrary distribution $f(\bm u, \bm x, t)$ that we want to approximate with the ansatz \eqref{eq:slicing_ansatz}. We assume that the distribution $f$ has compact support or some reasonable boundary conditions such as the periodic boundaries in $\Omega_x$. We assume partitioning of $\Omega_u$ into a finite sequence of disjoint sets $\{\Omega_{u,s}\}_{s=1}^L$, i.e., $\Omega_u = \cup_{s=1}^{L} \Omega_{u,s}$ and $\Omega_{u,i} \cap \Omega_{u,j} = \emptyset $ for $i \neq j$. We then define a slice $\Omega_s := \Omega_x \times \Omega_{u,s}$ as well as denote by $\bm u_s$ the centroid of $\Omega_{u,s}$. Next, we consider the moments with respect to the centroid $u_s$ of the restriction $f|_{\Omega_s}$, i.e., $\bm C_s^k = \int_{\Omega_{u,s}} f|_{\Omega_s} (\bm u-\bm u_s)^{\otimes k} du$. Assuming the parametrization $f_{s} = \sum_{p=1}^{N_s} n_{s,p}(\bm x, t) \delta(\bm u- \bm v_{s,p}(\bm x, t))$ for each slice, we require the moments of the restriction and the parametrization to match which leads to the conditions:
\begin{align} \notag
C_s^0 &= \sum_{p=1}^{N_s} n_{s,p}, \\ \notag
\bm C_s^1 &= \sum_{p=1}^{N_s} n_{s,p} (\bm v_{s,p} - \bm u_s), \\ \label{eq:moment_inversion}
\vdots \\ \notag
\bm C_s^k &= \sum_{p=1}^{N_s} n_{s,p} (\bm v_{s,p} - \bm u_s)^{\otimes k}.
\end{align}
If we denote the difference $(\bm v_{s,p} - \bm u_s)$ by some new variable, the system \eqref{eq:moment_inversion} can be seen as the moment inversion problem and can be solved with the moment inversion algorithms \cite[Chapter 3]{crestetto}. 
{
In particular for $N_s=2$ analytical solution is available, see \cite[Proposition 2.1]{DESJARDINS20082514}, that considers the moments $\{\bm C_s^0, \bm C_s^1,\bm  C_s^2, \bm C_s^3\}$. The analytical solution is given by: 
\begin{align}
y_{s,i} = \frac{C_{s,i}^1}{C_s^0},
\,
\sigma_{s,i} = \frac{\sqrt{C_s^0 C_{s,ii}^2 - (C_{s,i}^1)^2}}{C_s^0},
\,
q_s
=
\sum_{i=1}^{d_u} \left( \frac{ C_{s,iii}^3 - C_s^0 y_{s,i}^3 - 3 C_s^0 \sigma_{s,i}^2 y_{s,i}}{C_s^0} \right),
\, \\
x_s
=
\frac{q_s/2}{\sqrt{q_s^2 + 4 \left( \sum_{i=1}^{d_u}\sigma_{s,i}^3 \right)^2}} .
\end{align}
\begin{align}
n_{s,1} &= \left(\frac{1}{2}+x\right) C_s^0,\\
n_{s,2} &= \left(\frac{1}{2}-x\right) C_s^0,\\
v_{s,1,i} &= u_{s,i} + y_{s,i} - \sqrt{\frac{n_{s,2}}{n_{s,1}}}\,\sigma_{s,i}, \hspace{1cm} \text{for } i \in 1,...,d_u, \\
v_{s,2,i} &= u_{s,i} + y_{s,i} + \sqrt{\frac{n_{s,1}}{n_{s,2}}}\,\sigma_{s,i},\hspace{1cm} \text{for } i \in 1,...,d_u .
\end{align}
As discussed in \cite{DESJARDINS20082514}, these solutions are not unique (except for the 1D case) but they are the only ones invariant under the axis permutation. Furthermore, it is assumed that $C_{s}^0 >0$ and $C_s^0 C_{s,ii}^{2} \geq \left(C_{s,i}^1\right)^2$. The first condition holds for positive distributions and the second condition holds by Cauchy-Schwarz inequality.

For our numerical tests, the analytical solutions in 1D reduce to:
\begin{align} 
y = \frac{C_{s}^1}{C_s^0},
\qquad
\sigma = \frac{\sqrt{C_s^0 C_s^2 - (C_s^1)^2}}{C_s^0},
q
=
\frac{ C_s^3 - C_s^0 y^3 - 3 C_s^0 \sigma_p^2 d_p}{C_s^0},
\qquad
x
=
\frac{q/2}{\sqrt{q^2 + 4\sigma^6}} 
\end{align}
\begin{align} \label{eq:ns1}
n_{s,1} &= \left(\frac{1}{2}+x\right) C_s^0,\\
n_{s,2} &= \left(\frac{1}{2}-x\right) C_s^0,\\
v_{s,1} &= u_s + y - \sqrt{\frac{n_{s,2}}{n_{s,1}}}\,\sigma,\\ \label{eq:vs2}
v_{s,2} &= u_s + y + \sqrt{\frac{n_{s,1}}{n_{s,2}}}\,\sigma,
\end{align}
}
{Once the initial conditions are setup with the slicing method, then the values of $n_{s,p}$ and $v_{s,p}$ are evolved using the degree zero moment model given by (\ref{eq:deg_zero_v_p}-\ref{eq:deg_zero_n_p}) and coupled to $\bm{\mathcal {A}}[\bm F, \bm G, \sum_p g_p]$. In the Vlasov--Maxwell case the corresponding system is given by (\ref{eq:deg_zero_moment_system_v} - \ref{eq:deg_zero_moment_system_n}) coupled to the electric and the magnetic fields (\ref{eq:maxwell_e}-\ref{eq:maxwell_b}) by the current \eqref{eq:current_deg_zero}. In the relativistic Vlasov--Maxwell case, each node solves (\ref{eq:deg_zero_rel_v}-\ref{eq:deg_zero_rel_n}) coupled to the electric and magnetic fields (\ref{eq:maxwell_e}-\ref{eq:maxwell_b}) by the relativistic current \eqref{eq:current_degree_zero}.
}

The slicing method can also be applied to the more general ansatz \eqref{eq:burby_ansatz_node_p}. The simplest case is when each slice is approximated with one node. Thus, the abscissa $\bm v_{s,p}$ is set to the centroid $\bm u_s$ and the moments $\bm C_{s,p}^k$ correspond exactly to the centered moments of the slice with respect to $\bm u_s$. For example, the degree zero approximation with two nodes discussed earlier can be alternatively approximated with a degree three ($k=3$) single node approximation.

\section{Fluid to particle transformation} \label{sec:fluid_pic_transf}
The major disadvantage of the ansatz \eqref{eq:slicing_ansatz} is that in some simulations the mapping $\bm v_{s,p}$ ceases to be one-to-one and singularities appear in $n_{s,p}$. This phenomenon is known as multi-streaming \cite{Uhlemann_2018, burby2023} and the associated singularities can blow up numerical solvers. For the sake of simplicity, let us perform relabeling and refer to any node by a single unique index $p$, i.e., $n_p$ and $\bm v_p$ with $p$ running up from $0$ to $N_s \times L$. In order to avoid the singularities, we propose to replace the specific nodes with particles. First, a criterion is needed to determine whether the specific node $p$ develops singularity and needs replacement. We find the following criterion to perform well in practice:
\begin{align} \label{eq:criterion}
\|n_p(\cdot,t)\|_{\infty} >  C\, \|n_p(\cdot,0)\|_{\infty} 
\end{align}      
where $\|n_p( \cdot,t)\|_\infty$ is the maximum value of $n_p$ over the domain $\Omega_x$ at time $t$ and $C>0$ is the threshold. The value of $C$ should be kept rather small, in the range $(1,5]$, in order to transform the critical nodes before the numerical instabilities corrupt the simulation or lead to the loss of accuracy. In our simulations we consider $C \in (1,3]$.

{
If the node $p$ satisfies the condition \eqref{eq:criterion}, we then perform the transformation of $n_p$ and $\bm v_p$ into particles. The simplest transformation would be to partition $\Omega_x$ into slices $\Omega_{x,s}$, that we call $x$-slices, and determine their centroids $\bm x_s$. Then, on every $x$-slice we create a single particle with the position $\bm X = \bm x_s$ and momentum $\bm U = \frac{1}{|\Omega_{x,s}|}\int_{\Omega_{x,s}} \bm v_p(\bm x_s) dx$. The weight of the particle is given by $W = \int_{\Omega_{x,s}} n_p dx$.

In practice, the nodes $n_p$ are discretized with some numerical discretization such as the finite volume or the finite element methods. In this case each cell can be considered as an $x$-slice thus giving one particle per cell after the transformation of every node $p$. If the each cell is sliced into several $x$-slices before the transformation, then several particles are created per cell per fluid node. The number of particles that are created per cell per fluid node are shortly abbreviated as "ppcpfn". Since the cells are uniform this leads to uniform $x$-slices and in particular the $x$-positions of particles are aligned after the transformation of fluid nodes. The distribution of particles in the momentum space may or may not be uniform depending on the graph of $\bm v_{s,p}$. The uniformity of particles in the $x$-space can be removed by using non-uniform slicing. Nevertheless, in all numerical experiments we consider the case when the $x$-slicing is uniform.

\begin{remark}
The fluid to particle transformation is not exact since the parametrized distributions before and after the transformation are not the same in particular the particle ansatz is represented by $\delta(\bm x - \bm X_k)$-basis over the configuration space while the fluid is represented by $n_p(\bm x,t)$. Despite of the loss of exactness at the moment of transformation, the exactness is maintained with respect to the transformed distribution after the transformation. Our numerical experiments show that this transformation possesses good approximation properties.
\end{remark}

Lastly, we mention that when transforming all fluid nodes at time $t=0$ into particles, we can consider the method as a way to sample particles. In the numerical section, we would compare this approach with the Monte-Carlo sampling \cite{bottino}.

}

\section{Numerical results} \label{sec:numerics}
In this section, we present the numerical results in 1D-1U for the Vlasov--Amp\`ere system given by (\ref{eq:maxwell_f}-\ref{eq:maxwell_e}) and the relativistic Vlasov--Amp\`ere system given by \eqref{eq:rel_maxwell_f} and \eqref{eq:maxwell_e}. The plasma is assumed to be unmagnetized with the magnetic field $\bm B=0$.  We use the slicing moment method from Section \ref{sec:slicing} to set the initial conditions for $n_{s,p}$ and $v_{s,p}$. We consider two-nodes per slice ($N_s=2$) and use the analytical solutions (\ref{eq:ns1}-\ref{eq:vs2}) for the moment inversion problem.

For simplicity we assume normalized units such that the speed of light, mass and charge of the electrons are unity. We use SI formulation that can be obtained from the Gauss formulation by replacing $4 \pi$ by $1/\varepsilon_0$ and $\bm B / c$ by $\bm B$, where $\varepsilon_0$ is the electric constant that we also set to unity. 

We have implemented the equations of motions using a finite volume essentially non-oscillatory scheme with third-order reconstruction in space \cite{Shu1998} and a strong stability preserving Runge-Kutta method of order three in time \cite[Equation 4.2]{ssp_rk}.
{
We also test the fluid to particle transformation as described from Section \ref{sec:fluid_pic_transf}. For this purpose, during the transformation every finite volume cell is sliced into one or several $x$-slices leading to one or several particles per cell after the transformation of each fluid node. For most of the experiments we consider one particle per cell per fluid node (ppcpfn) to keep the total number of degrees of freedom before and after transformation constant.} The particle simulation cannot use a delta-function representation of the distribution function in the computation of the current due to the Finite Volume field solver. Instead we use a hat function around the particle position with a support equal to the width of two grid cells.

In some simulations, we use pure particle simulation for comparison that is obtained by transforming all fluid nodes at time $t=0$ into particles. We refer to this case as reference.

\subsection{Weak Landau damping} \label{sec:weak_landau}
We study the electrostatic example of weak Landau damping with initial distribution
\begin{align} \label{eq:weak_landau_f}
f_0(u, x) = \frac{1}{\sqrt{2 \pi} } \exp\left(-\frac{u^2}{2}\right) \left(1 + \alpha \cos(k x) \right),
\end{align}
on the domain $\Omega_x=[0, 4 \pi]$ and $\Omega_u = [-6, 6]$.
In our expriments, we choose $k=0.5$ and for the weak Landau damping experiment we set $\alpha=0.01$. We use pure fluid represention without any particle transformation. The initial electric field is chosen to satisfy Gauss's law, namely
\begin{align} \label{eq:weak_landau_e}
E_0(x) = \frac{\alpha}{k} \sin(kx)
\end{align}
We test $L=\{5,50,200,500\}$ slices. We use finite volumes on a grid of $2^7$ cells. The time step is set to $\Delta t = 0.05 \, \Delta x / \max_{s}^L \max_{p}^{N_s}  v_{s,p}$. 

Figure \ref{fig:weak_landau_damping} shows that the damping rate for the electric energy agrees with the rate from the linear theory around the Maxwellian equilibrium \cite{sonnendrucker2010approximation}. Already for $L=50$, the accuracy of the method outperforms the quadrature methods of moment, cf. \cite[Figure 20 and 23]{TAUNAY2023111700}.

\begin{figure}[H]
    \centering
    \includegraphics[width=\textwidth]{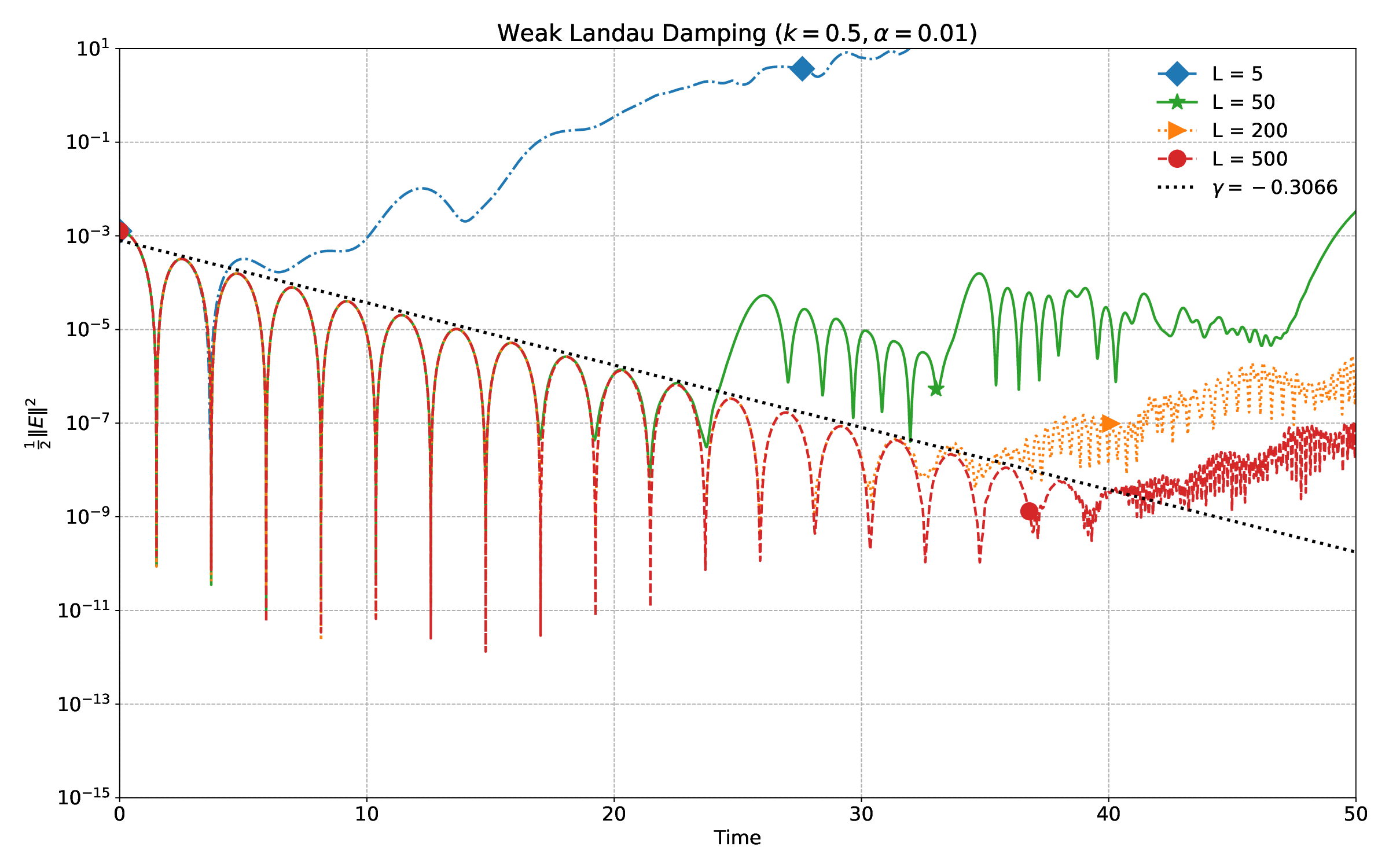}
    
    \caption{Weak Landau damping experiment in 1D. The initial conditions are approximated by $L=\{5,50,500\}$ slices. The theoretical damping rate from the linear theory is given by the dashed line.}
    \label{fig:weak_landau_damping}
\end{figure}

\subsection{Convergence of the fluid-to-particle transformation}
In this section we study the convergence of the fluid to particle transformation. The setup is the same as in Section \ref{sec:weak_landau} except that at time $t=0$ a fraction of fluid nodes is transformed into particles. We consider the case with $L=200$ slices and perform transformation of $\{0\%,25\%,50\%,75\%, 100\%\}$ fluid nodes. The case with $0\%$ transformation corresponds to the pure fluid case, while the case with $100\%$ transformation corresponds to the pure particle case. We consider one ppcpfn for the transformation. Figure \ref{fig:weak_landau_damping_percentage} shows that the transformation does not destroy the accuracy of the fluid simulation. 

\begin{figure}[H]
    \centering
    \includegraphics[width=\textwidth]{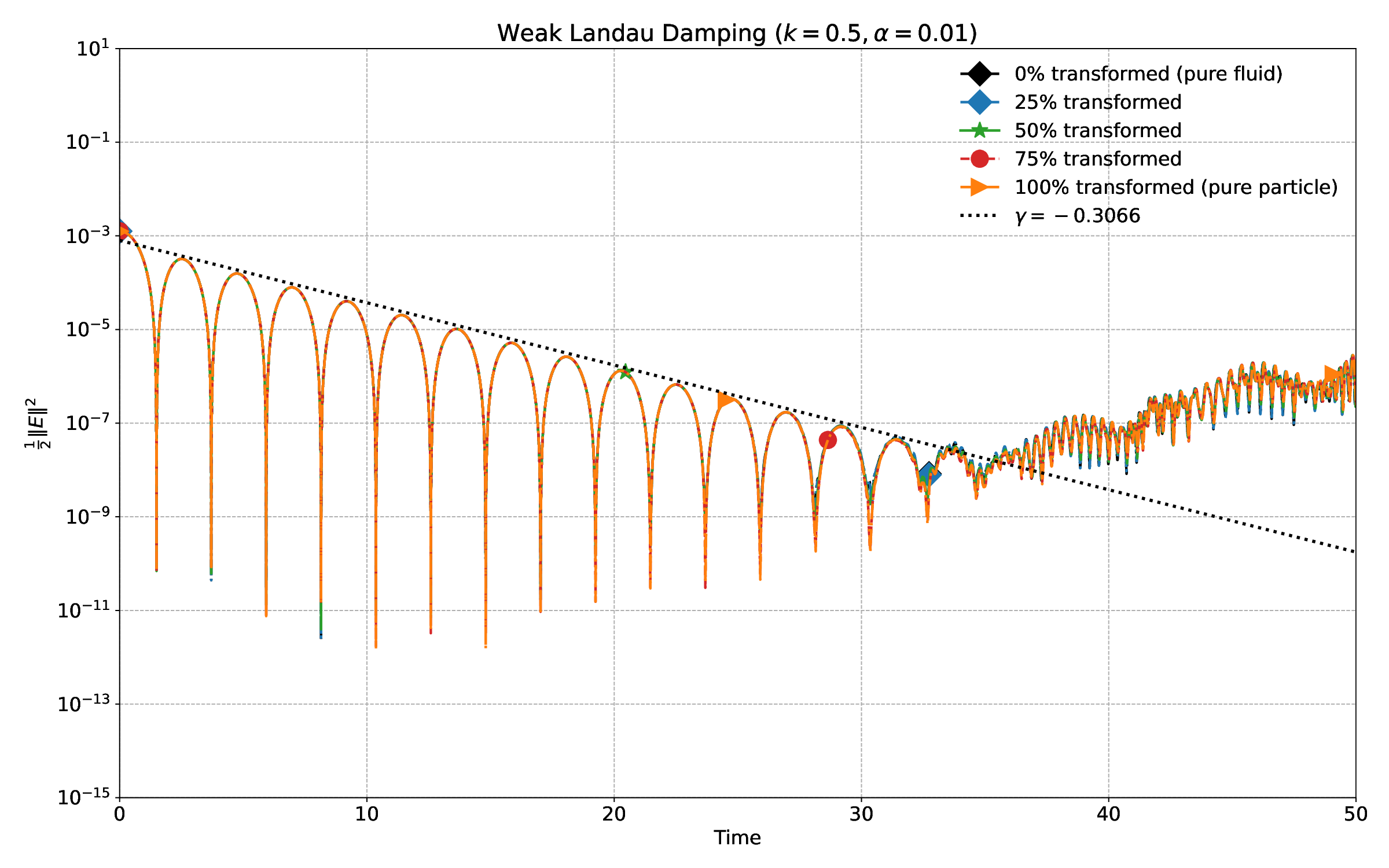}
    
    \caption{Weak Landau damping experiment in 1D. The initial conditions are approximated by $L=200$ slices. The fraction of fluid nodes subject to transformation is varied.}
    \label{fig:weak_landau_damping_percentage}
\end{figure}

Next, we study the effect of the number of particles on the accuracy of the transformation and consider $\{0.5,1,2,4,8\}$ ppcpfn. We perform the transformation of all fluid nodes at time $t=0$. Figure \ref{fig:weak_landau_damping_particle} shows that already good results are obtained with $1$ ppcpfn. Larger values of ppcpfn do not improve the results while the lower number such as $0.5$ ppcpfn which corresponds to one particle per two cells per each fluid node is unstable.

Lastly, we can view this method as particle sampling method and compare it with Monte-Carlo sampling \cite{bottino}. For the Monte-Carlo method we sample the particle momenta from normal distribution with zero mean and unit standard deviation and we sample the particle positions from uniform distribution. Thus the sampling distribution can be written as $g_0(u, x) = \frac{1}{\sqrt{2 \pi} } \exp\left(-\frac{u^2}{2}\right)$. The particle weights are given by $W_k = f_0(U_k, X_k) / g_0(U_k, X_k)$, see \cite{bottino}. We use the same number of particles for the Monte-Carlo sampling as for the transformation sampling. Figure \ref{fig:weak_landau_damping_particle} shows that the Monte-Carlo sampling does not perform well compared to the transformation method, since the statistical noise level is of the order of perturbation. We note that the accuracy of the Monte-Carlo method can be improved using control-variates which we do not consider.

A natural question is to ask about the advantage of the fluid method compared to the pure particle method with the transformation sampling described above. The main advantage of the fluid method lies in the continuity of the fluid ansatz over the configuration space thus one can expect that the fluid method would yield more accurate results than the particle method on the examples when the dynamics is predominantly in the configuration space.

\begin{figure}[H]
    \centering
    \includegraphics[width=\textwidth]{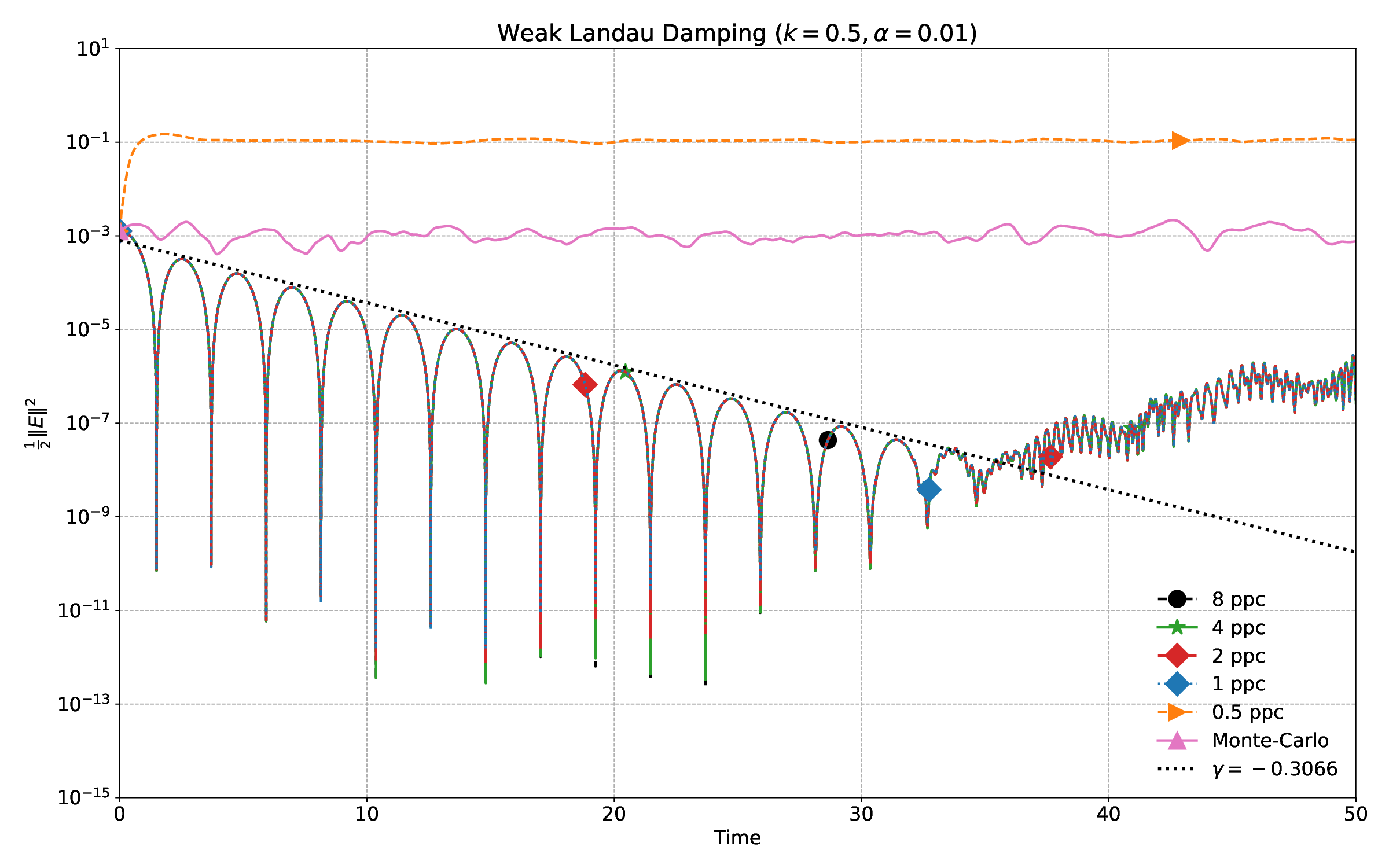}
    
    \caption{Weak Landau damping experiment in 1D. The initial conditions are approximated by $L=200$ slices. We vary the particles per cell per fluid node (ppcpfn) during the transformation.}
    \label{fig:weak_landau_damping_particle}
\end{figure}

\subsection{Two-stream instability}
Next, we consider the two-stream instability. We assume the distribution function
\begin{align}
f(u, x) = \frac{1}{2 \sqrt{2 \pi}} \left( \exp \left( \frac{(u-v_0)^2}{2} \right) + \exp \left(\frac{(u+v_0)^2}{2}\right) \right) (1+\alpha \cos(kx))
\end{align}
with $v_0=2.4$, $\alpha=0.001$, $k=0.2$. The domain is given by $\Omega_x=[0, 2 \pi / k]$ and $\Omega_u = [-8, 8]$. We use $L=400$ slices to approximate the initial conditions. For the finite volume grid we use $2^8$ cells.  The time step is set to $\Delta t = 0.05 \, \Delta x / \max_{s}^L \max_{p}^{N_s}  v_{s,p}$. 

First we test the pure fluid method without any particle transformation to illustrate the multi-streaming and the associated singularities. Figure \ref{fig:tis} shows that the growth rate agrees well with the theoretical growth rate. The method however is not useful past $t \approx 30$. This issue stems from the failure of the mapping $v_p$ on some nodes $p$ to remain one-to-one which leads to singularities in $n_p$.

\begin{figure}[H]
    \centering
    \includegraphics[width=\textwidth]{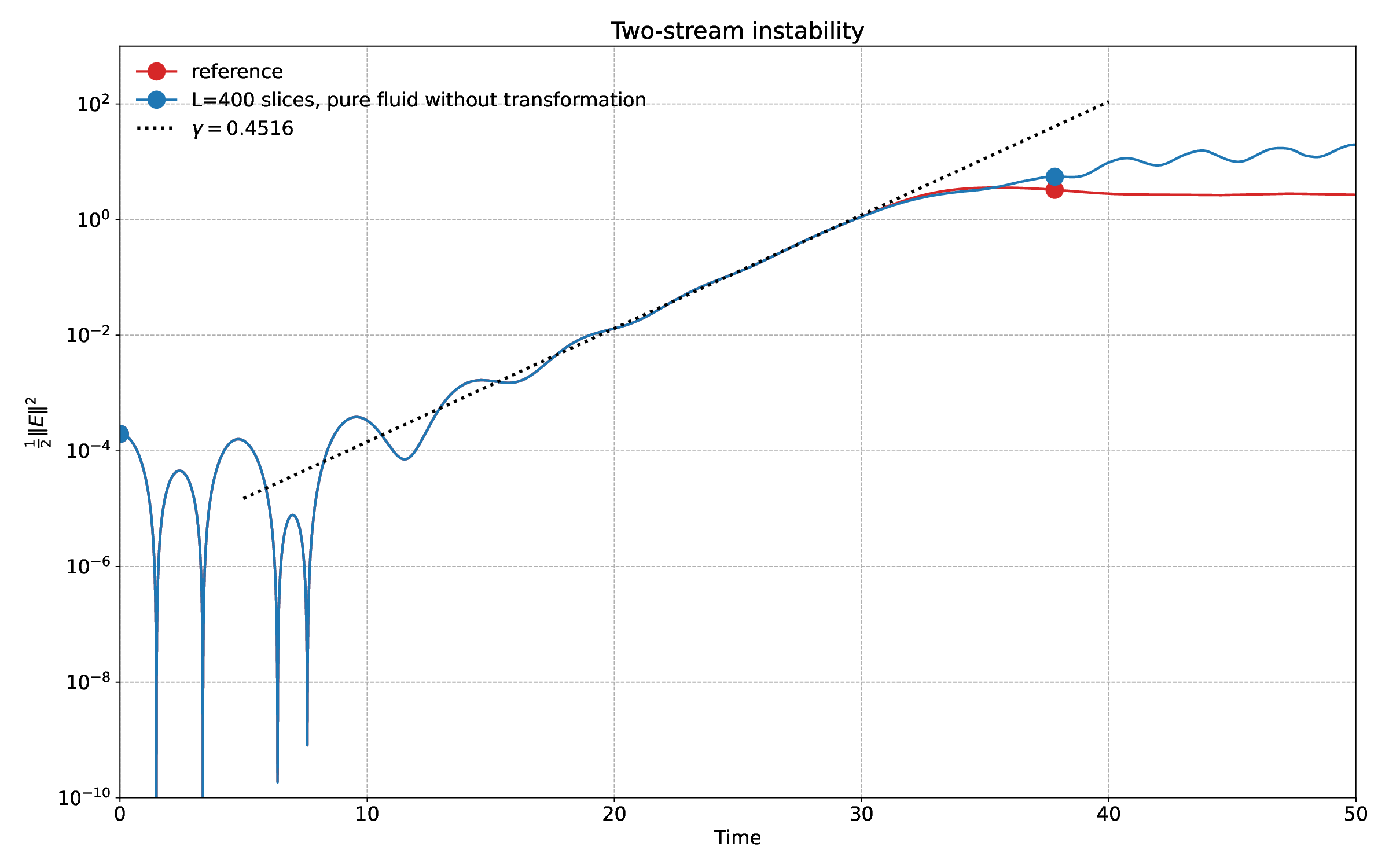}
    
    \caption{Two-stream instability experiment in 1D with fluid discretization and without transformation to particles. The initial conditions are approximated by $L=400$. The theoretical growth rate from the linear theory is given by the dashed line.}
    \label{fig:tis}
\end{figure}

Figure \ref{fig:tis} plots the distribution function at $t=\{0, 30, 40, 50\}$, showcasing the problem and the lack of dynamics in the middle portion.

\begin{figure}[H]
\centering
\begin{subfigure}{0.24\textwidth}
    \includegraphics[width=\textwidth]{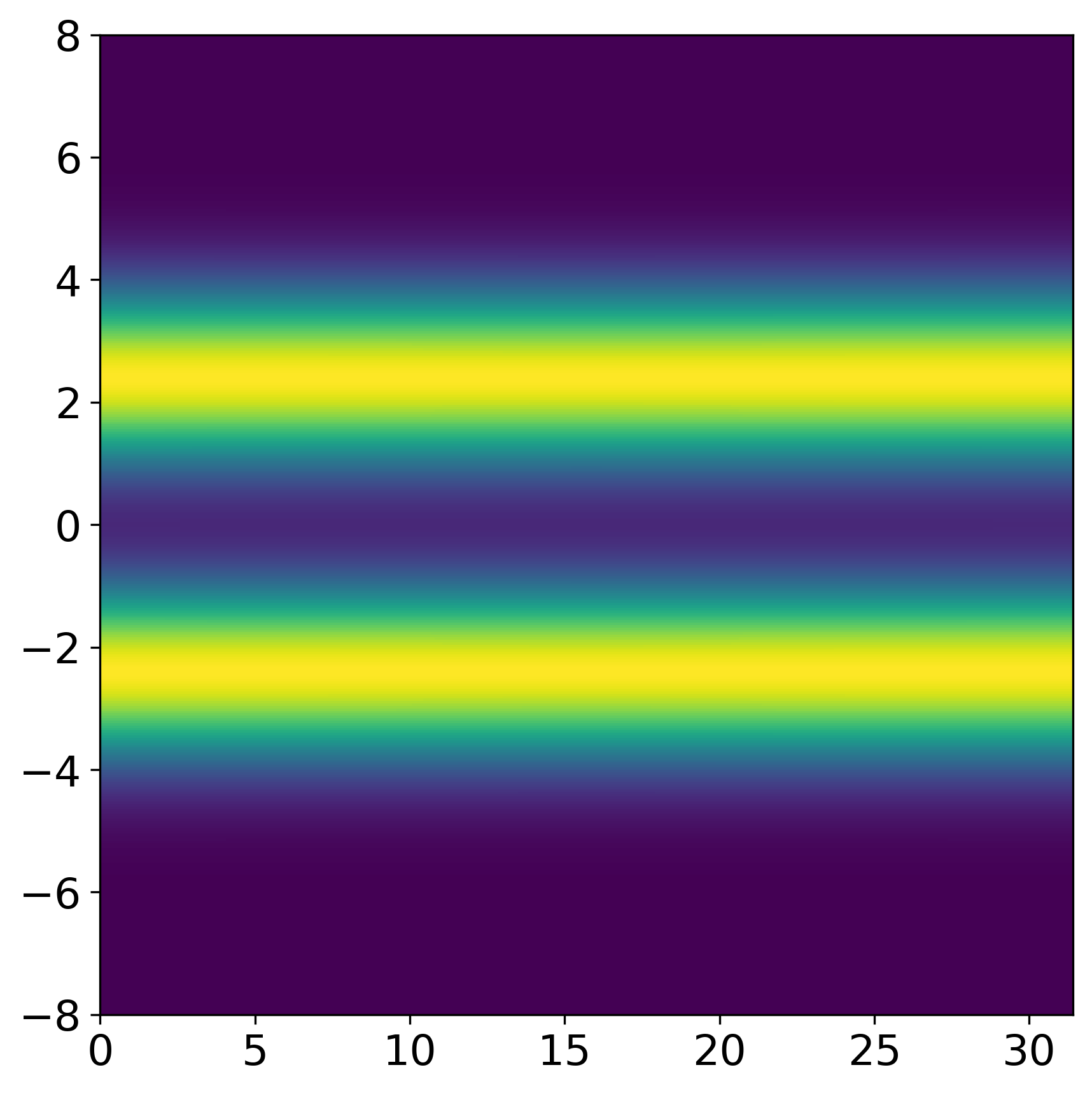}
    \caption{$L=400$, $t=0$}
    \label{fig:s200_0}
\end{subfigure}
\hfill
\begin{subfigure}{0.24\textwidth}
    \includegraphics[width=\textwidth]{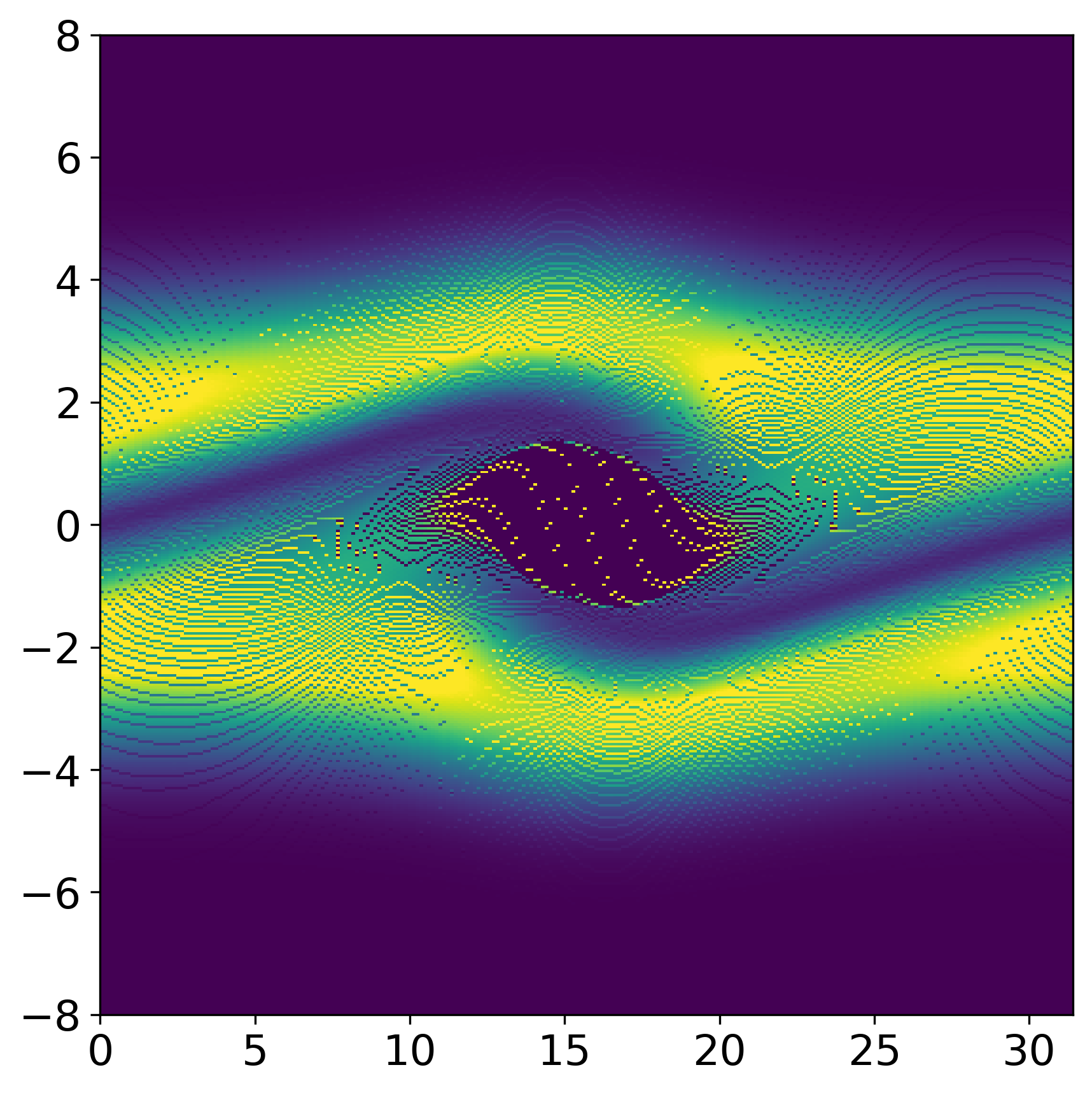}
    \caption{$L=400$, $t=30$}
    \label{fig:s200_28}
\end{subfigure}
\hfill
\begin{subfigure}{0.24\textwidth}
    \includegraphics[width=\textwidth]{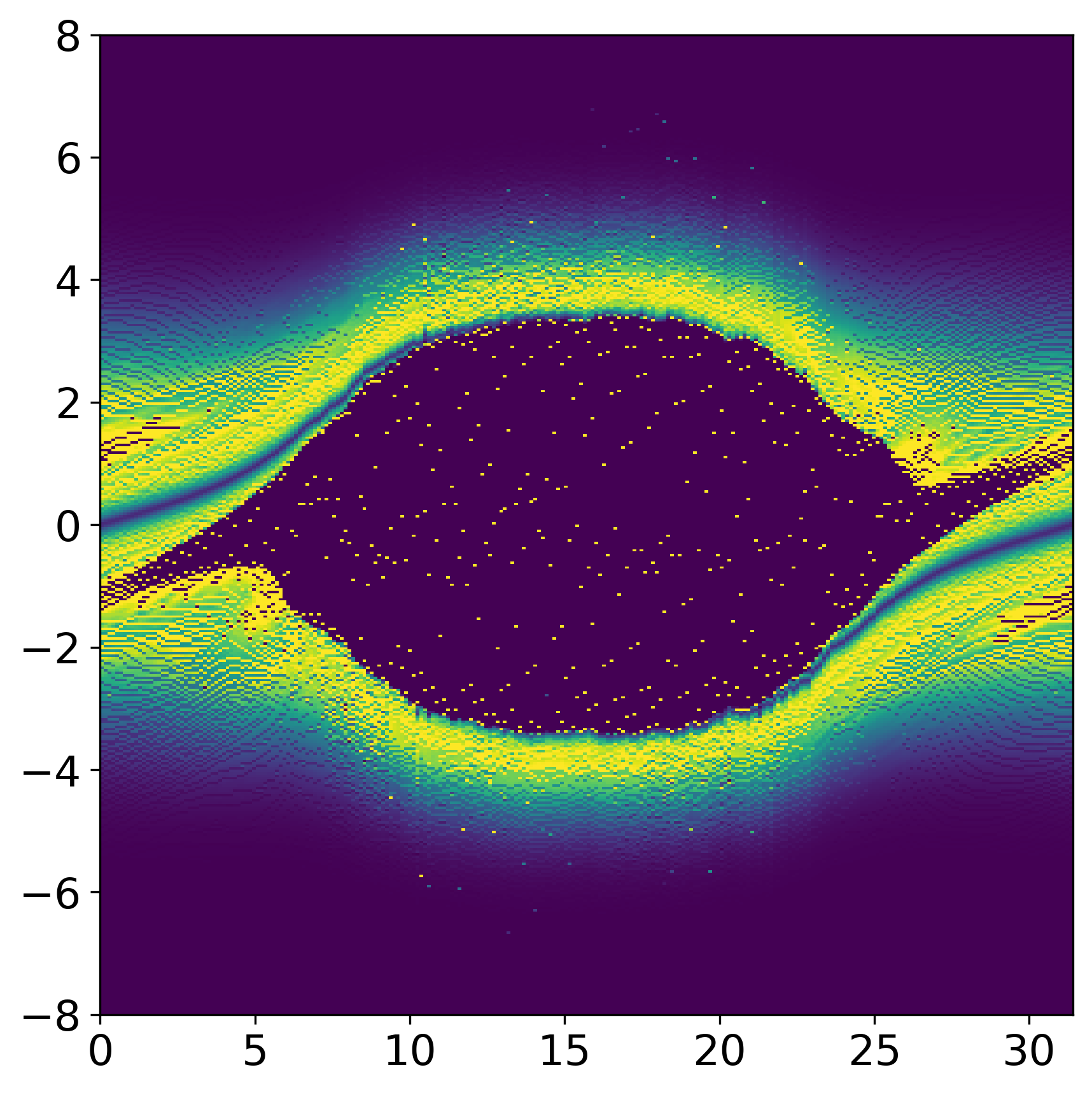}
    \caption{$L=400$, $t=40$}
    \label{fig:s200_30}
\end{subfigure}
\hfill
\begin{subfigure}{0.24\textwidth}
    \includegraphics[width=\textwidth]{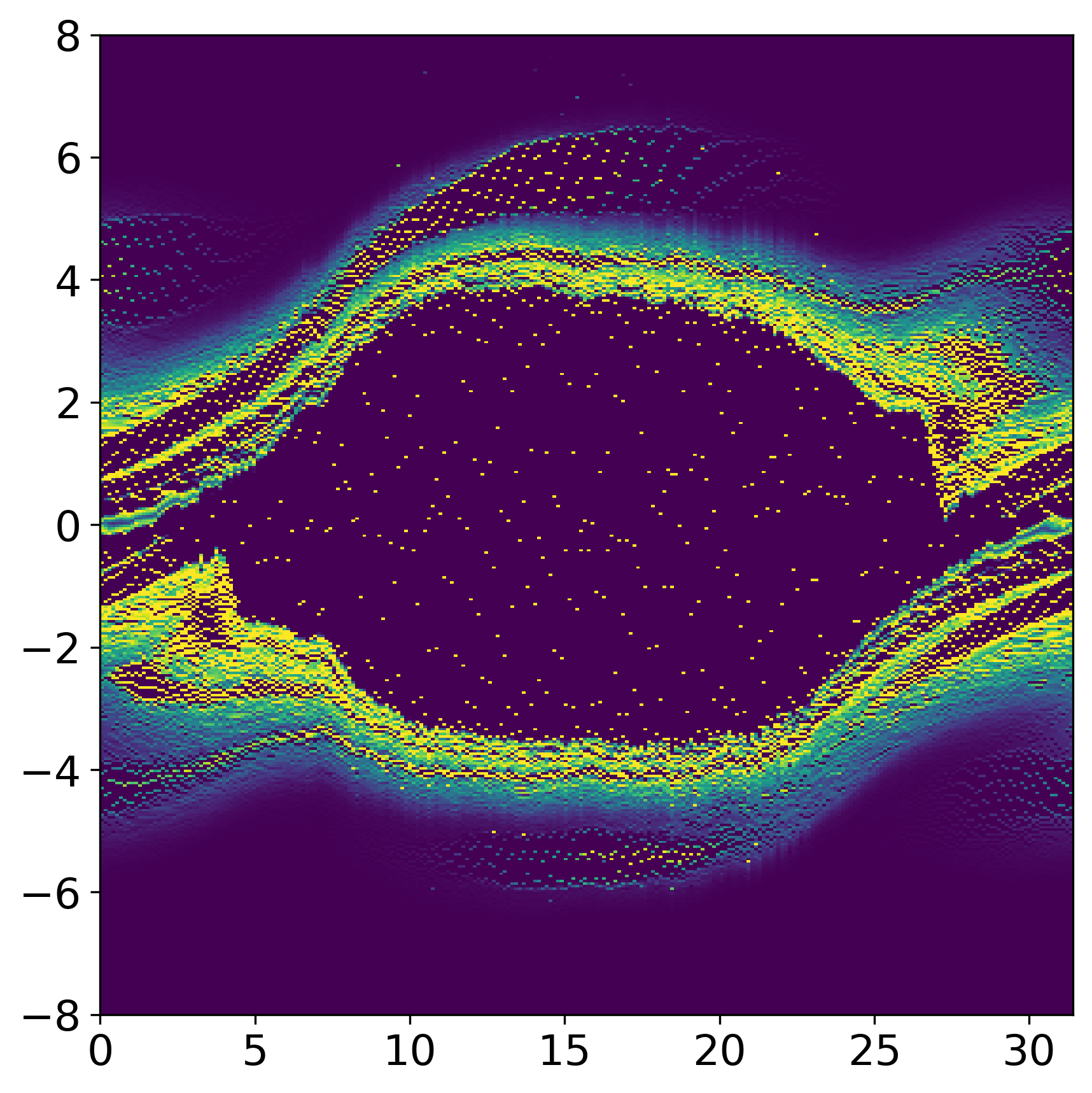}
    \caption{$L=400$, $t=50$}
    \label{fig:s200_34}
\end{subfigure}               
               
\caption{Two-stream instability experiment in 1D  with fluid discretization and without transformation to particles. The initial conditions are approximated by $L=400$ slices.}
\label{fig:tis}
\end{figure}

Figure \ref{fig:tis} plots the velocity $v_{s,p}$ and the density $n_{s,p}$ for one of the problematic nodes illustrating the loss of the one-to-one mapping and the growth of the singularity over time. 
\begin{figure}[H]
    \centering
    \includegraphics[width=0.5\textwidth]{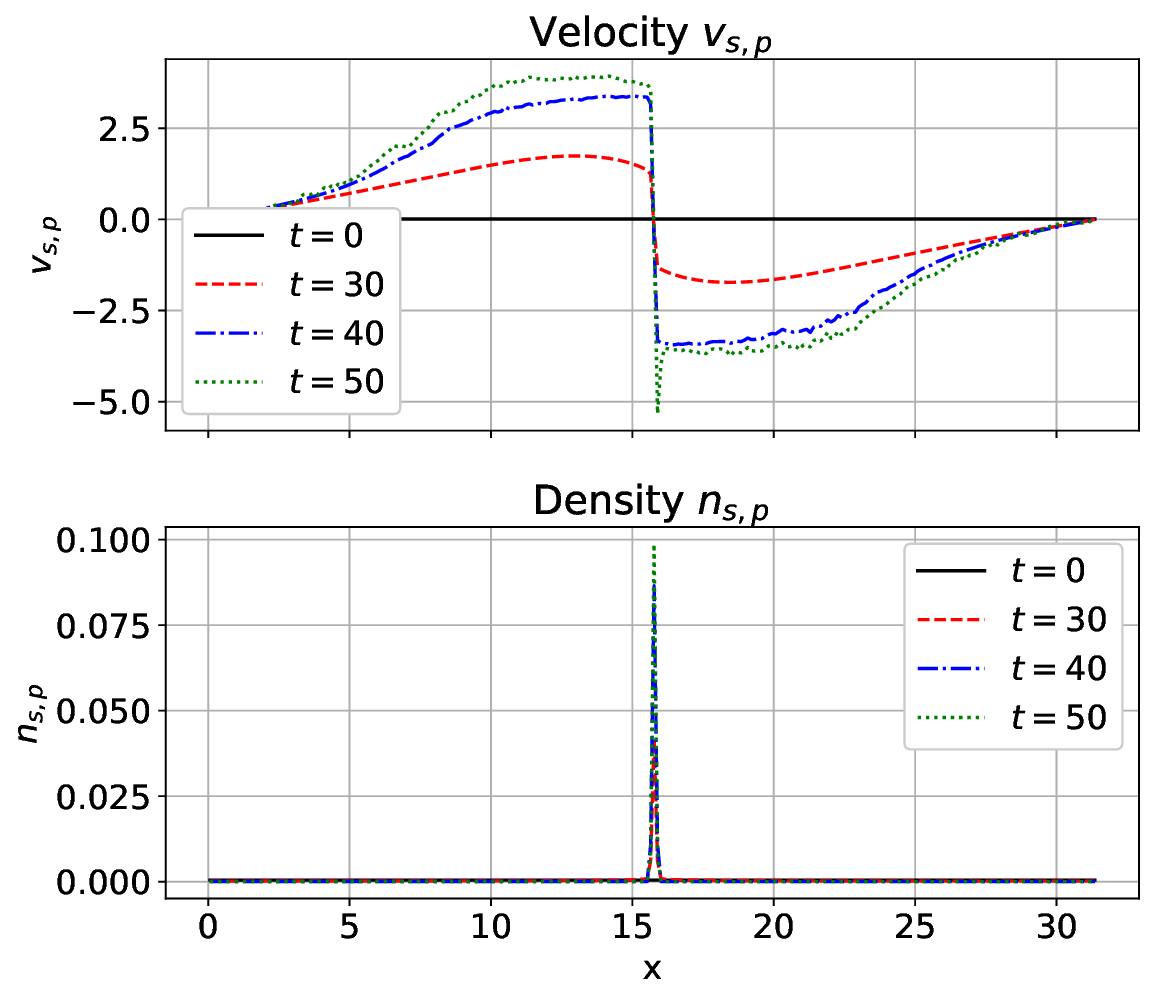}
    
    \caption{Two-stream instability experiment in 1D with fluid discretization and without transformation to particles. Illustrated is the loss of one-to-one mapping of $v_{s,p}$ and the growth of the singularity of $n_{s,p}$ for one of the problematic fluid nodes.}
    \label{fig:tis}
\end{figure}

Next, we test the adaptive fluid to particle transformation from Section \ref{sec:fluid_pic_transf}. We keep the numerical parameters related to the fluid quantities. For the transformation to particles, we use one particle per cell per fluid node (ppcpfn) thus keeping the total number of degrees of freedom before and after the transformation constant. In the criterion \eqref{eq:criterion}, we choose $C=3$. The time step is set to $\Delta t = 0.05 \, \Delta x / v_{max}$, where $v_{max}$ is the maximum velocity for all fluids and particles. Figure \ref{fig:tis_with_transf} plots the fluid component, the particle component and the total distribution in the first, second, and third rows respectively. The fourth row shows the reference pure particle solution with the same number of degrees of freedom obtained by transforming all fluid nodes at time $t=0$ into particles. We can see that up to time $t \approx 30$ the whole distribution is discretized entirely with fluids. As the simulation evolves, a portion of the distribution develops a whirl activating the criterion \eqref{eq:criterion} with subsequent transformation of fluid nodes into particles within that portion. 

\begin{figure}[H]
\centering
\begin{subfigure}{0.24\textwidth}
    \includegraphics[width=\textwidth]{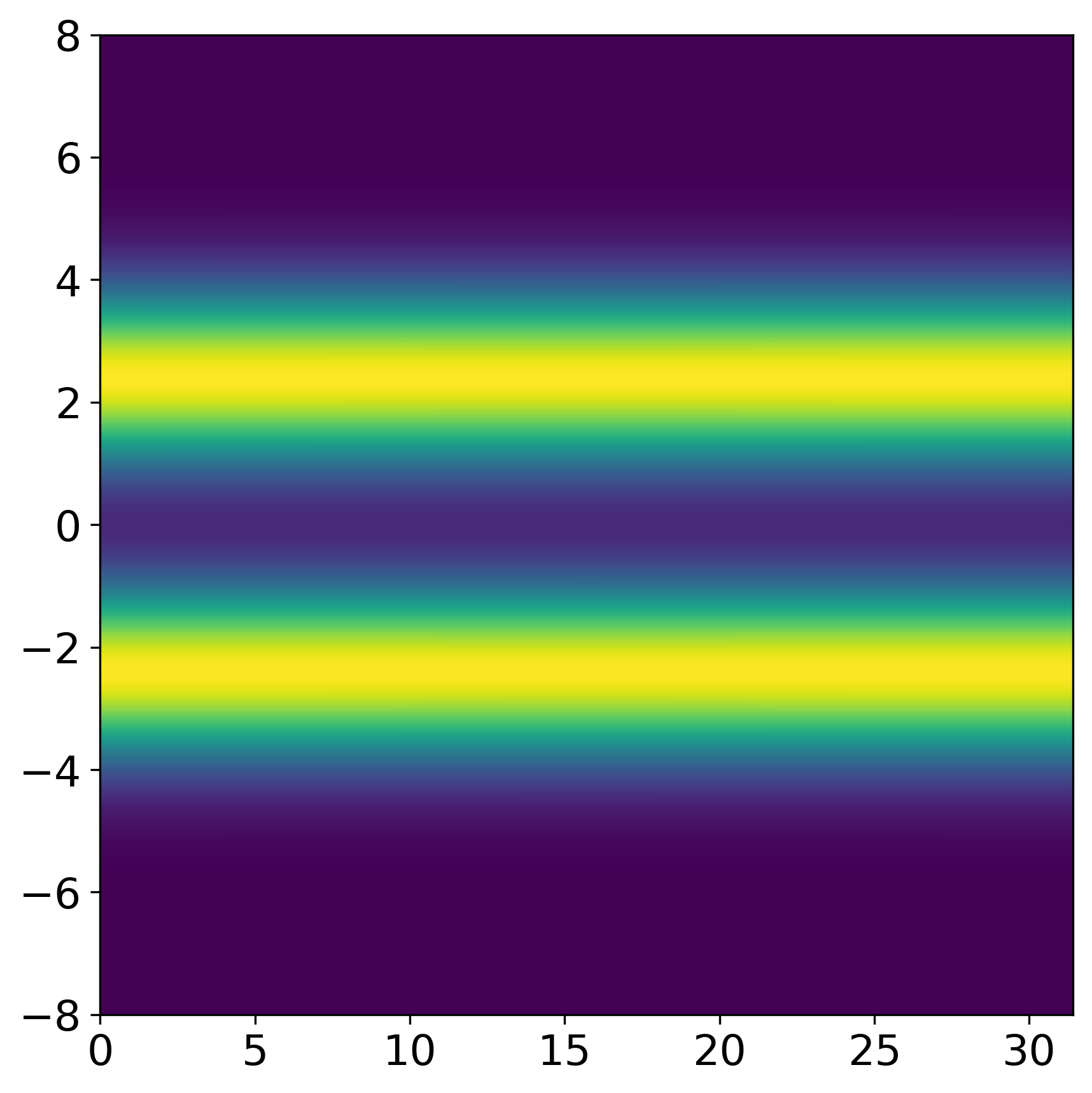}
    \caption{$f_f$, $t=0$}
    \label{fig:s200_0}
\end{subfigure}
\hfill
\begin{subfigure}{0.24\textwidth}
    \includegraphics[width=\textwidth]{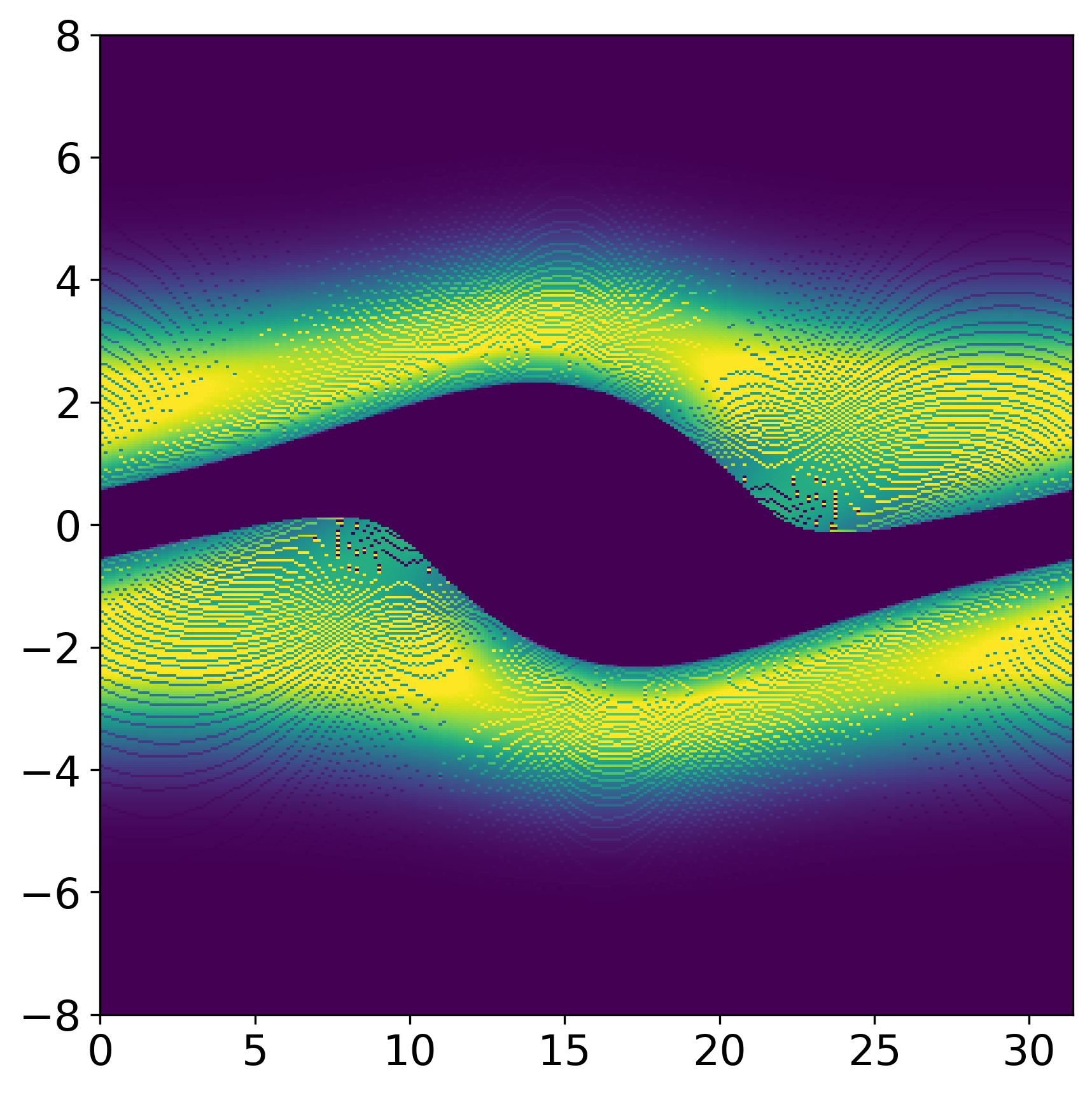}
    \caption{$f_f$, $t=30$}
    \label{fig:s200_28}
\end{subfigure}
\hfill
\begin{subfigure}{0.24\textwidth}
    \includegraphics[width=\textwidth]{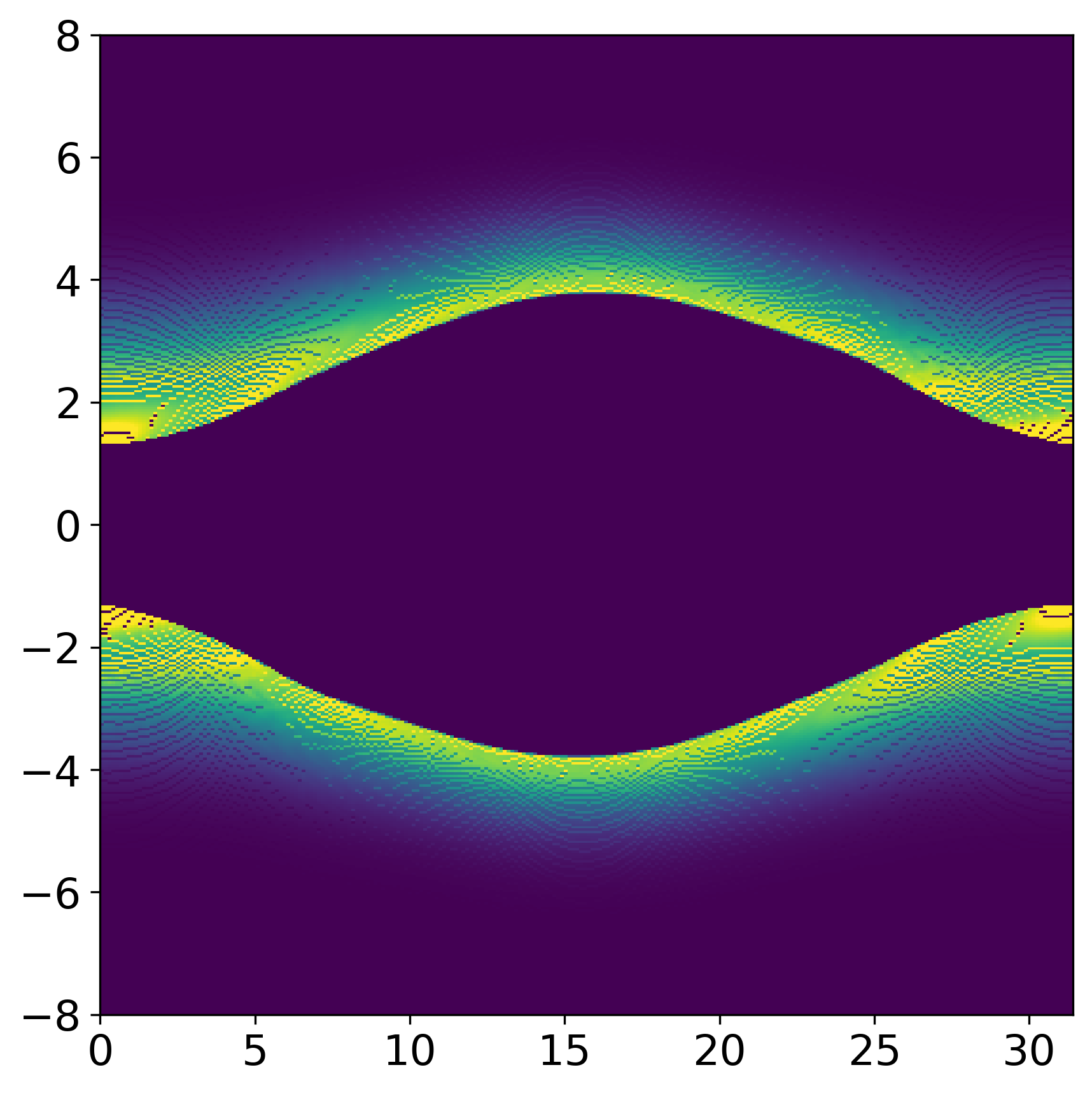}
    \caption{$f_f$, $t=40$}
    \label{fig:s200_30}
\end{subfigure}
\hfill
\begin{subfigure}{0.24\textwidth}
    \includegraphics[width=\textwidth]{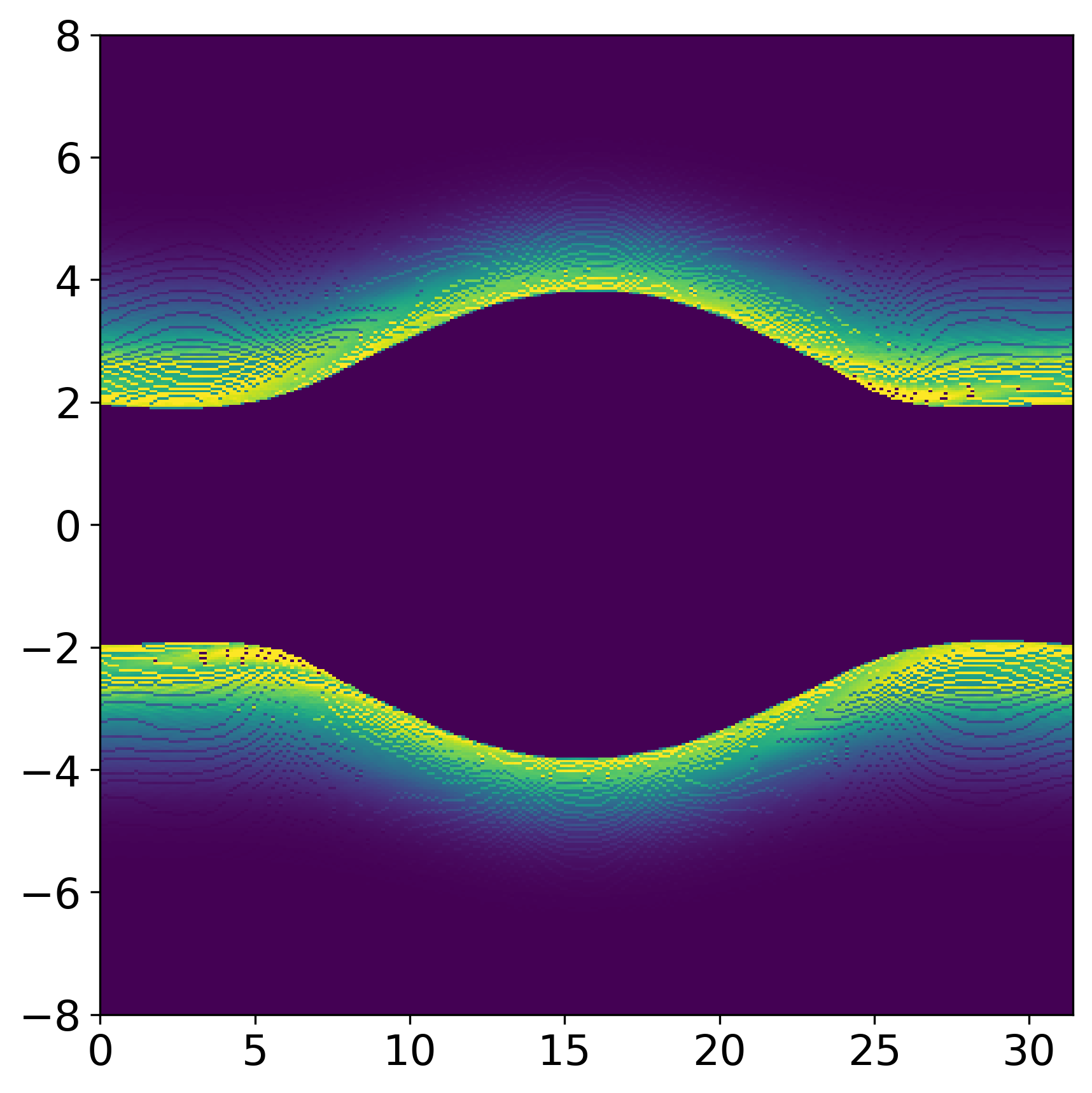}
    \caption{$f_f$, $t=50$}
    \label{fig:s200_34}
\end{subfigure}
\\
\vspace{0.25cm}
+
\vspace{0.25cm}
\\
\begin{subfigure}{0.24\textwidth}
    \includegraphics[width=\textwidth]{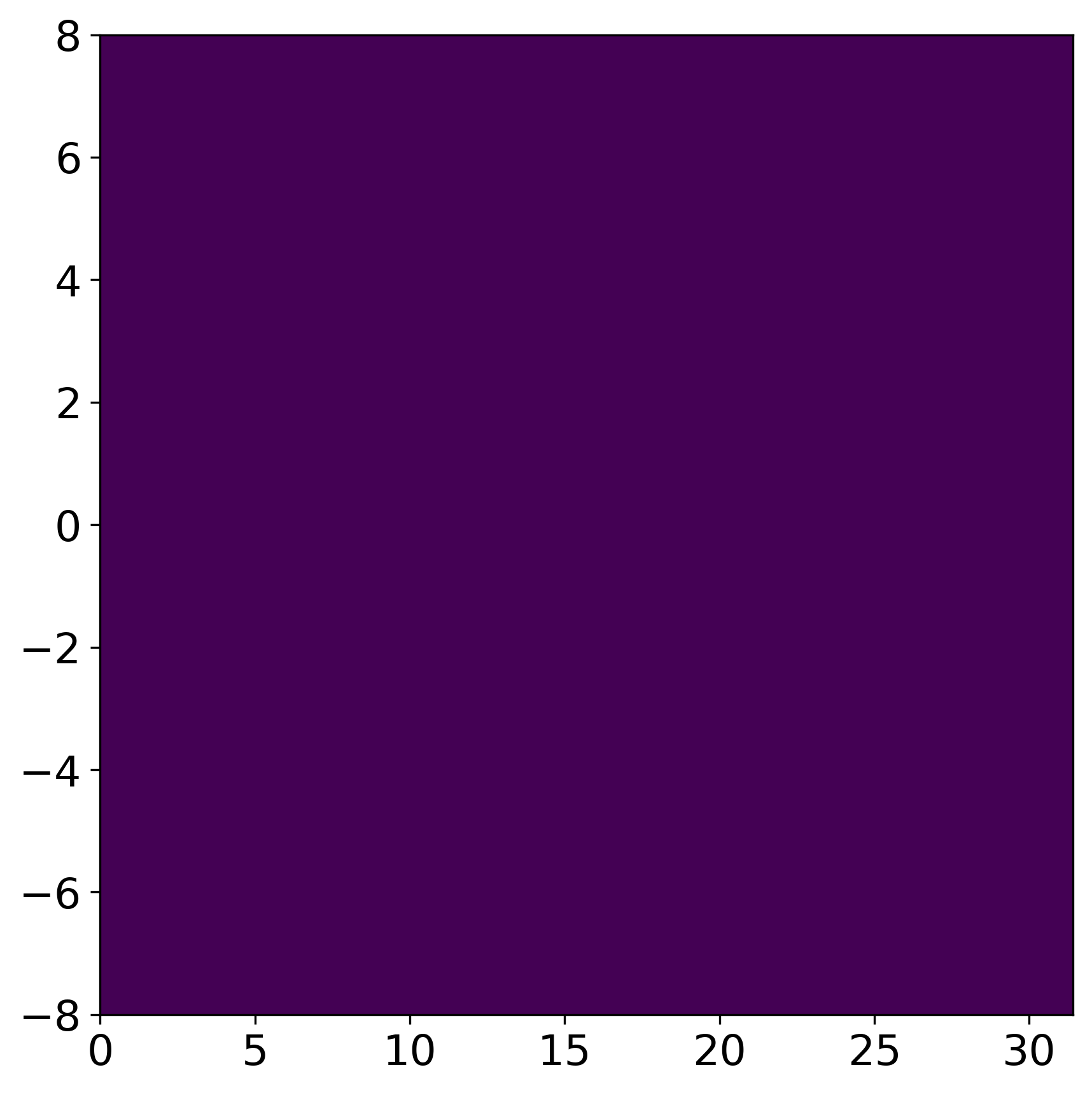}
    \caption{$f_p$, $t=0$}
    \label{fig:s1000_0}
\end{subfigure}
\hfill
\begin{subfigure}{0.24\textwidth}
    \includegraphics[width=\textwidth]{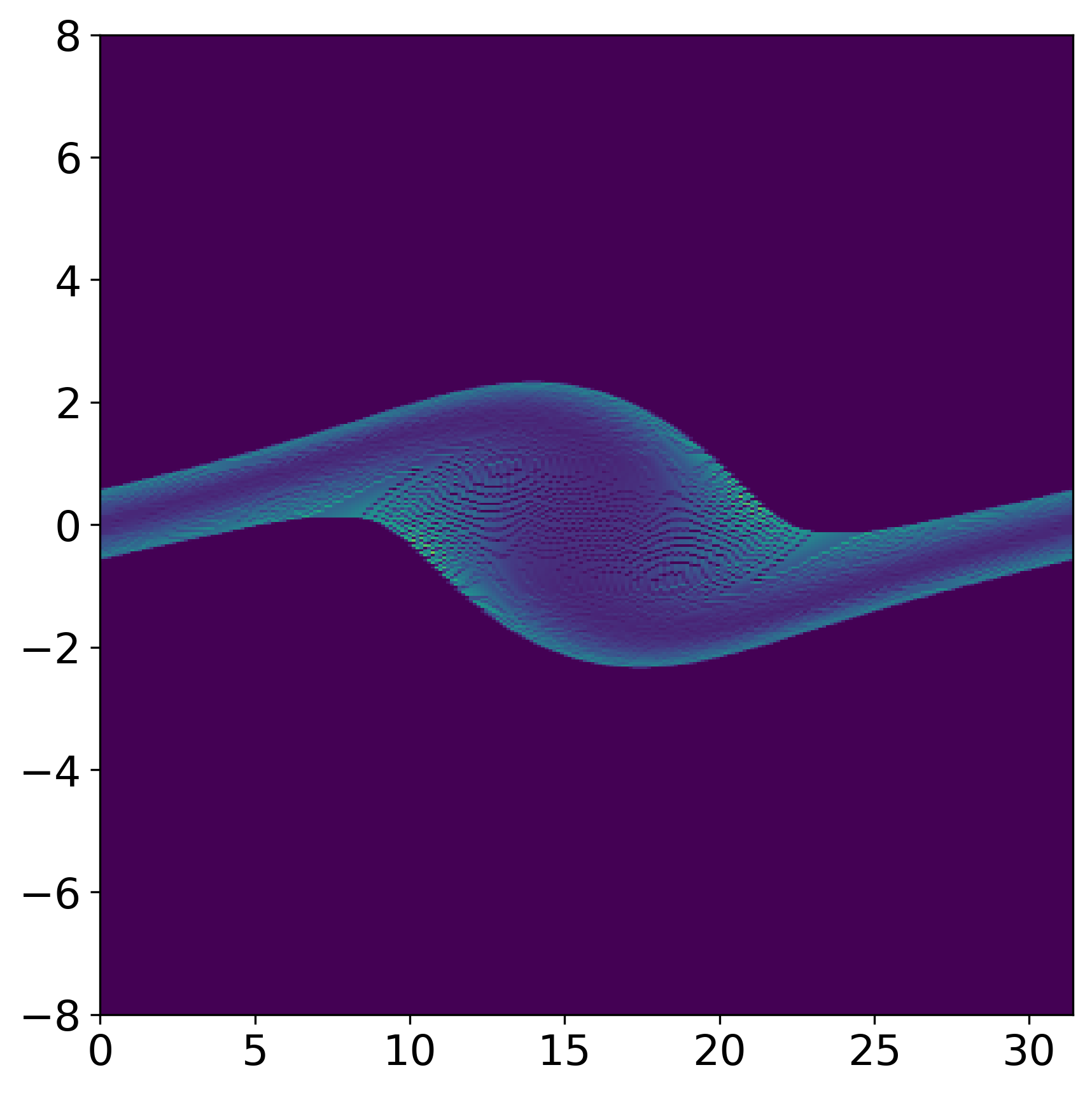}
    \caption{$f_p$, $t=30$}
    \label{fig:s1000_28}
\end{subfigure}
\hfill
\begin{subfigure}{0.24\textwidth}
    \includegraphics[width=\textwidth]{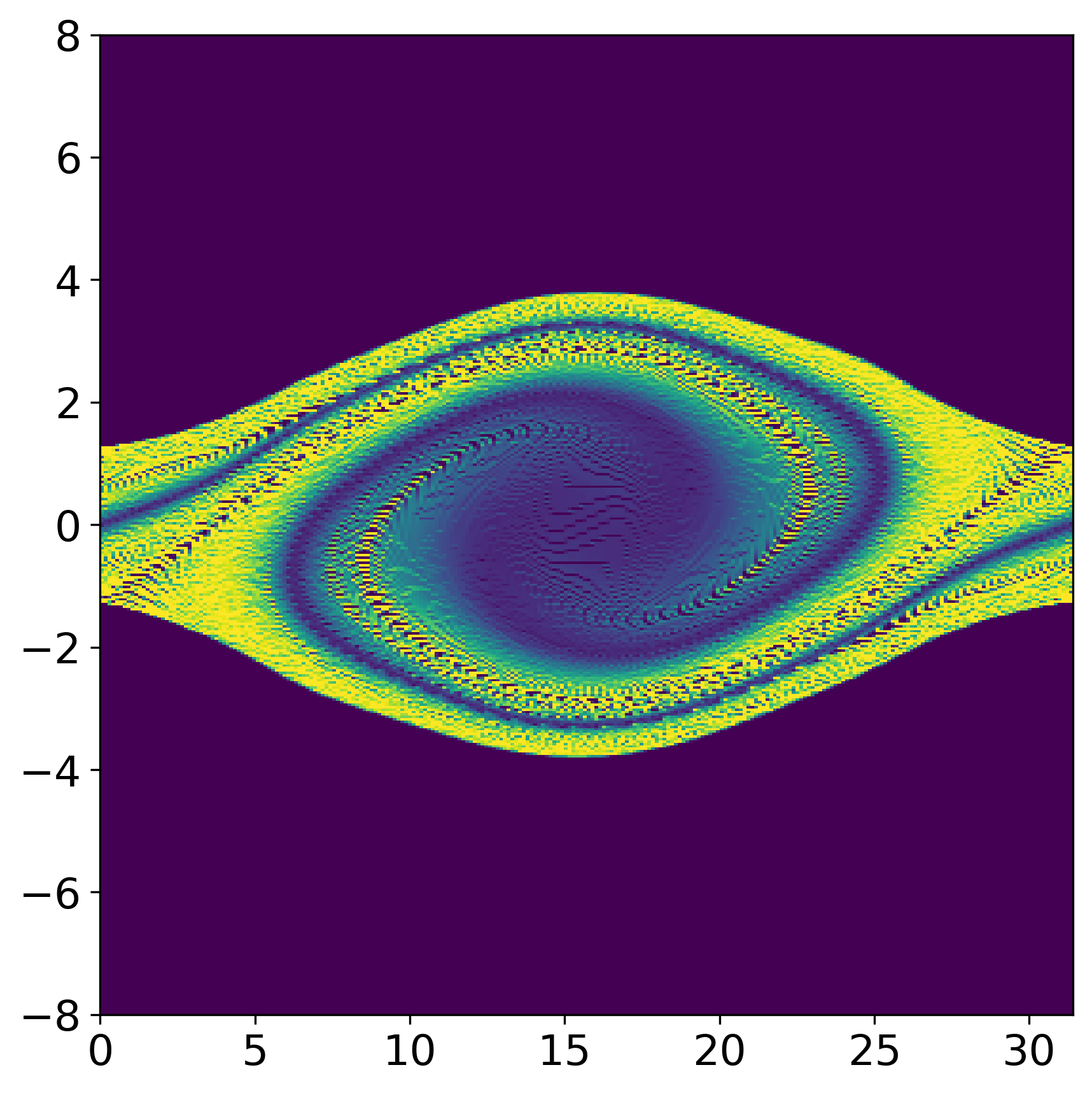}
    \caption{$f_p$, $t=40$}
    \label{fig:rho_six}
\end{subfigure}
\hfill
\begin{subfigure}{0.24\textwidth}
    \includegraphics[width=\textwidth]{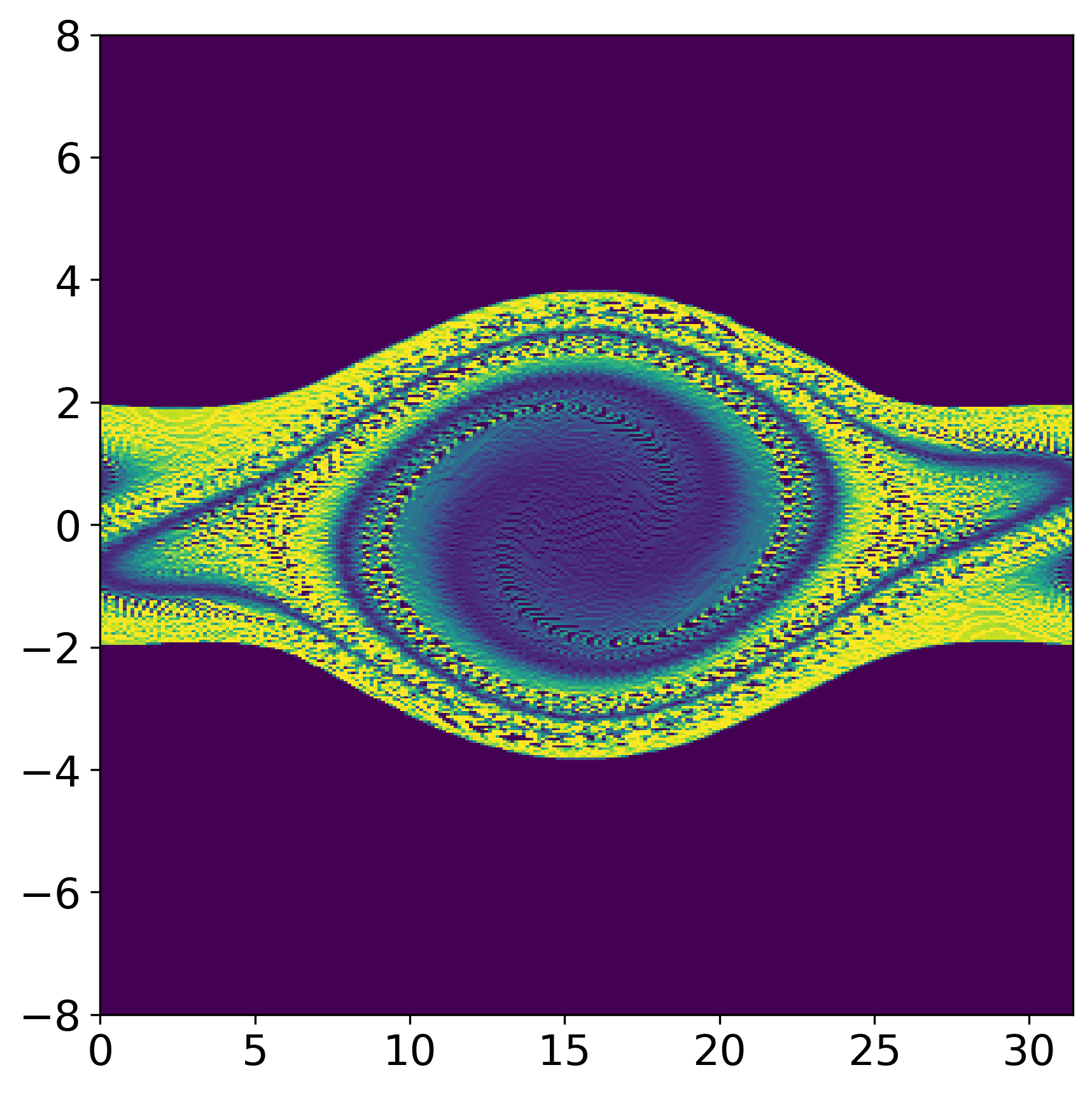}
    \caption{$f_p$, $t=50$}
    \label{fig:rho_six}
\end{subfigure}
\\
\vspace{0.25cm}
=
\vspace{0.25cm}
\\
\begin{subfigure}{0.24\textwidth}
    \includegraphics[width=\textwidth]{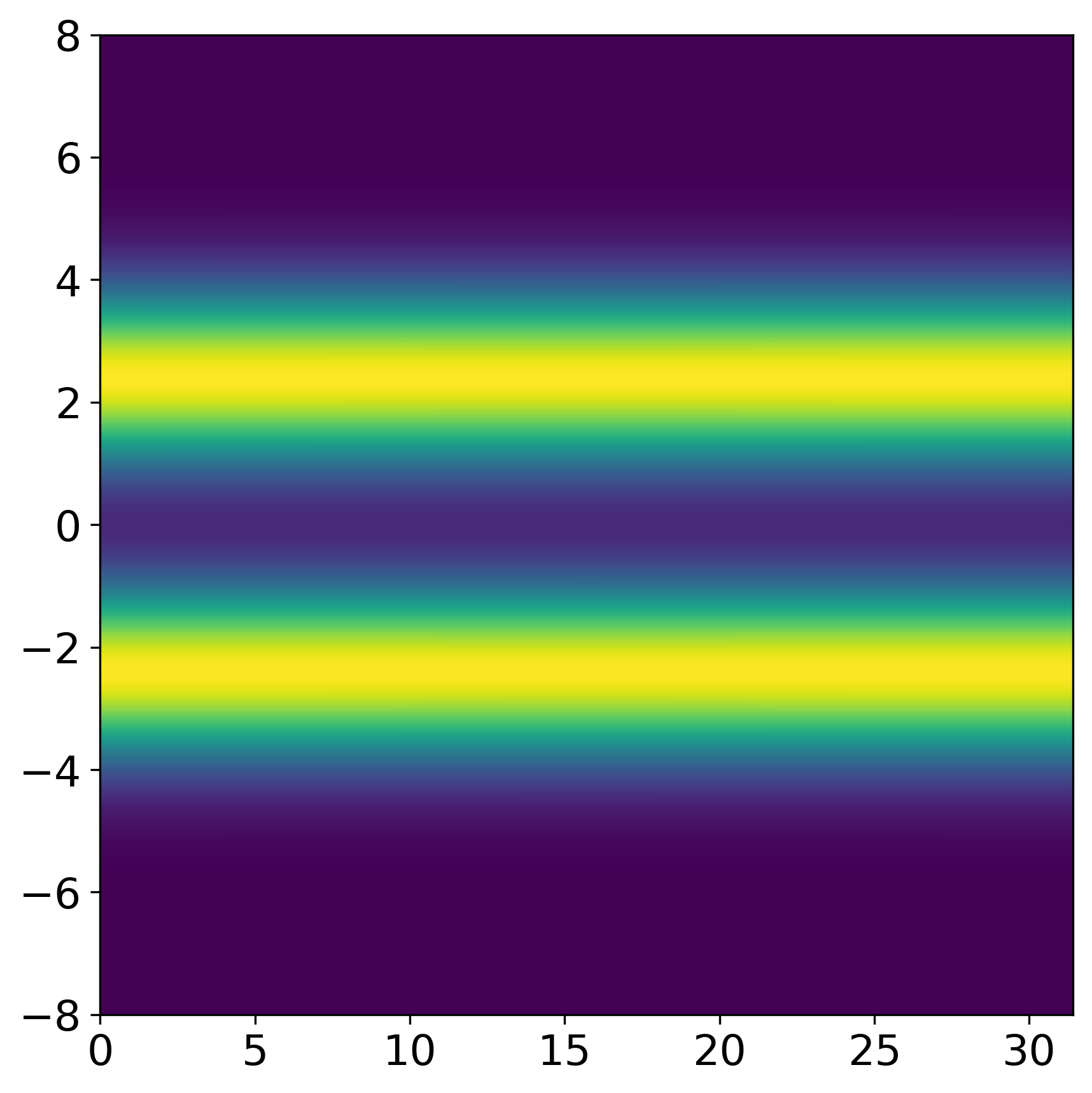}
    \caption{$f_t$, $t=0$}
    \label{fig:s1000_0}
\end{subfigure}
\hfill
\begin{subfigure}{0.24\textwidth}
    \includegraphics[width=\textwidth]{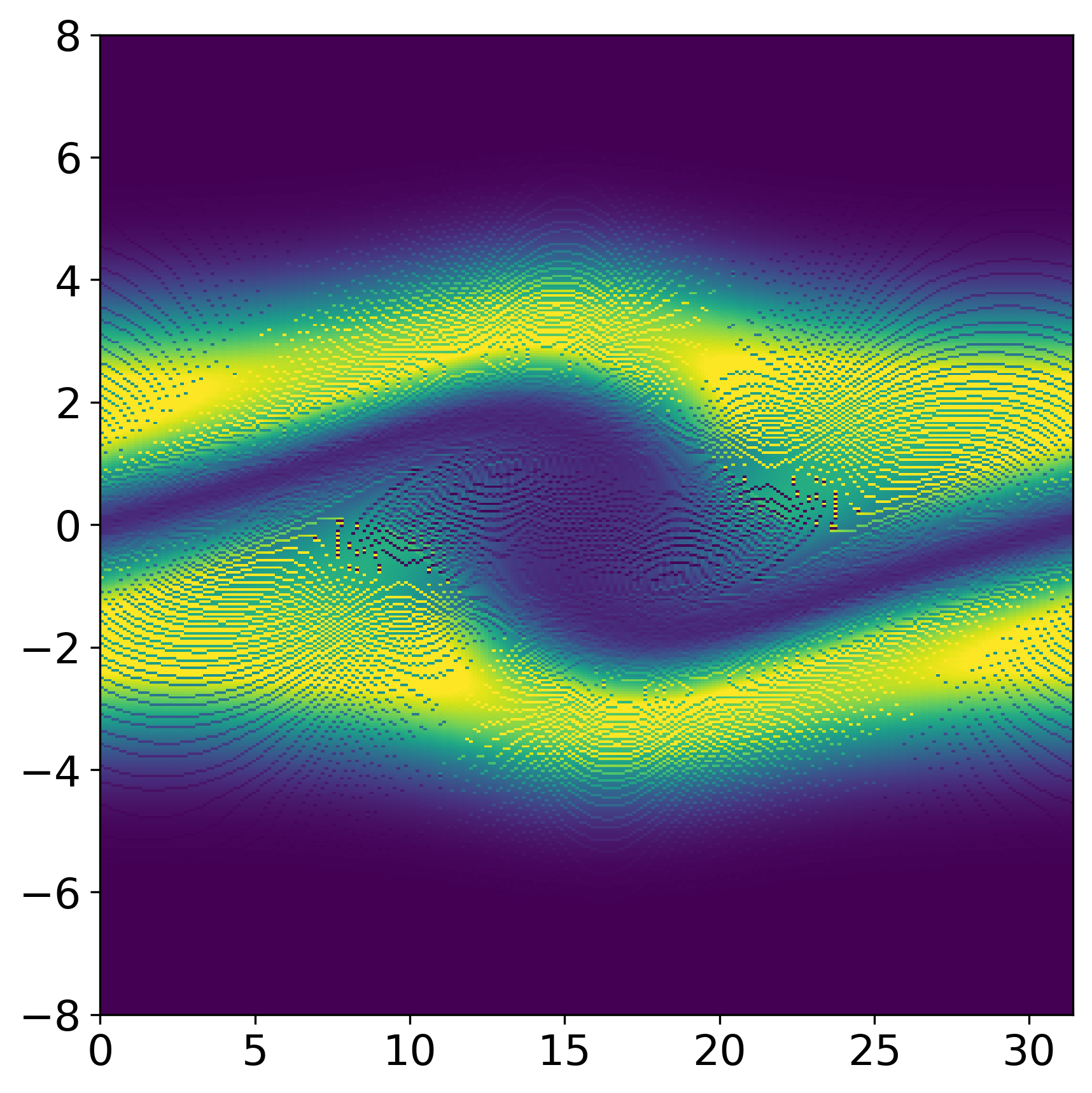}
    \caption{$f_t$, $t=30$}
    \label{fig:s1000_28}
\end{subfigure}
\hfill
\begin{subfigure}{0.24\textwidth}
    \includegraphics[width=\textwidth]{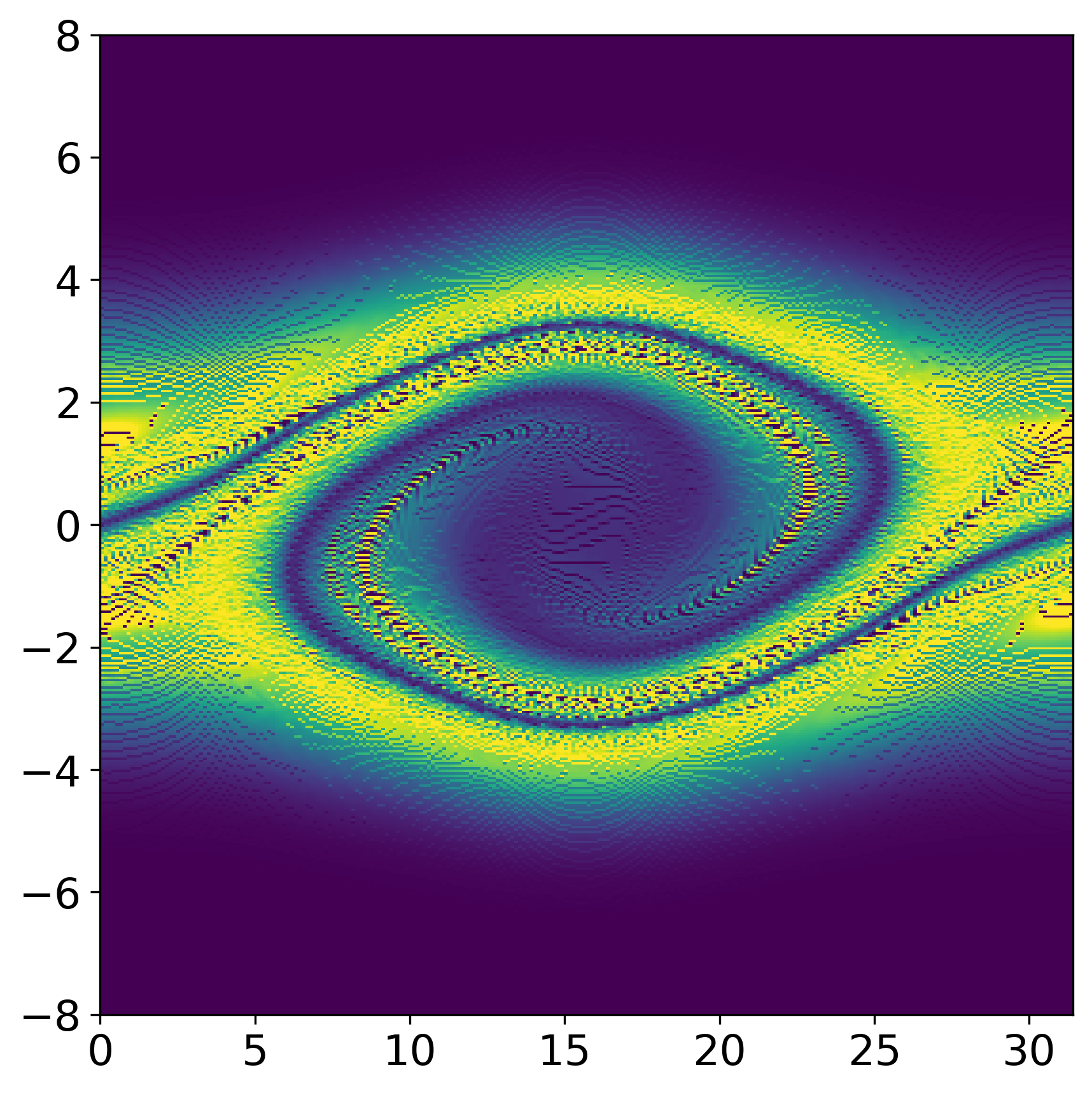}
    \caption{$f_t$, $t=40$}
    \label{fig:rho_six}
\end{subfigure}
\hfill
\begin{subfigure}{0.24\textwidth}
    \includegraphics[width=\textwidth]{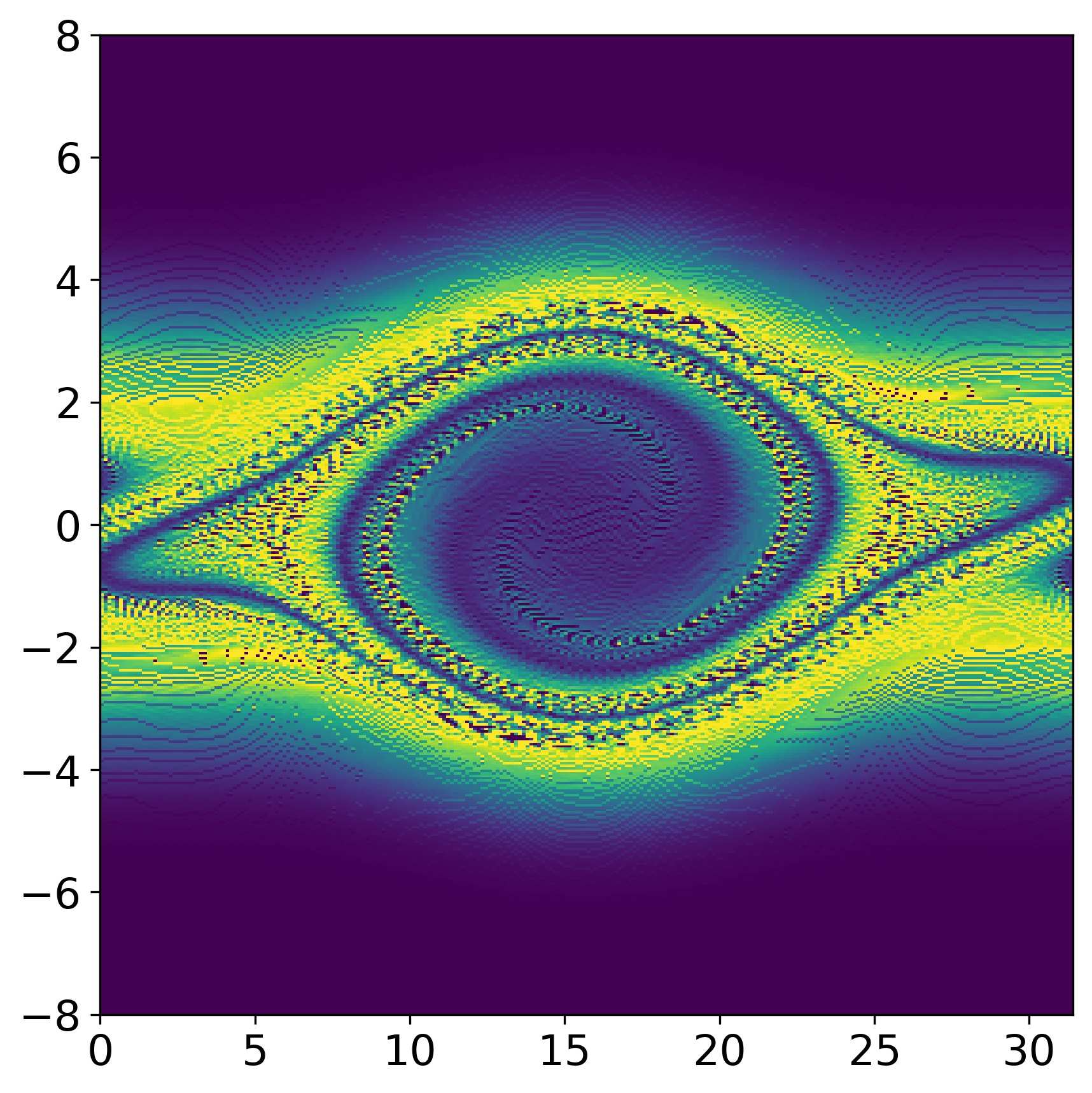}
    \caption{$f_t$, $t=50$}
    \label{fig:rho_six}
\end{subfigure}
\\
\begin{subfigure}{0.24\textwidth}
    \includegraphics[width=\textwidth]{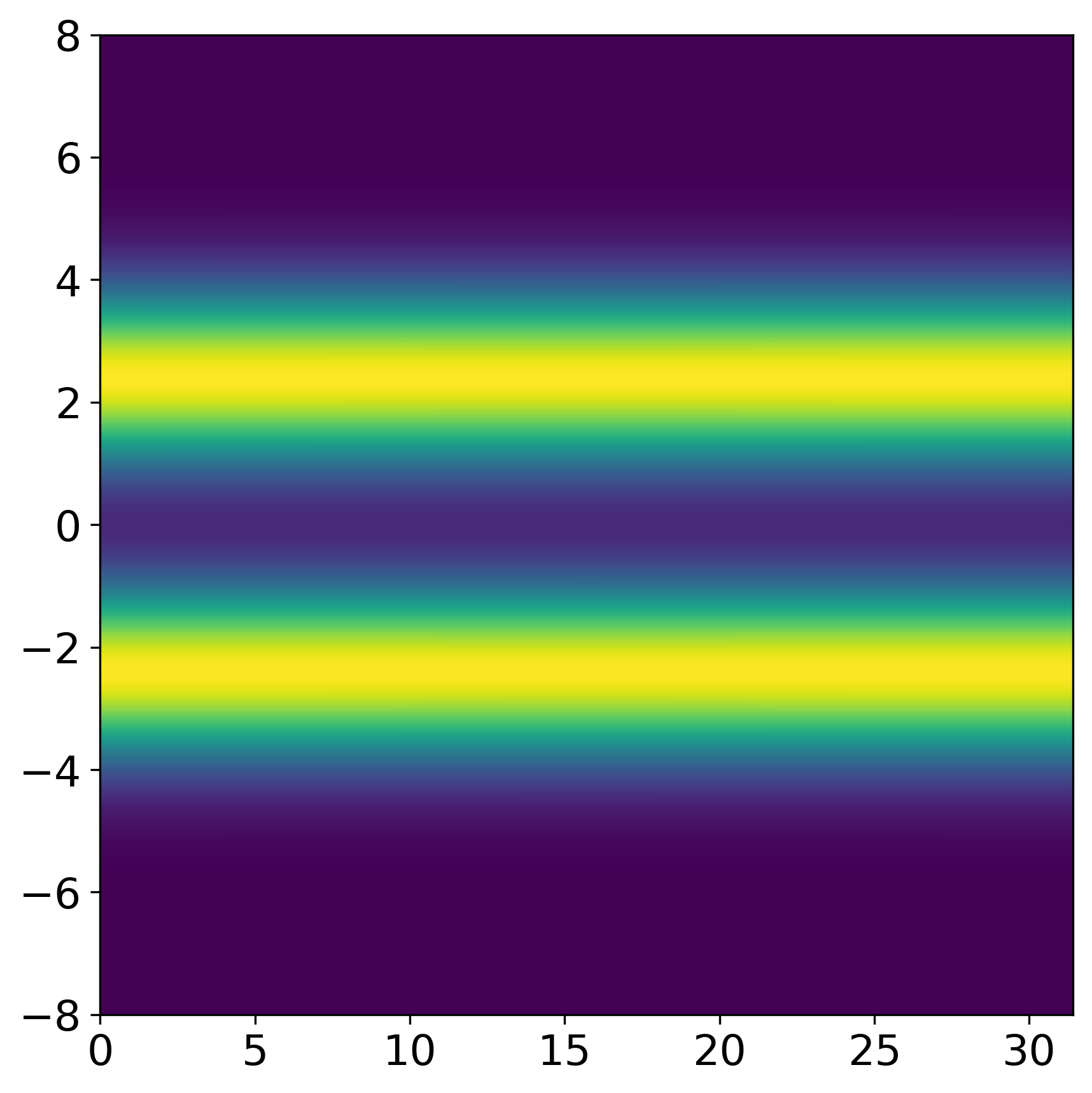}
    \caption{reference, $t=0$}
    \label{fig:s1000_0}
\end{subfigure}
\hfill
\begin{subfigure}{0.24\textwidth}
    \includegraphics[width=\textwidth]{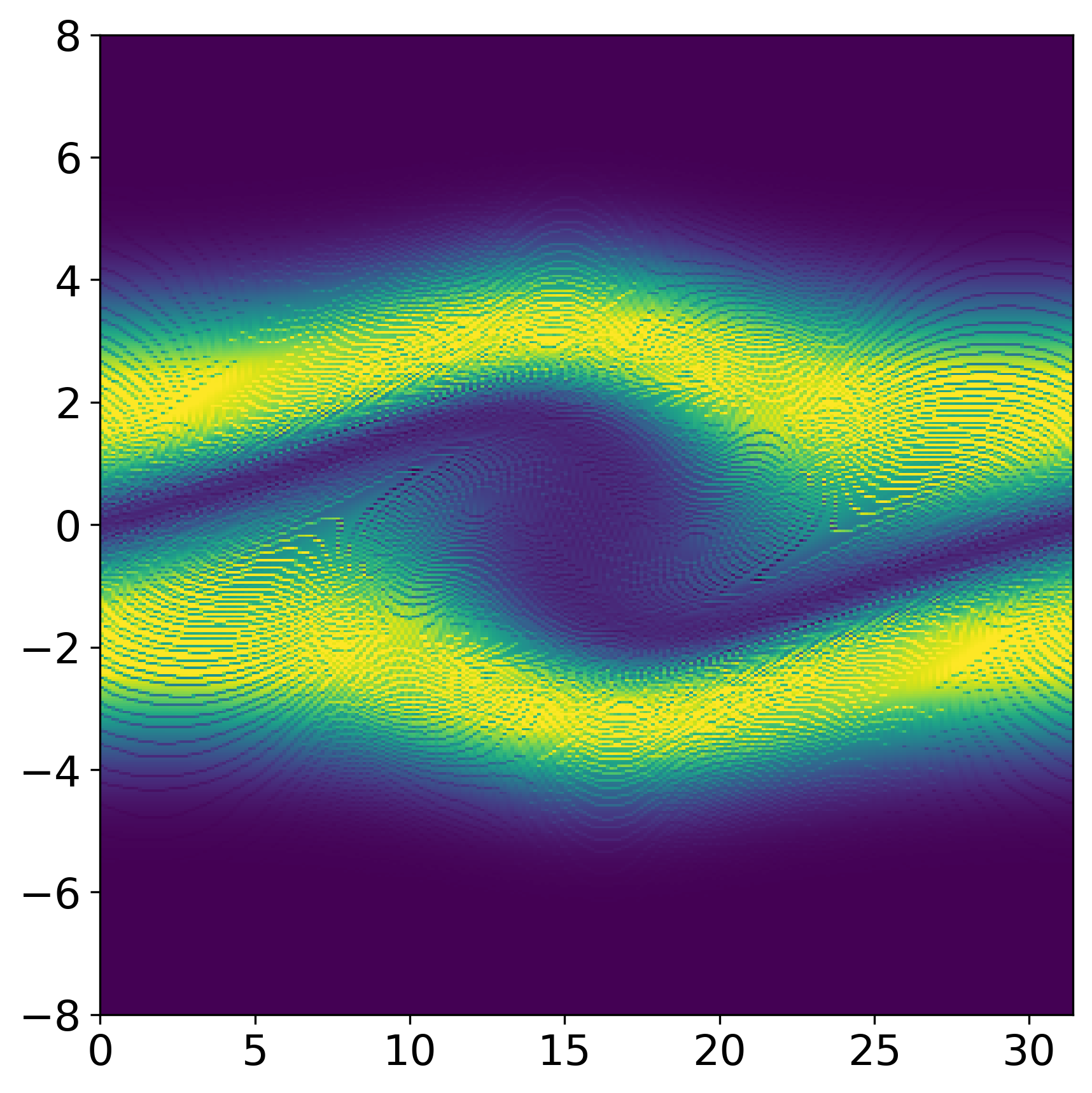}
    \caption{reference, $t=30$}
    \label{fig:s1000_28}
\end{subfigure}
\hfill
\begin{subfigure}{0.24\textwidth}
    \includegraphics[width=\textwidth]{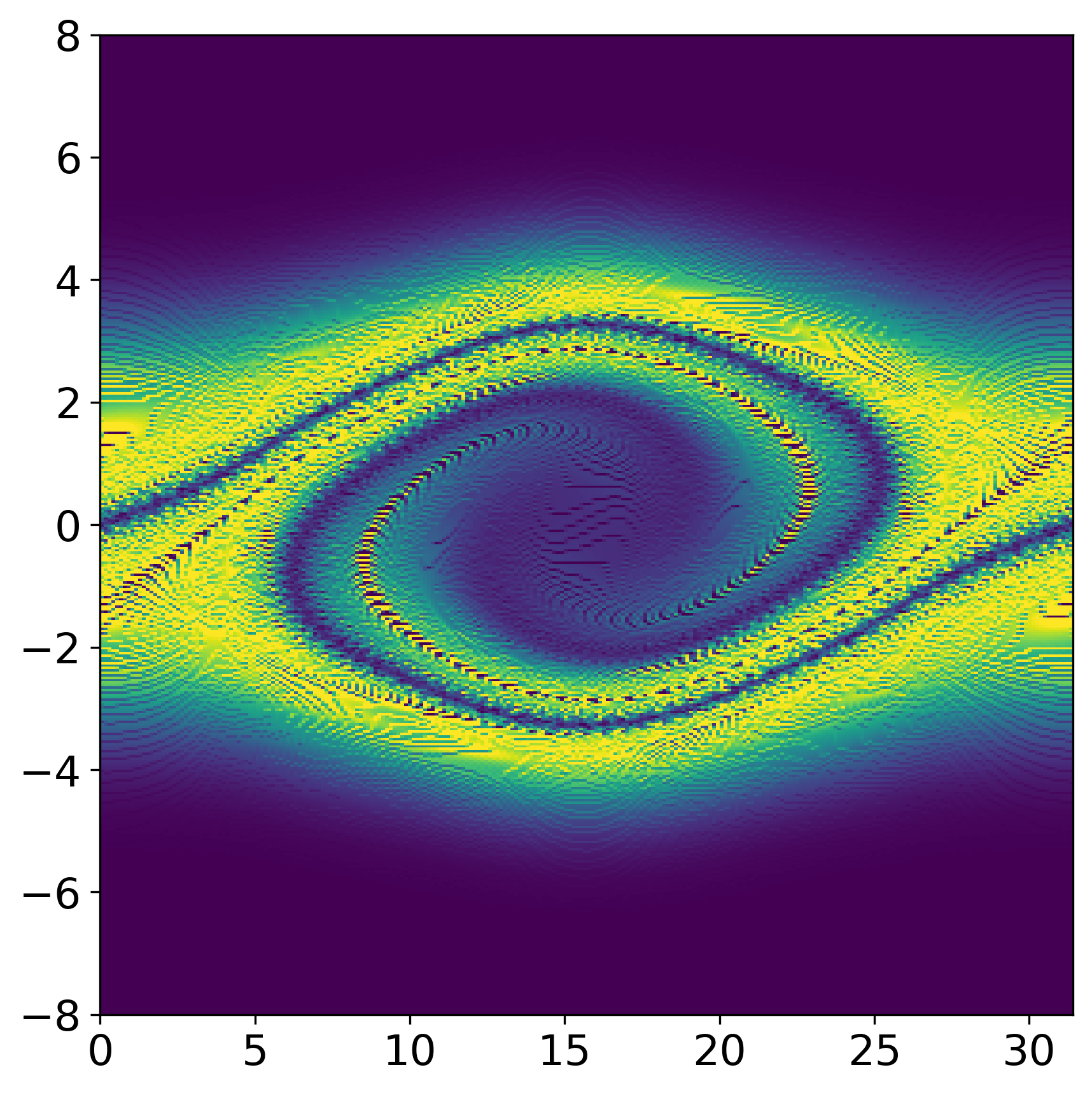}
    \caption{reference, $t=40$}
    \label{fig:rho_six}
\end{subfigure}
\hfill
\begin{subfigure}{0.24\textwidth}
    \includegraphics[width=\textwidth]{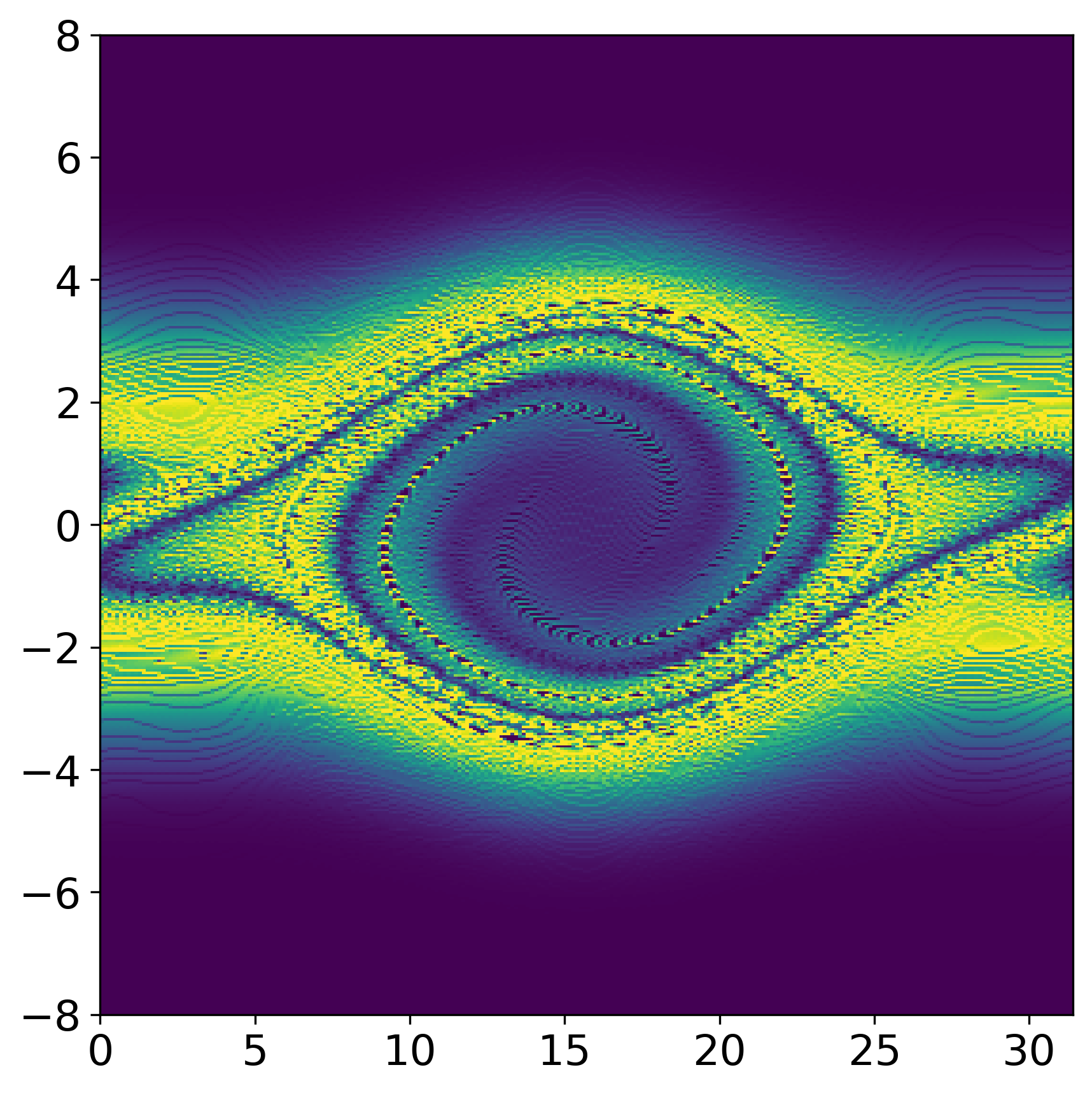}
    \caption{reference, $t=50$}
    \label{fig:rho_six}
\end{subfigure}

\caption{Two-stream instability experiment in 1D with fluid to particle transformation. The initial conditions are approximated by $L=400$ slices. The figures in the top row show the fluid component ($f_f$). The figures in the second row show the particle component ($f_p$). The figures in the third row show the total distribution, i.e. the sum $f_t = f_f + f_p$. The figures at the bottom show the reference pure particle simulation with the same number of degrees of freedom.}
\label{fig:tis_with_transf}
\end{figure}

Figure \ref{fig:tis_fluid_ratio} shows the ratio of fluid nodes at time $t$ to the total number of nodes at the beginning of the simulation. We can see that by the end of the simulation $65.75\%$ of the fluid nodes remain and correspondingly $34.25\%$ of the nodes were transformed into particles. 

Figure \ref{fig:tis_with_transf_energy} shows excellent agreement of the electric field energy of the hybrid simulation and the reference particle simulation.

\begin{figure}[H]
    \centering
    \includegraphics[width=0.5\textwidth]{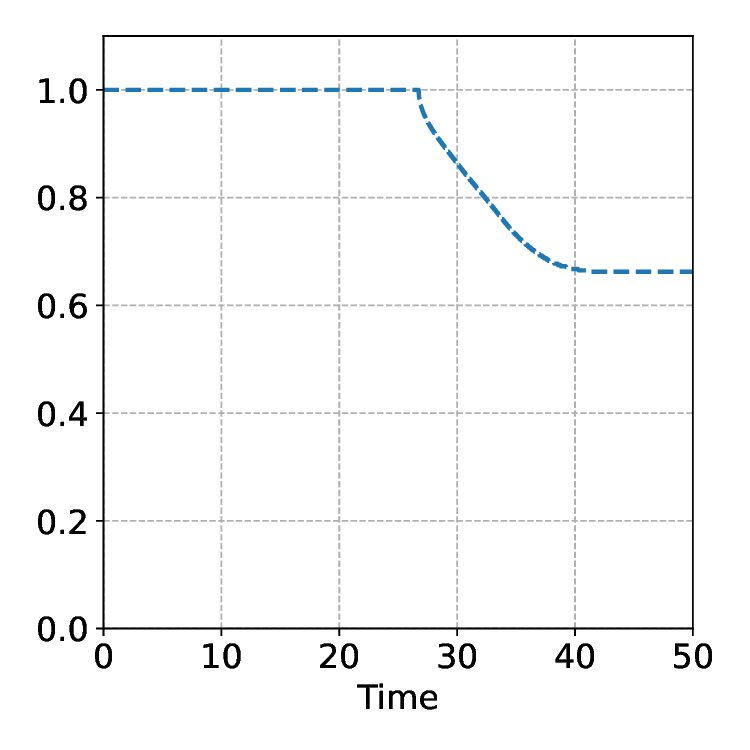}
    
    \caption{Two-stream instability experiment in 1D with fluid to particle transformation. The vertical axis plots the ratio of the number of fluid nodes at time $t$ to the number of fluid nodes at time $t=0$.}
    \label{fig:tis_fluid_ratio}
\end{figure}

\begin{figure}[H]
    \centering
    \includegraphics[width=\textwidth]{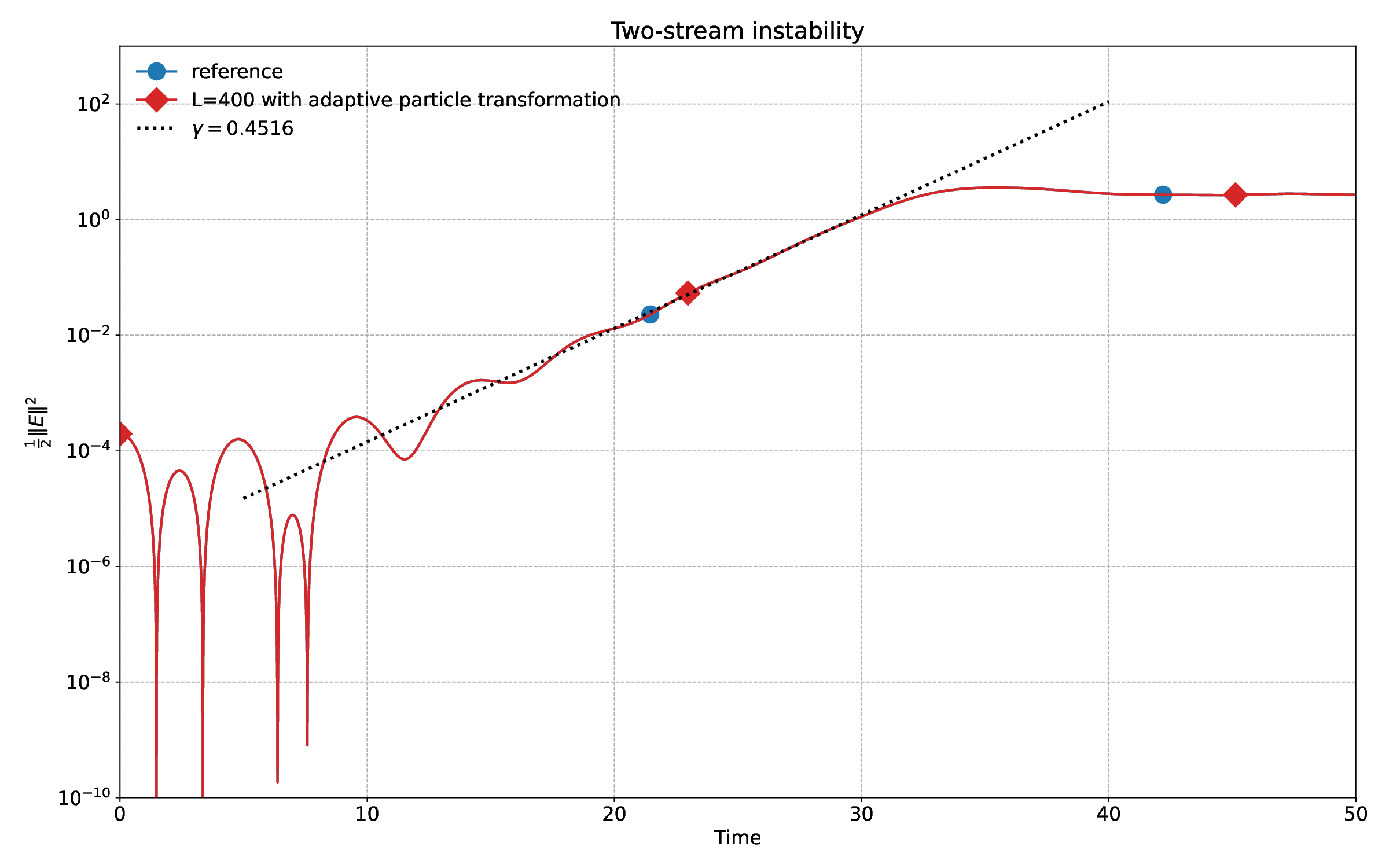}
    
    \caption{Two-stream instability experiment in 1D  with fluid to particle transformation. The reference solution is given by the pure particle simulation. The theoretical growth rate from the linear theory is given by the dashed line.}
    \label{fig:tis_with_transf_energy}
\end{figure}

\subsection{Strong Landau damping}
We study the electrostatic example of strong Landau damping. The setup is similar to the weak Landau damping in Section \ref{sec:weak_landau}, in particular the initial distribution is given by \eqref{eq:weak_landau_f}, the electric field by \eqref{eq:weak_landau_e}, and the domain is the same. For the strong Landau damping experiment we use $k = \alpha=0.5$. 

We use the adaptive fluid to particle transformation from Section \ref{sec:fluid_pic_transf} with the coefficient $C=\{1.5,2,3\}$ in \eqref{eq:criterion} and consider one particle per cell per fluid node (ppcpfn) . The initial conditions are approximated by fluid discretization with $L=400$ slices. For the finite volume grid we use $2^7$ cells. The time step is set to $\Delta t = 0.05 \, \Delta x / v_{max}$, where $v_{max}$ is the maximum velocity for all fluids and particles. For comparison we use a particle solution with the same number of degrees of freedom as a reference.

Figure \ref{fig:weak_landau_damping} shows that the solution agrees well with the reference solution. The damping and growth rates for the electric energy agrees with the rates given in \cite{gempic}. 
 
Figure \ref{fig:ld_fluid_ratio} shows the ratio of fluid nodes at time $t$ to the total number of nodes at the beginning of the simulation. By the time $t \approx 4.0$ all fluid nodes are transformed completely into particles. Larger values of $C$ delay the transformation.

\begin{figure}[H]
    \centering
    \includegraphics[width=\textwidth]{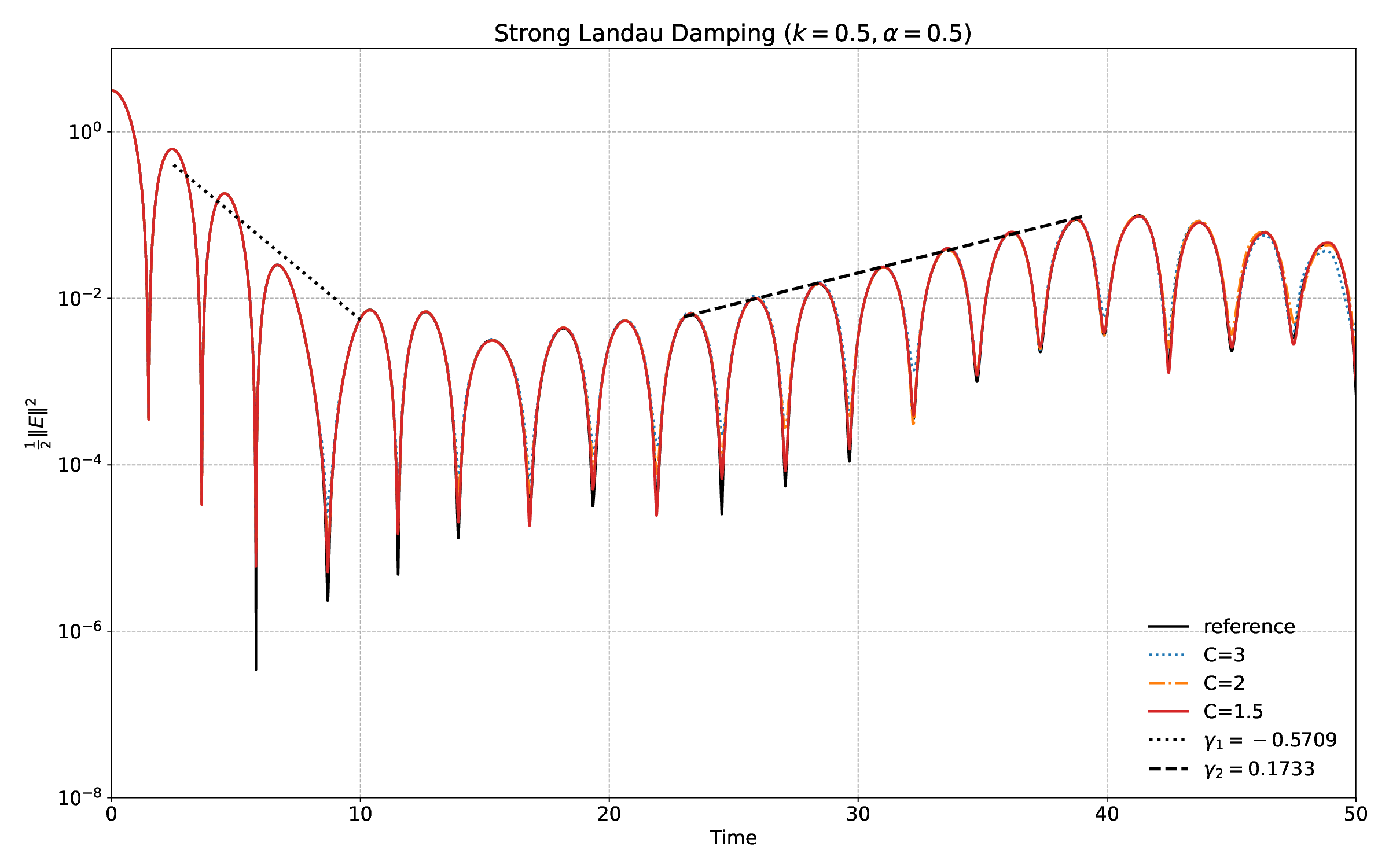}
    
    \caption{Strong Landau damping experiment in 1D with fluid to particle transformation. The damping and the growth rates are taken from \cite{gempic}.}
    \label{fig:weak_landau_damping}
\end{figure}

\begin{figure}[H]
    \centering
    \includegraphics[width=0.5\textwidth]{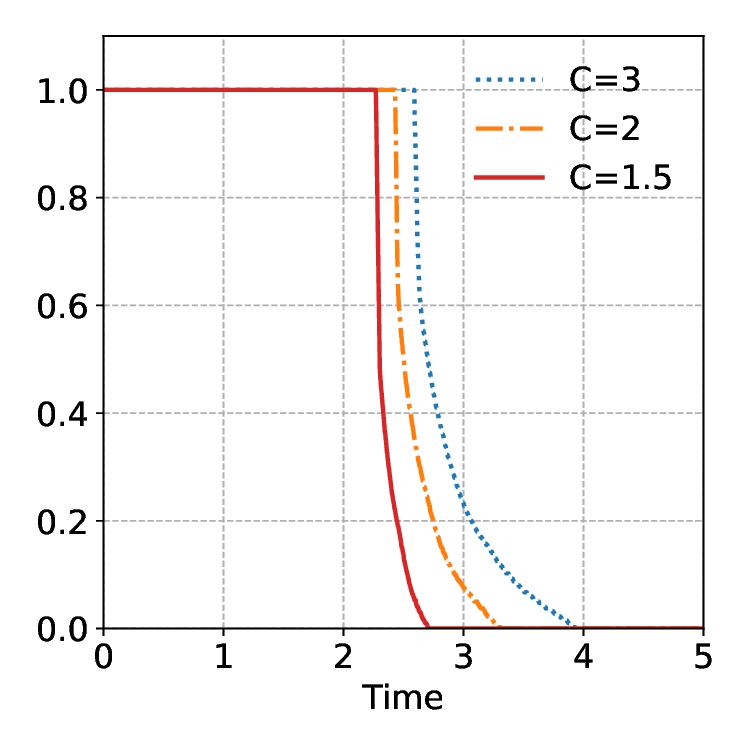}
    
    \caption{Strong Landau damping experiment in 1D with fluid to particle transformation. The vertical axis plots the ratio of the number of fluid nodes at time $t$ to the number of fluid nodes at time $t=0$.}
    \label{fig:ld_fluid_ratio}
\end{figure}

\subsection{Relativistic two-stream instability}
In this section, we consider the relativistic version of the two-stream instability. In particular, the initial condition is given by the two hump Maxwell-Jüttner distribution
\begin{align}\notag
f_0(u,x) = (1+0.001 \cos(2 \pi x/|\Omega_x|)) \frac{1}{K} & \left( \exp \left(-\alpha \left(\gamma(u) \gamma_0 - \frac{\gamma_0 v_0 u}{mc^2} - 1 \right) \right) \right. \\ 
& \hspace{0.5cm}+ \left. \exp \left(-\alpha \left(\gamma(u) \gamma_0 + \frac{\gamma_0 v_0 u}{mc^2} - 1 \right) \right) \right)
\end{align}
where $K=0.829525$ is the normalization constant, $\alpha=mc^2/(k_BT) = 102.1997$, $v_0 = 0.8 c$, $\gamma_0 = \sqrt{1-v_0^2/c^2}$. The domain $\Omega_x=[0,18]$ and $\Omega_u=[-10,10]$. The initial electric field is chosen to satisfy the Gauss's law.

We vary the number of slices $L=\{10,50,100,200, 400\}$ and the finite volume grid resolution $\{2^4,2^5, 2^6, 2^7\}$. We set the coefficient $C=3$ in the transformation criterion \eqref{eq:criterion} and consider one particle per cell per fluid node (ppcpfn) thereby keeping the total degrees of freedom before and after the transformation the same. Figure \ref{fig:rel_tis_total} compares the results for different number of slices and resolutions. The number of slices and the number of grid cells directly affect the quality of the results with the greater number revealing better results. 

The case with $L=50$ slices and $2^5$ cells already captures the major features. Increasing the number of slices allows to capture finer features. Thus, in practice first a simulation with a low number of uniformly distributed slices can be used to identify the region with complex dynamics. Then in the second run the number of slices can be increased only in the target region thereby maintaining balance between performance and accuracy.

\begin{figure}[H]
\centering
\begin{subfigure}{0.24\textwidth}
    \includegraphics[width=\textwidth]{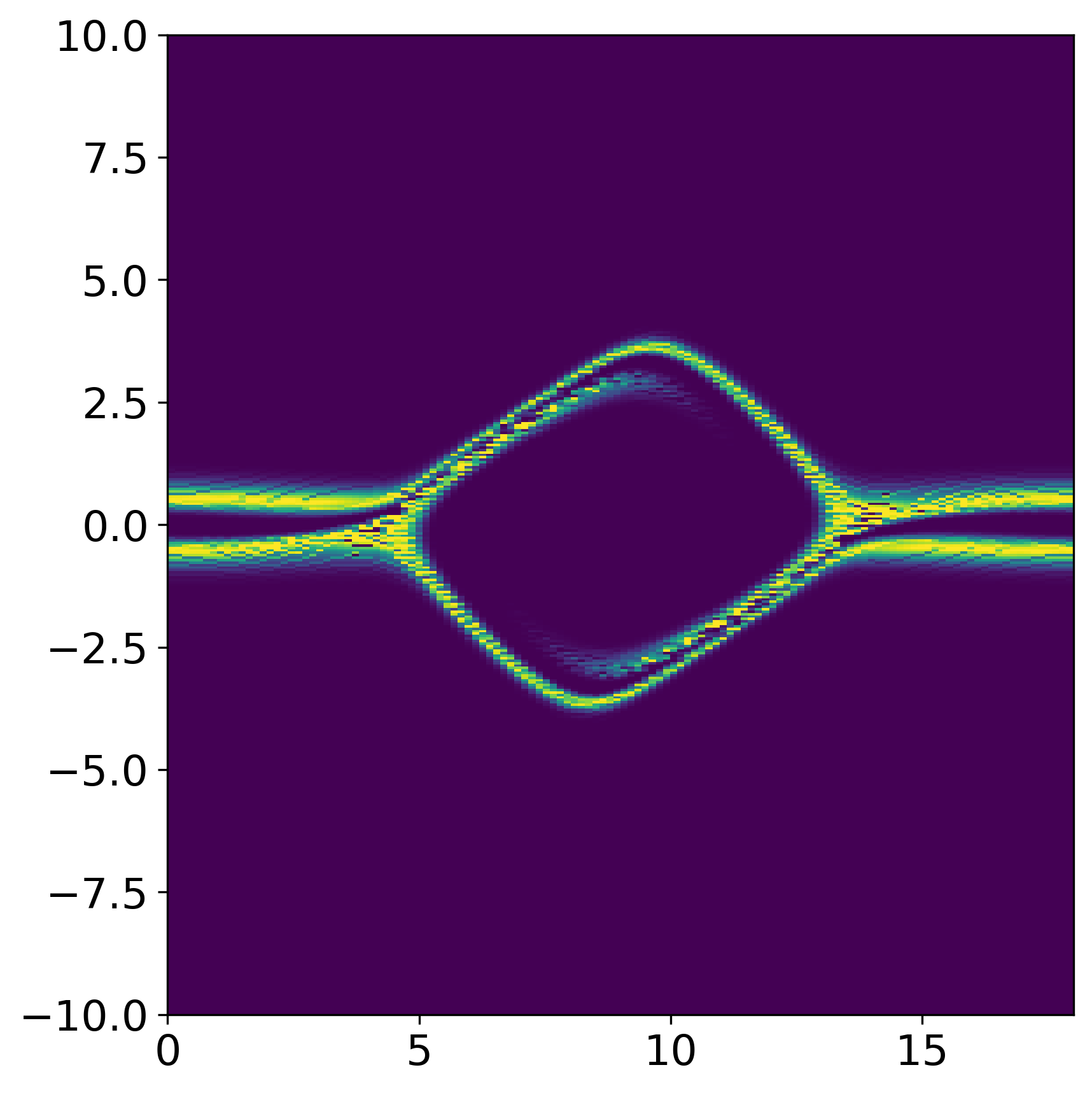}
    \caption{$L=400$ slices, $2^7$ cells}
    \label{fig:s1000_0}
\end{subfigure}
\hfill
\begin{subfigure}{0.24\textwidth}
    \includegraphics[width=\textwidth]{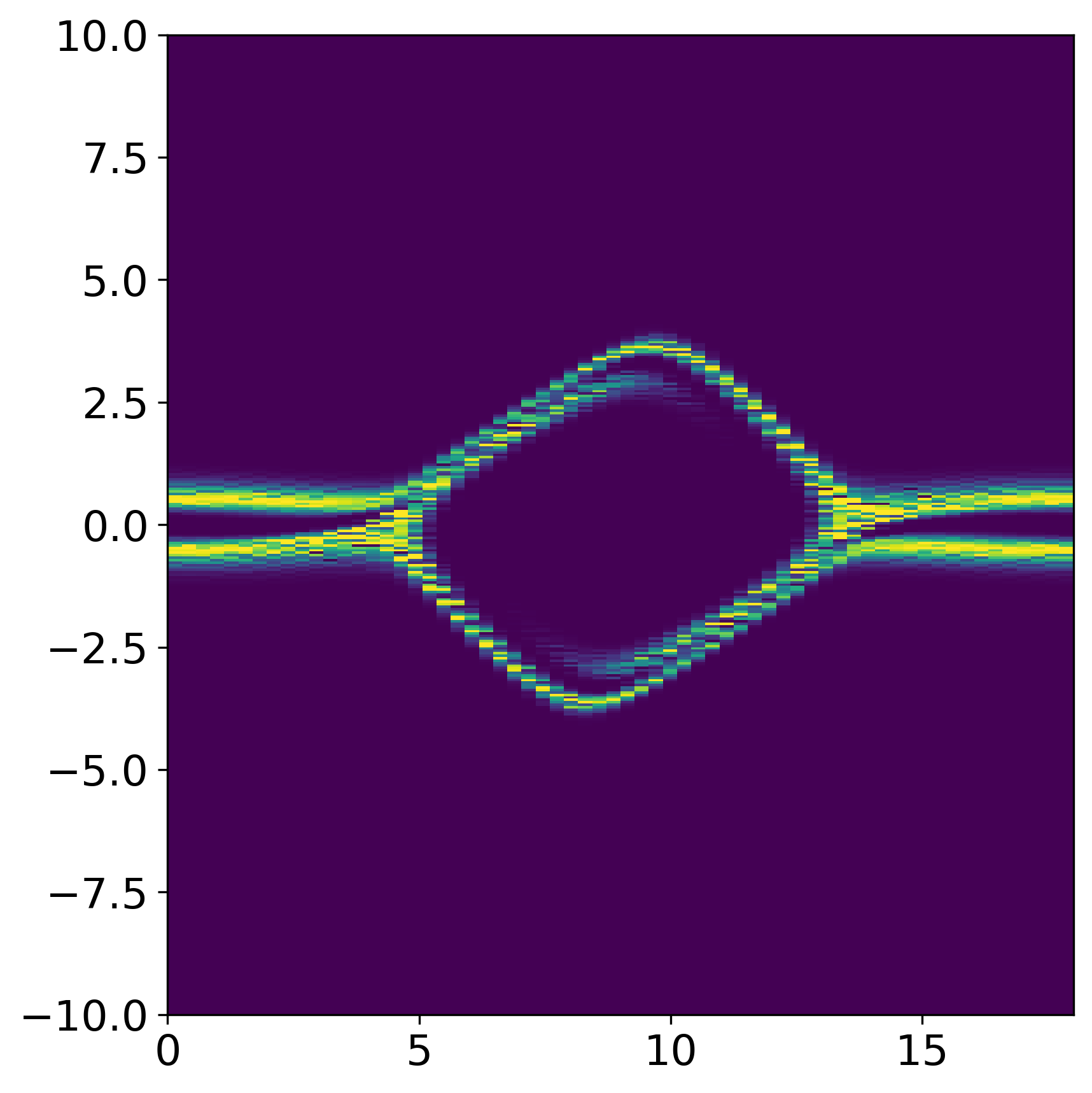}
    \caption{$L=400$ slices, $2^6$ cells}
    \label{fig:s1000_28}
\end{subfigure}
\hfill
\begin{subfigure}{0.24\textwidth}
    \includegraphics[width=\textwidth]{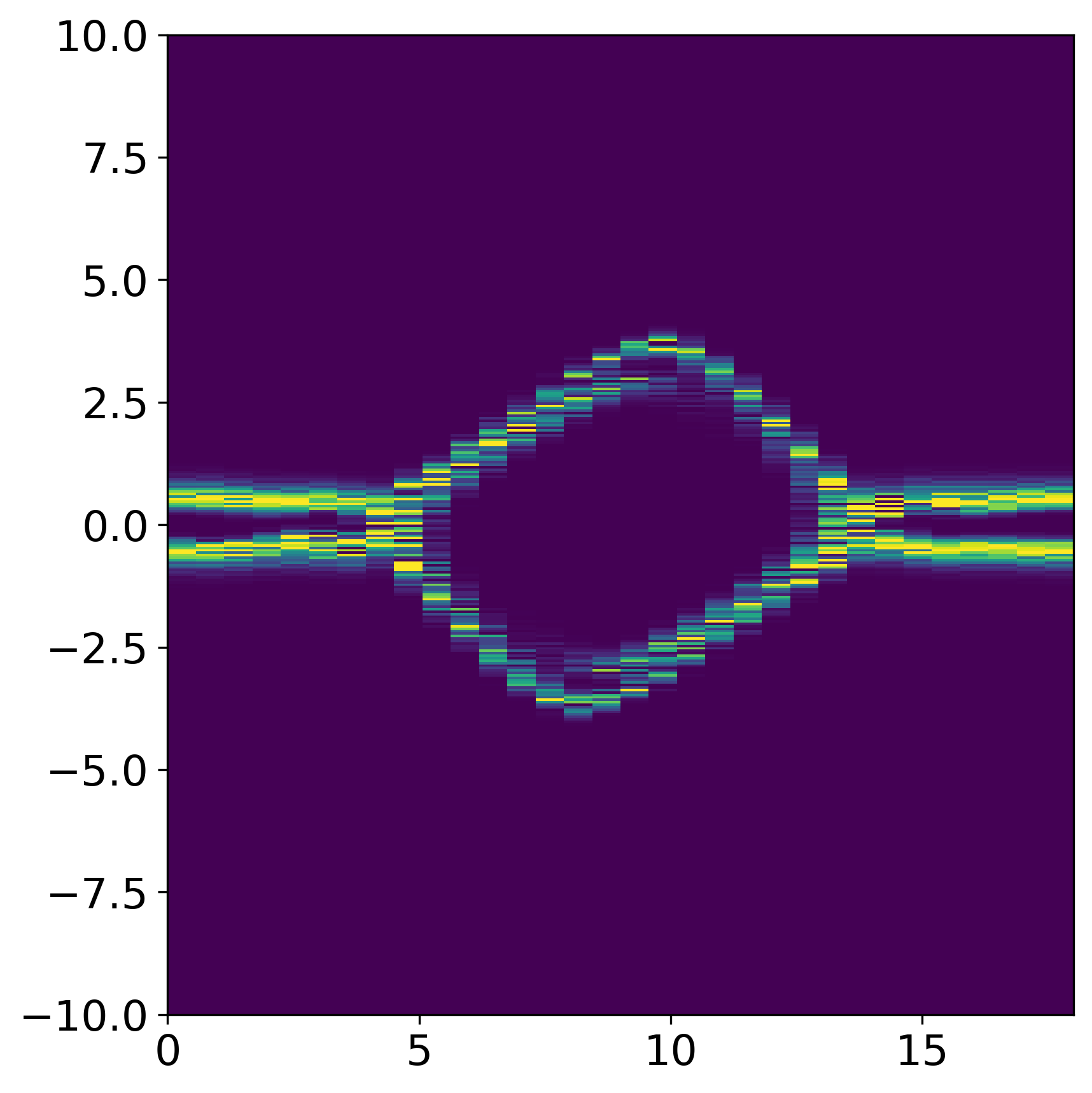}
    \caption{$L=400$ slices, $2^5$ cells}
    \label{fig:rho_six}
\end{subfigure}
\hfill
\begin{subfigure}{0.24\textwidth}
    \includegraphics[width=\textwidth]{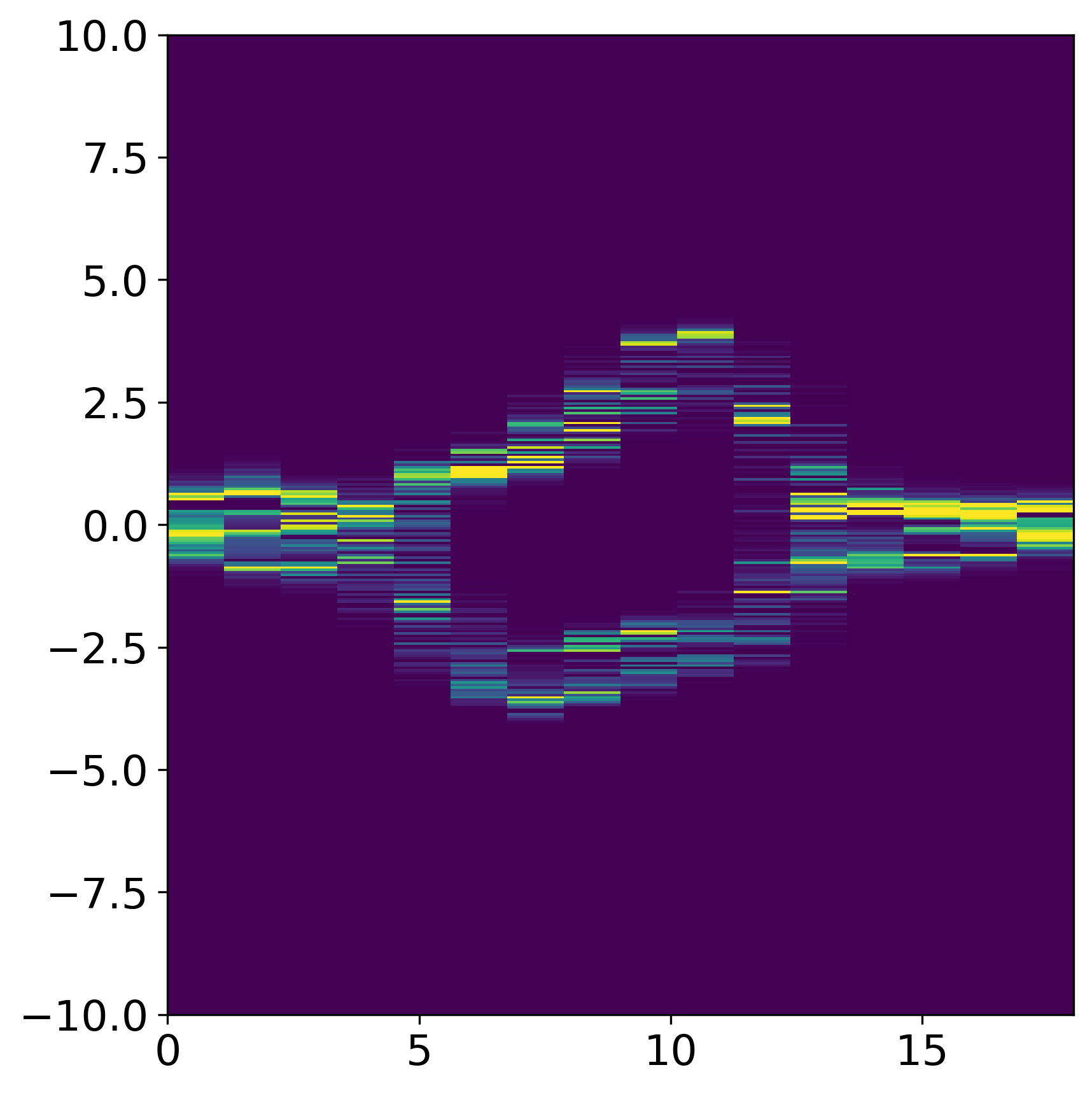}
    \caption{$L=400$ slices, $2^4$ cells}
    \label{fig:rho_six}
\end{subfigure}
\\ 
\begin{subfigure}{0.24\textwidth}
    \includegraphics[width=\textwidth]{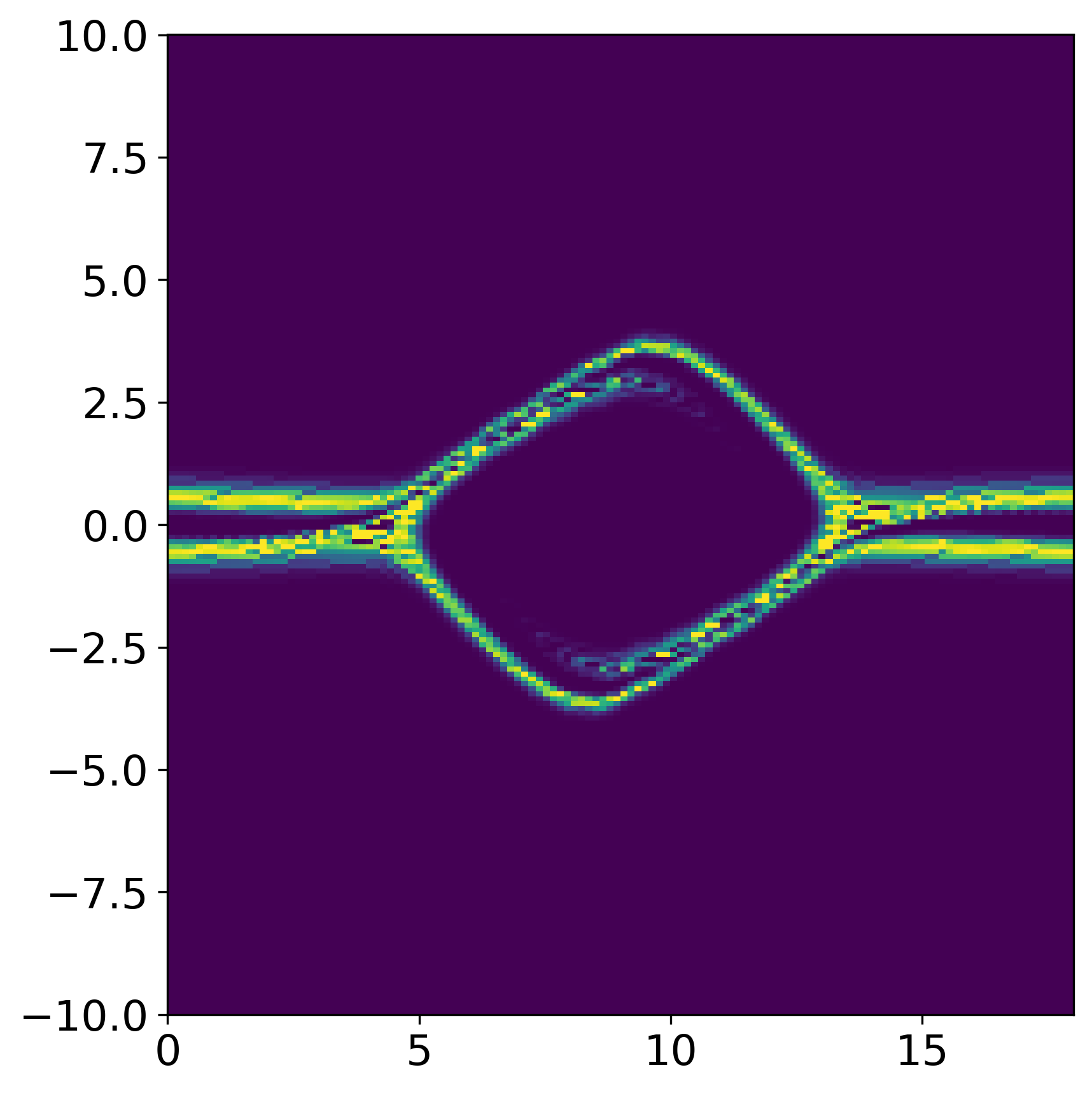}
    \caption{$L=200$ slices, $2^7$ cells}
    \label{fig:s1000_0}
\end{subfigure}
\hfill
\begin{subfigure}{0.24\textwidth}
    \includegraphics[width=\textwidth]{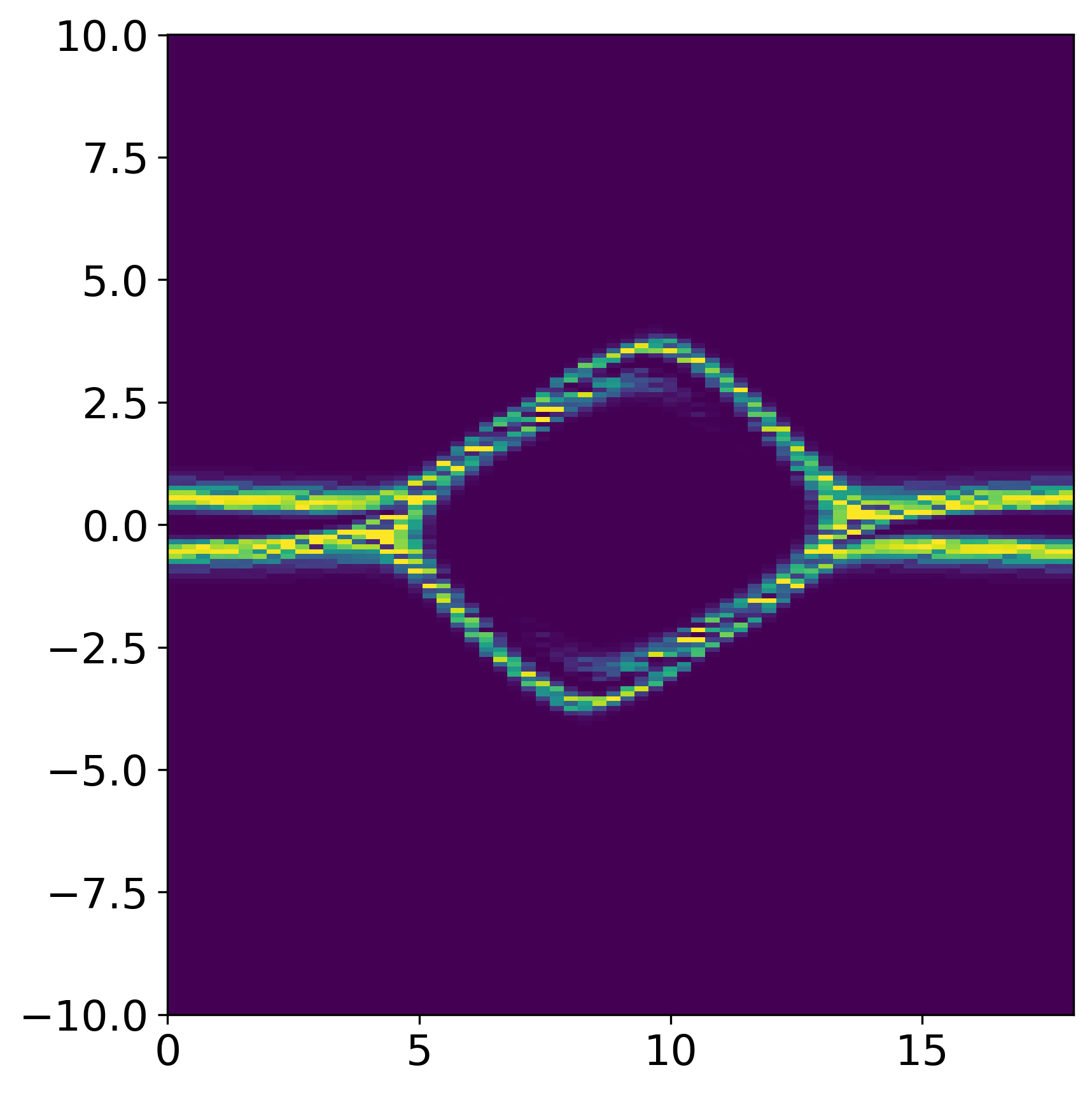}
    \caption{$L=200$ slices, $2^6$ cells}
    \label{fig:s1000_28}
\end{subfigure}
\hfill
\begin{subfigure}{0.24\textwidth}
    \includegraphics[width=\textwidth]{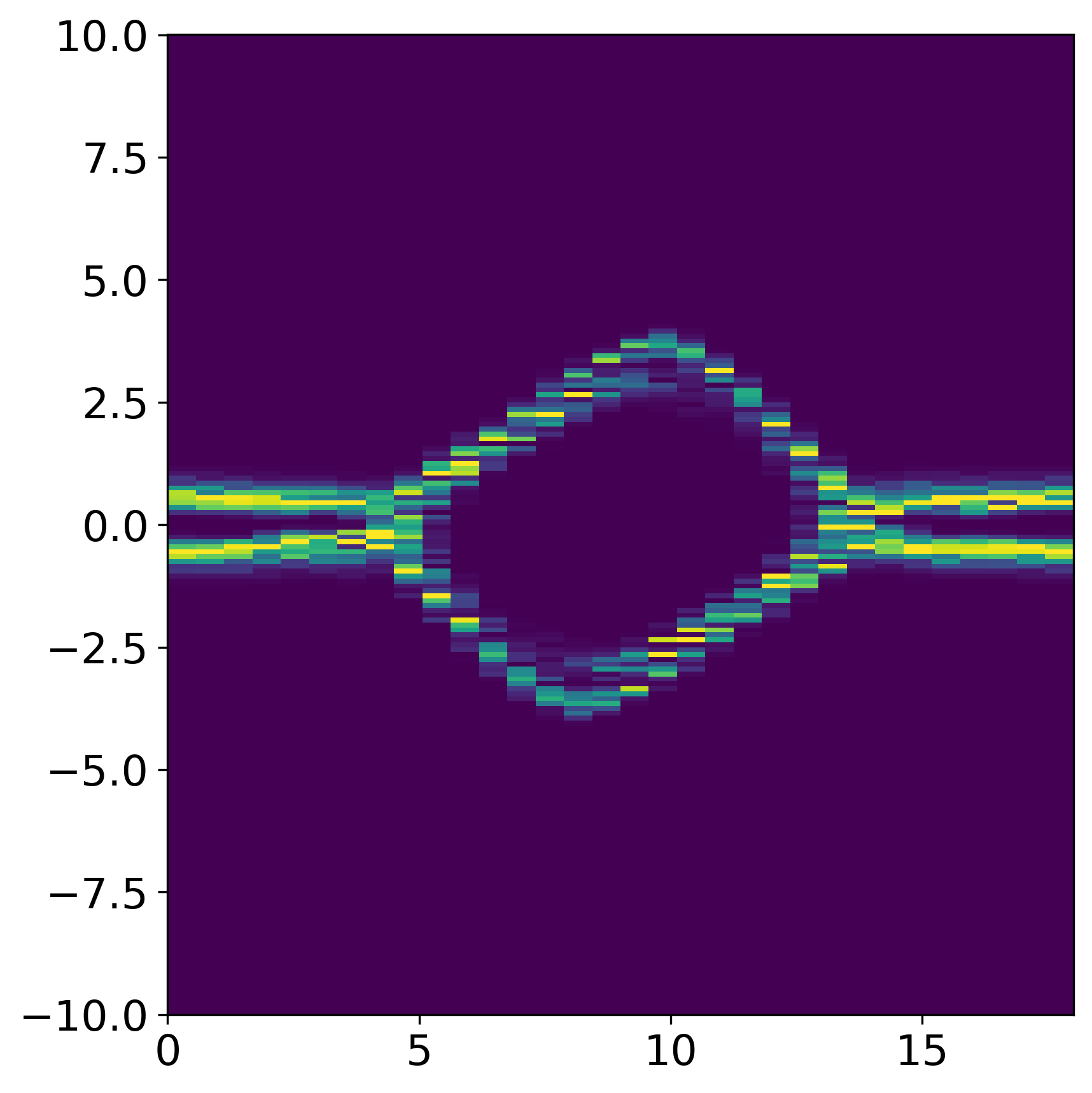}
    \caption{$L=200$ slices, $2^5$ cells}
    \label{fig:rho_six}
\end{subfigure}
\hfill
\begin{subfigure}{0.24\textwidth}
    \includegraphics[width=\textwidth]{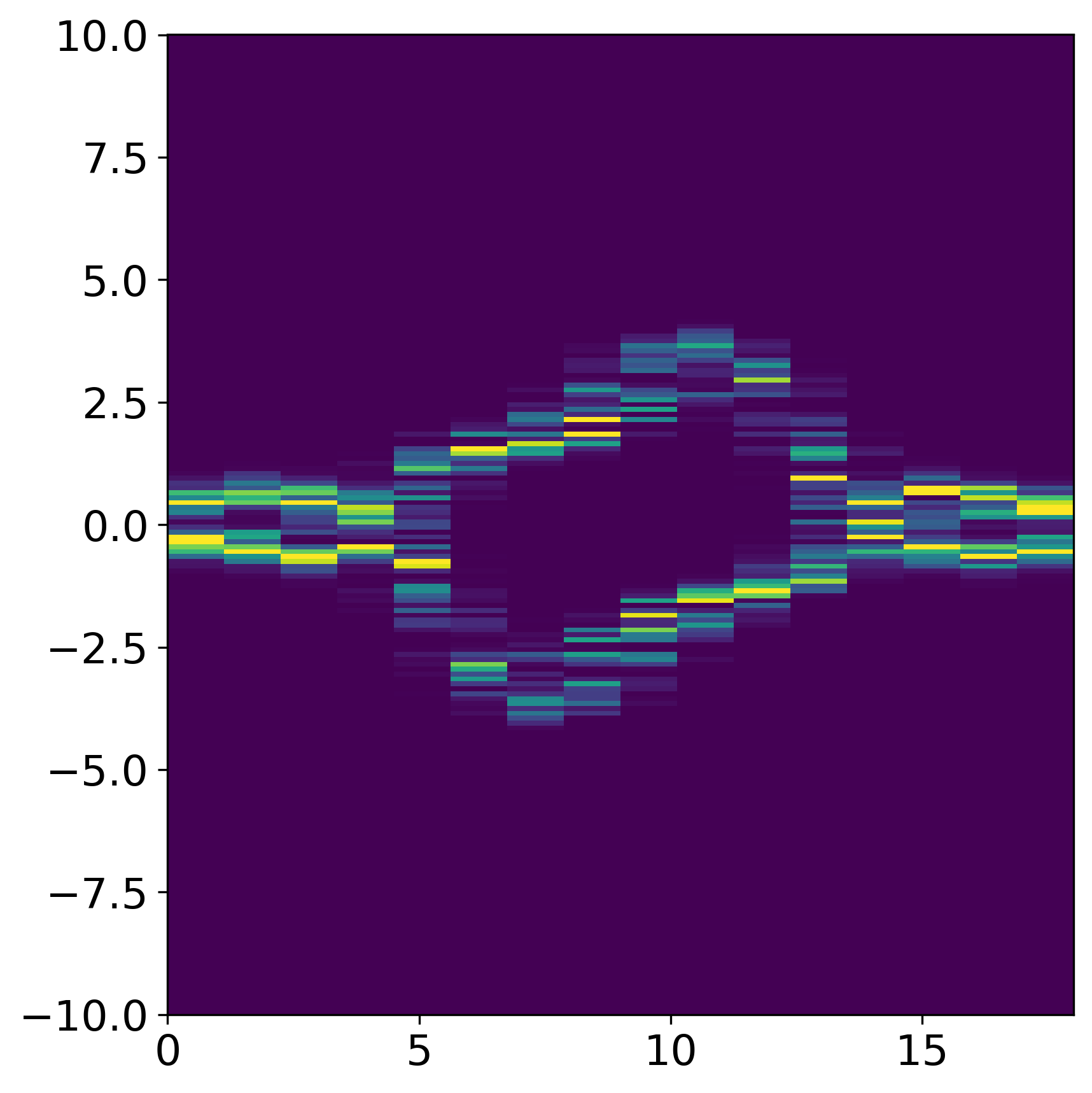}
    \caption{$L=200$ slices, $2^4$ cells}
    \label{fig:rho_six}
\end{subfigure}
\\ 
\begin{subfigure}{0.24\textwidth}
    \includegraphics[width=\textwidth]{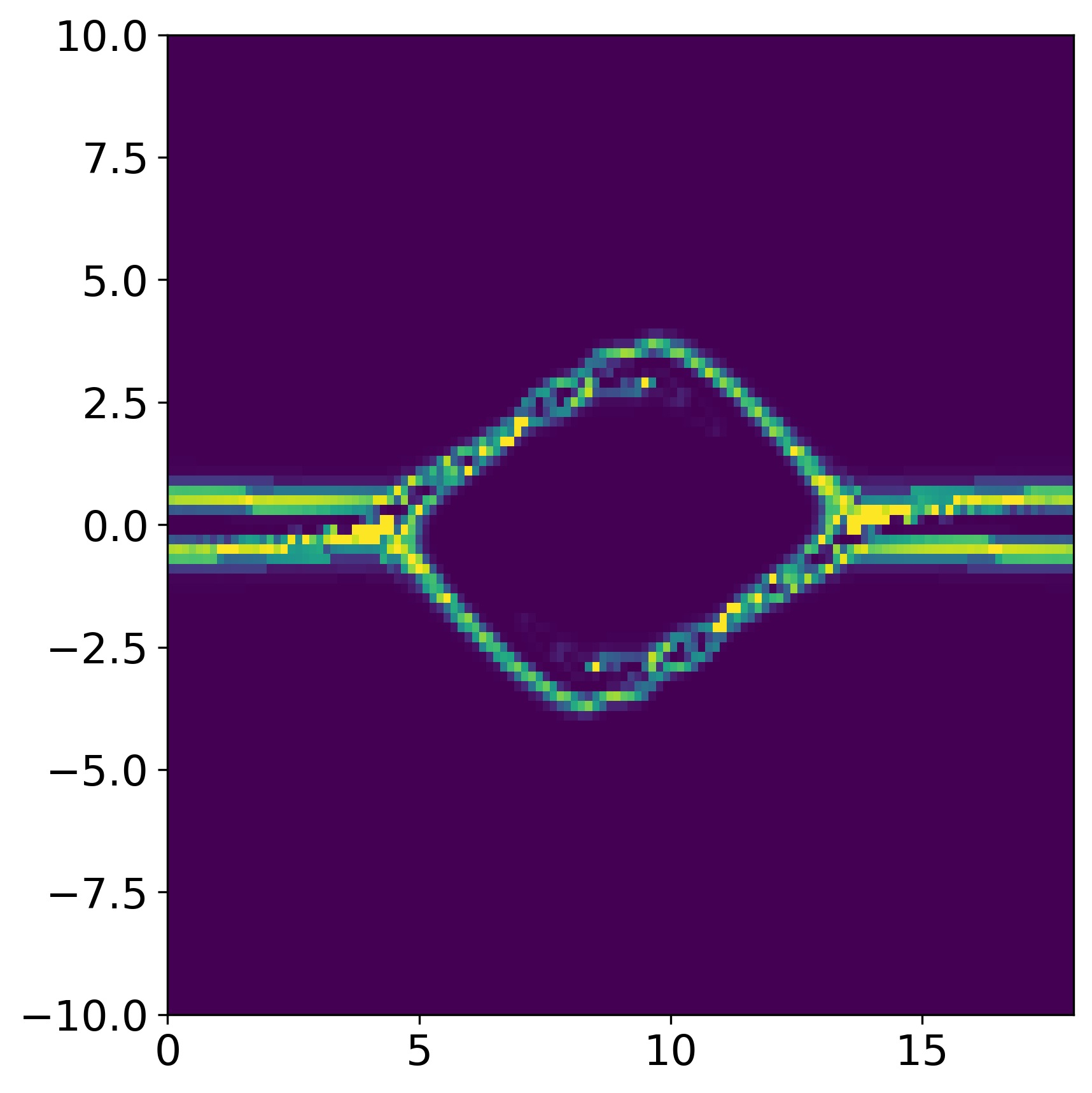}
    \caption{$L=100$ slices, $2^7$ cells}
    \label{fig:s1000_0}
\end{subfigure}
\hfill
\begin{subfigure}{0.24\textwidth}
    \includegraphics[width=\textwidth]{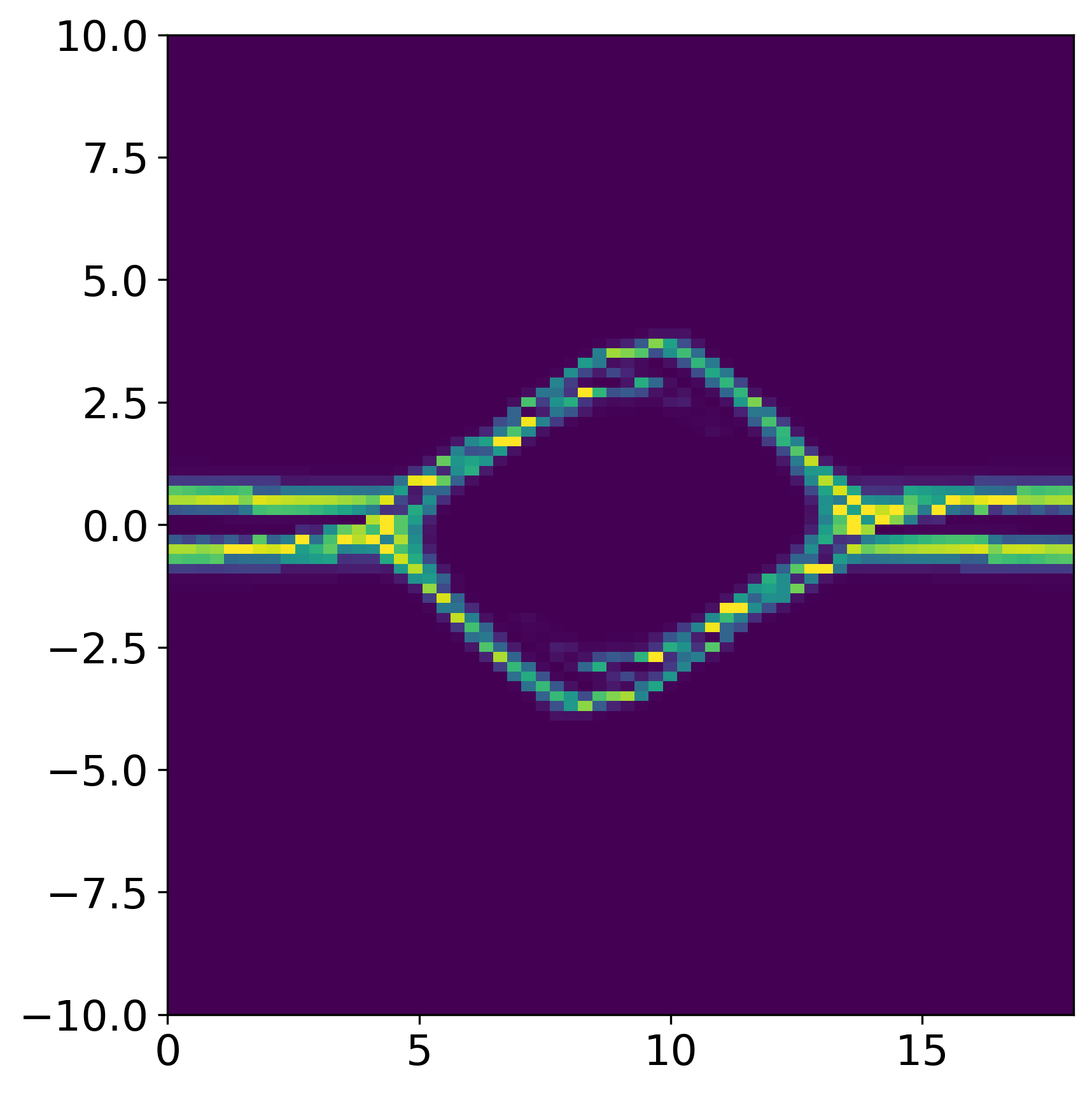}
    \caption{$L=100$ slices, $2^6$ cells}
    \label{fig:s1000_28}
\end{subfigure}
\hfill
\begin{subfigure}{0.24\textwidth}
    \includegraphics[width=\textwidth]{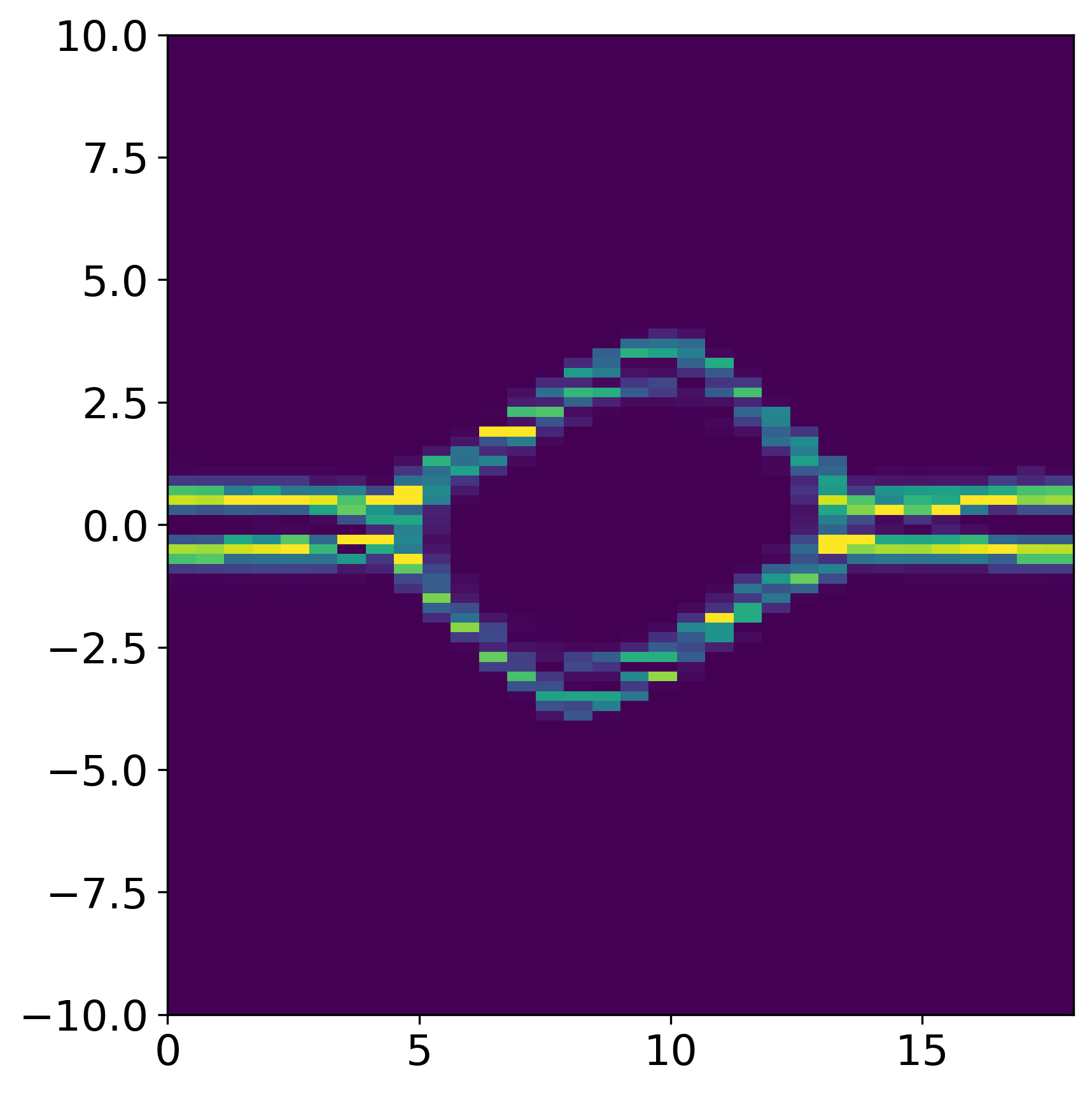}
    \caption{$L=100$ slices, $2^5$ cells}
    \label{fig:rho_six}
\end{subfigure}
\hfill
\begin{subfigure}{0.24\textwidth}
    \includegraphics[width=\textwidth]{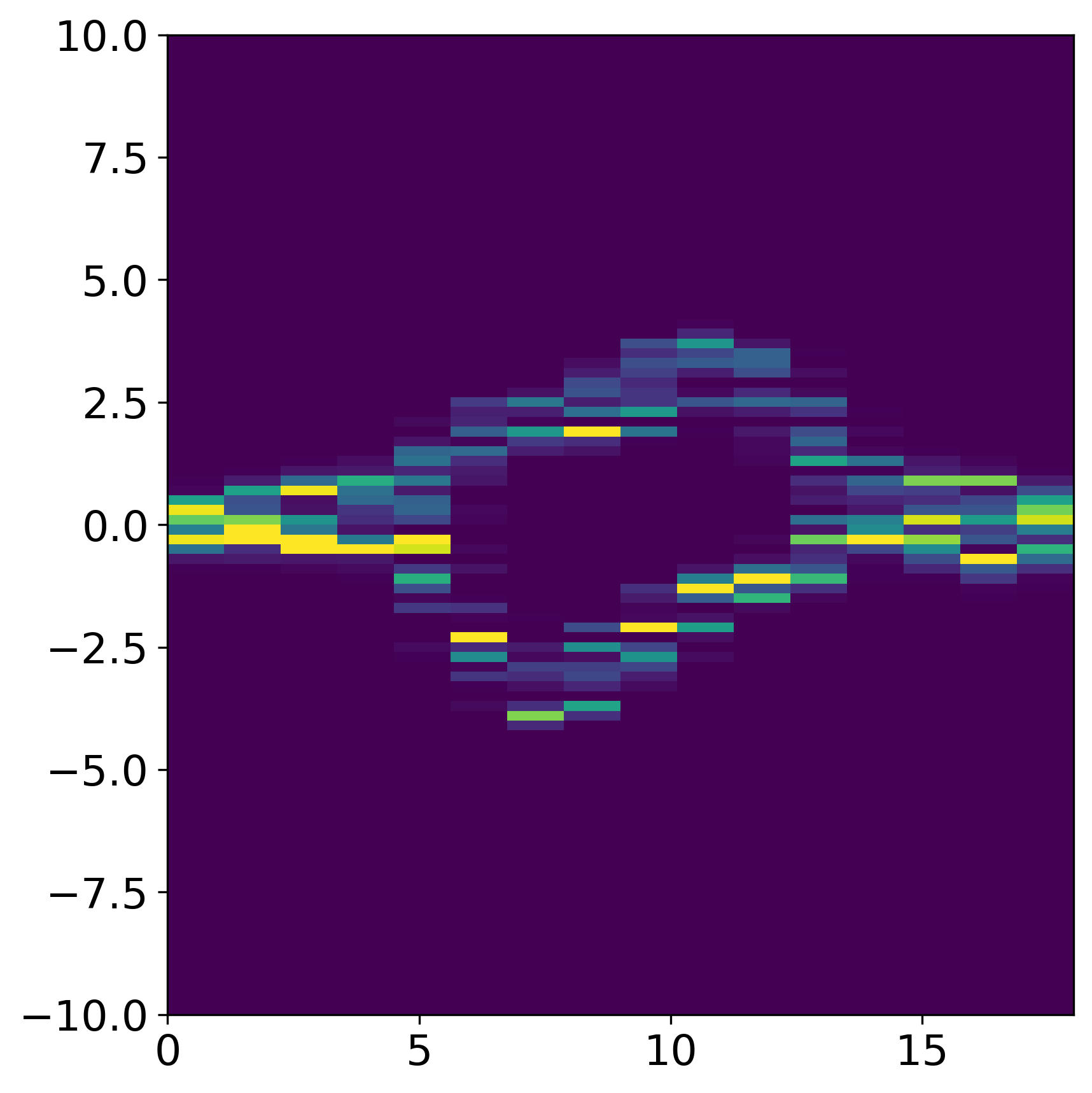}
    \caption{$L=100$ slices, $2^4$ cells}
    \label{fig:rho_six}
\end{subfigure}
\\ 
\centering
\begin{subfigure}{0.24\textwidth}
    \includegraphics[width=\textwidth]{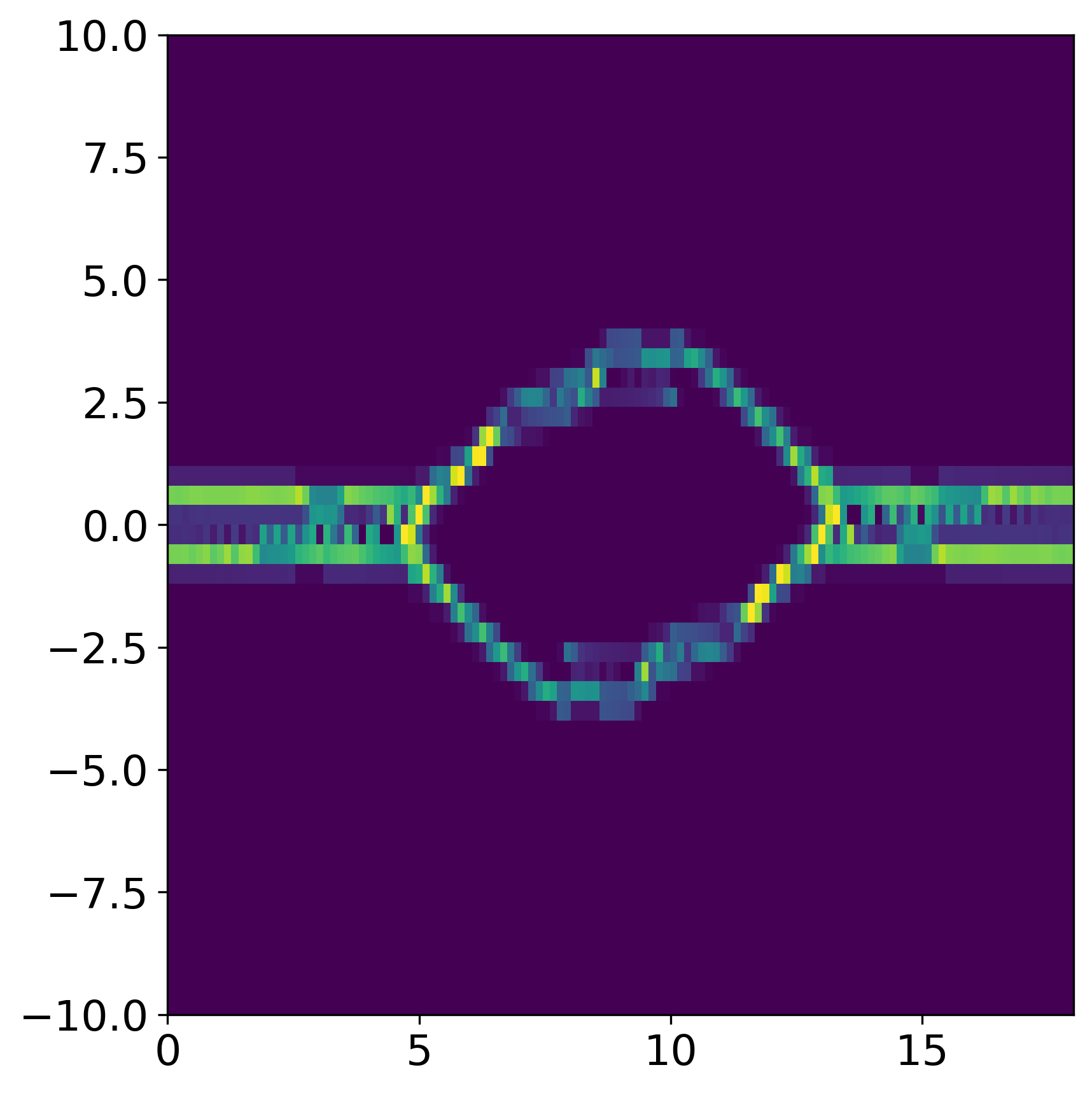}
    \caption{$L=50$ slices, $2^7$ cells}
    \label{fig:s1000_0}
\end{subfigure}
\hfill
\begin{subfigure}{0.24\textwidth}
    \includegraphics[width=\textwidth]{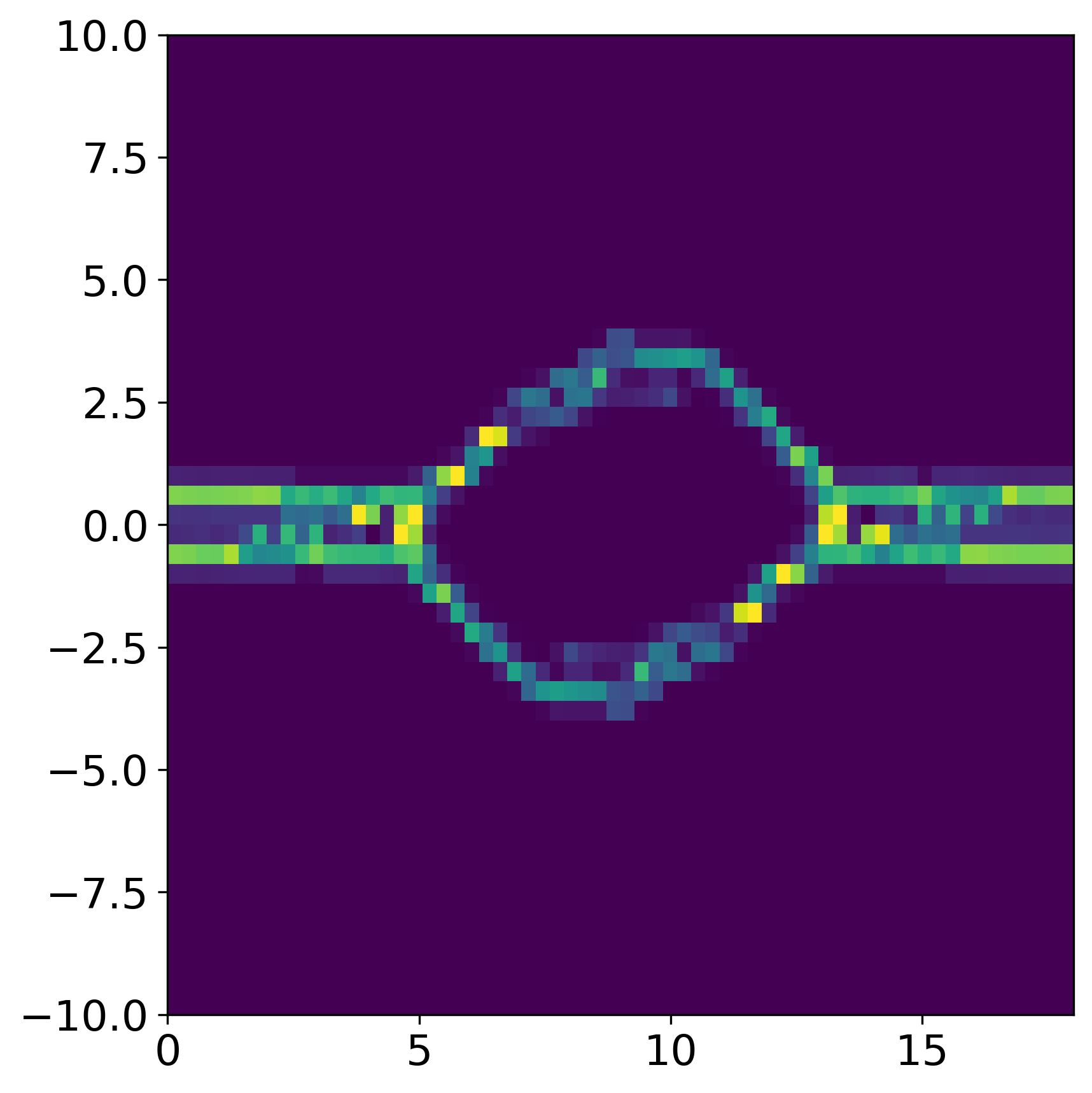}
    \caption{$L=50$ slices, $2^6$ cells}
    \label{fig:s1000_28}
\end{subfigure}
\hfill
\begin{subfigure}{0.24\textwidth}
    \includegraphics[width=\textwidth]{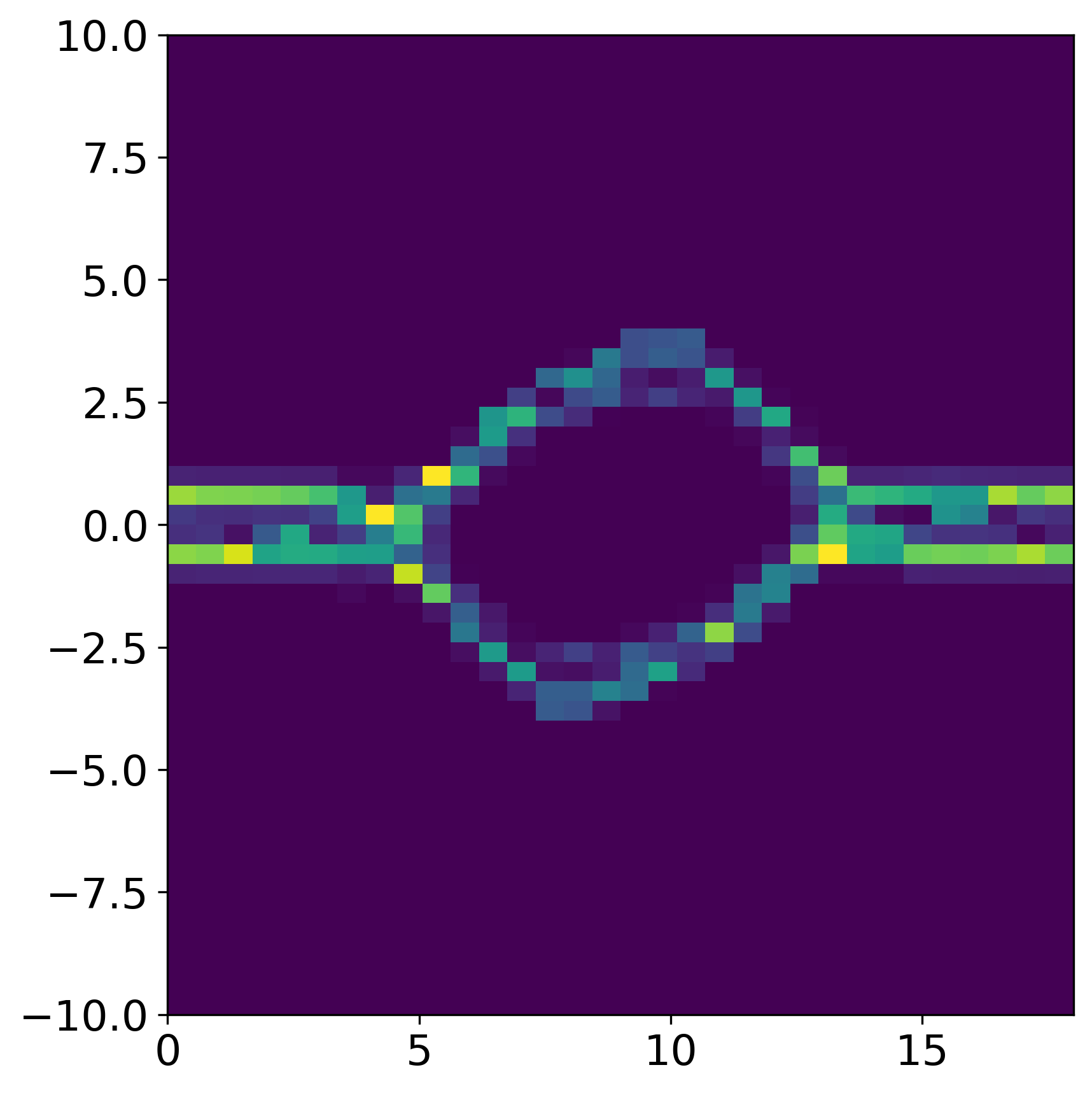}
    \caption{$L=50$ slices, $2^5$ cells}
    \label{fig:rho_six}
\end{subfigure}
\hfill
\begin{subfigure}{0.24\textwidth}
    \includegraphics[width=\textwidth]{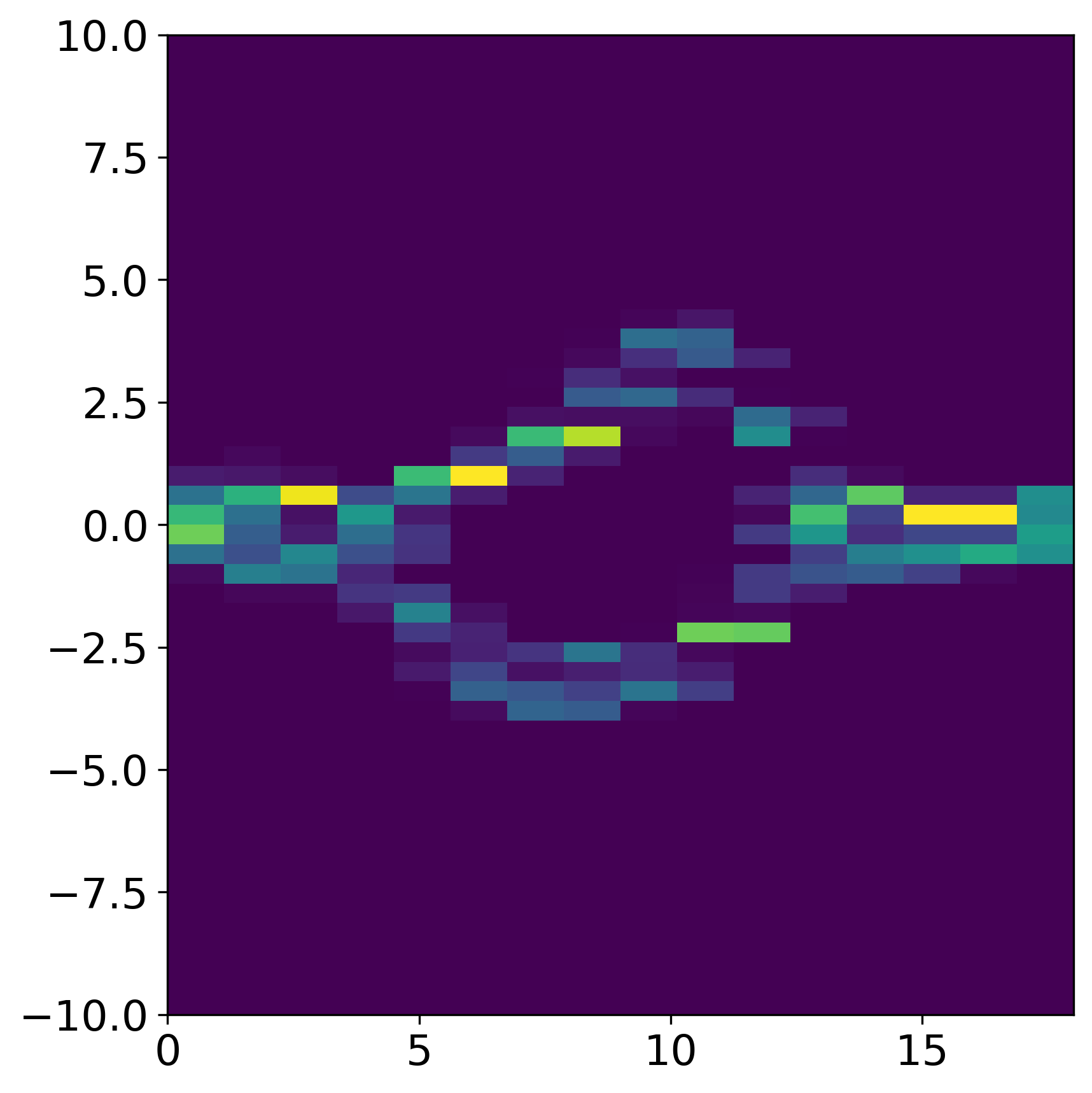}
    \caption{$L=50$ slices, $2^4$ cells}
    \label{fig:rho_six}
\end{subfigure}

  \caption{Relativistic two stream instability experiment in 1D with fluid to particle transformation. Snapshots at time $t=50$ for $L=\{400,200,100,50\}$ slices and $\{2^7, 2^6, 2^5, 2^4\}$ cells.}
    \label{fig:rel_tis_total}
\end{figure}

In Figure \ref{fig:rel_tis} we plot the electric energy for several cases and compare against the theoretical growth rate \cite{SHIROTO201932} and the reference pure particle simulation. 

\begin{figure}[H]
    \centering
    \includegraphics[width=1.0\textwidth]{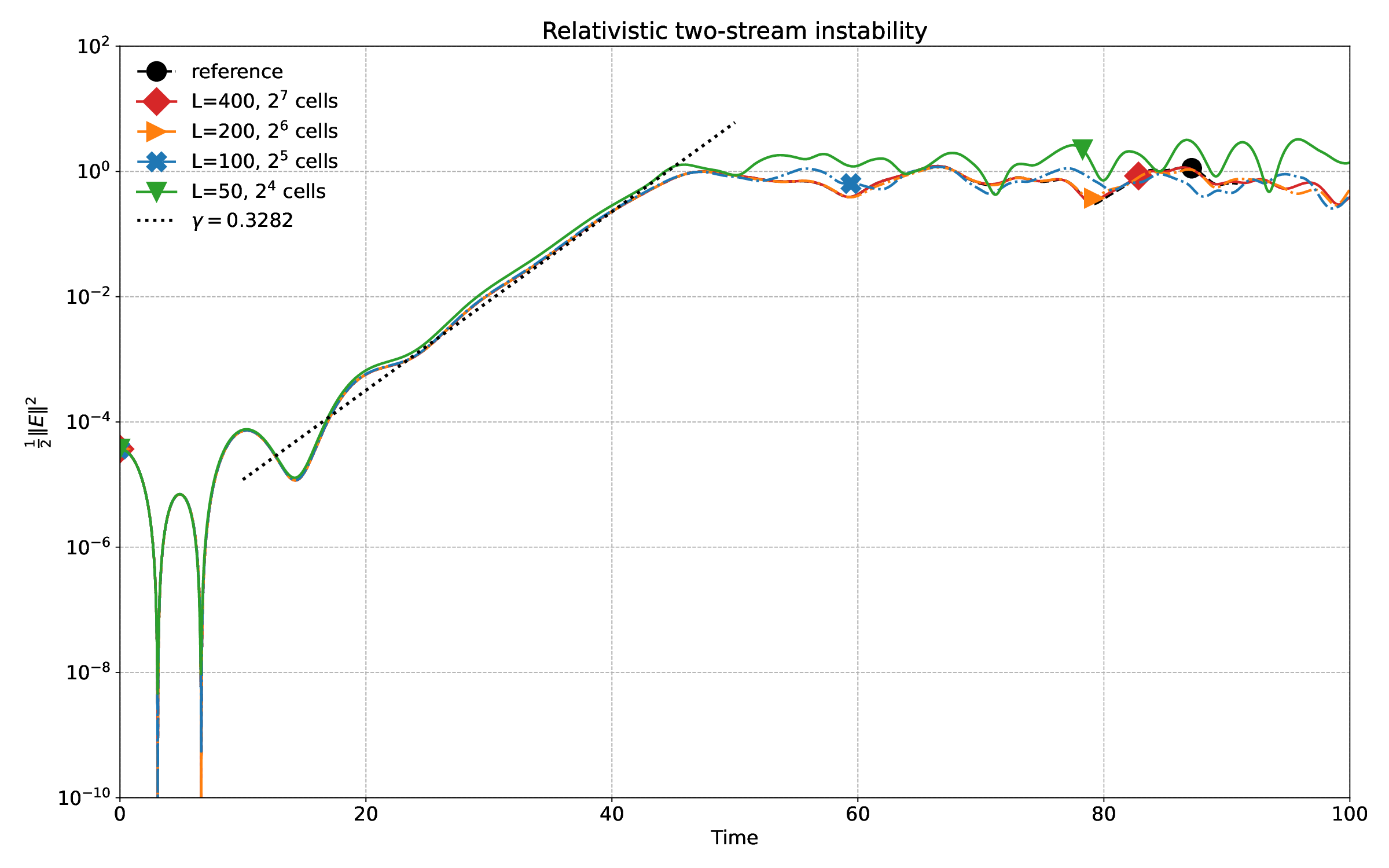}
    
    \caption{Relativistic two stream instability experiment in 1D with fluid to particle transformation.}
    \label{fig:rel_tis}
\end{figure}

Figure \ref{fig:rel_tis_ratio} plots the ratio of the number of fluid nodes to the initial number of fluid nodes at $t=0$ for several cases.

\begin{figure}[H]
    \centering
    \includegraphics[width=0.5\textwidth]{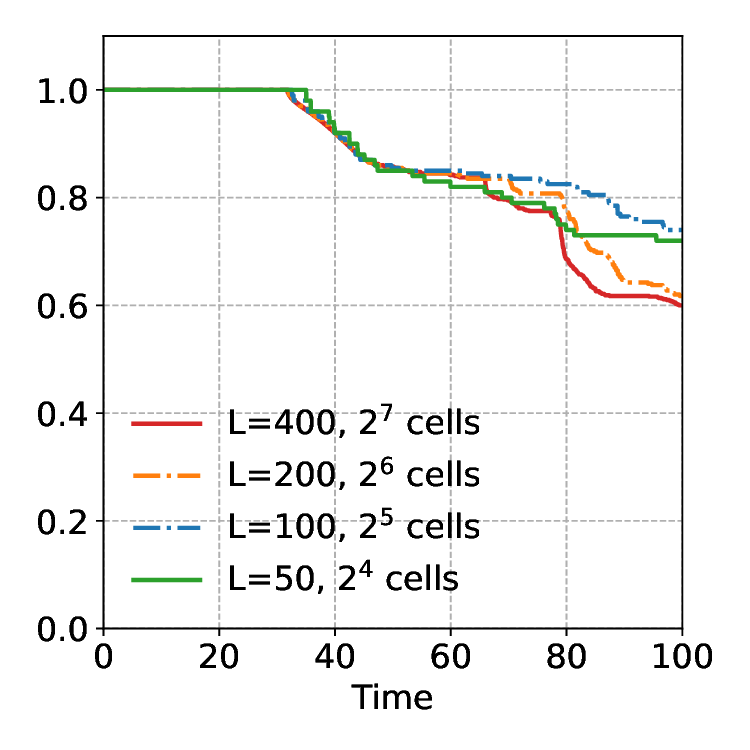}
    
    \caption{Relativistic two stream instability experiment in 1D with fluid to particle transformation. The vertical axis plots the ratio of the number of fluid nodes at time $t$ to the number of fluid nodes at time $t=0$.}
    \label{fig:rel_tis_ratio}
\end{figure}

Lastly, we study how the number of particles per cell of the fluid node affect the quality of the results. For this we choose the case of $L=100$ slices, finite volume grid resolution $2^5$ cells, and consider $\{1,2,4,8\}$ ppcpfn. Figure \ref{fig:rel_tis_ppf} shows that increasing the number of particles slightly improves the results but is less important than the number of slices and the grid resolution. 
\begin{figure}[H]
\begin{subfigure}{0.24\textwidth}
    \includegraphics[width=\textwidth]{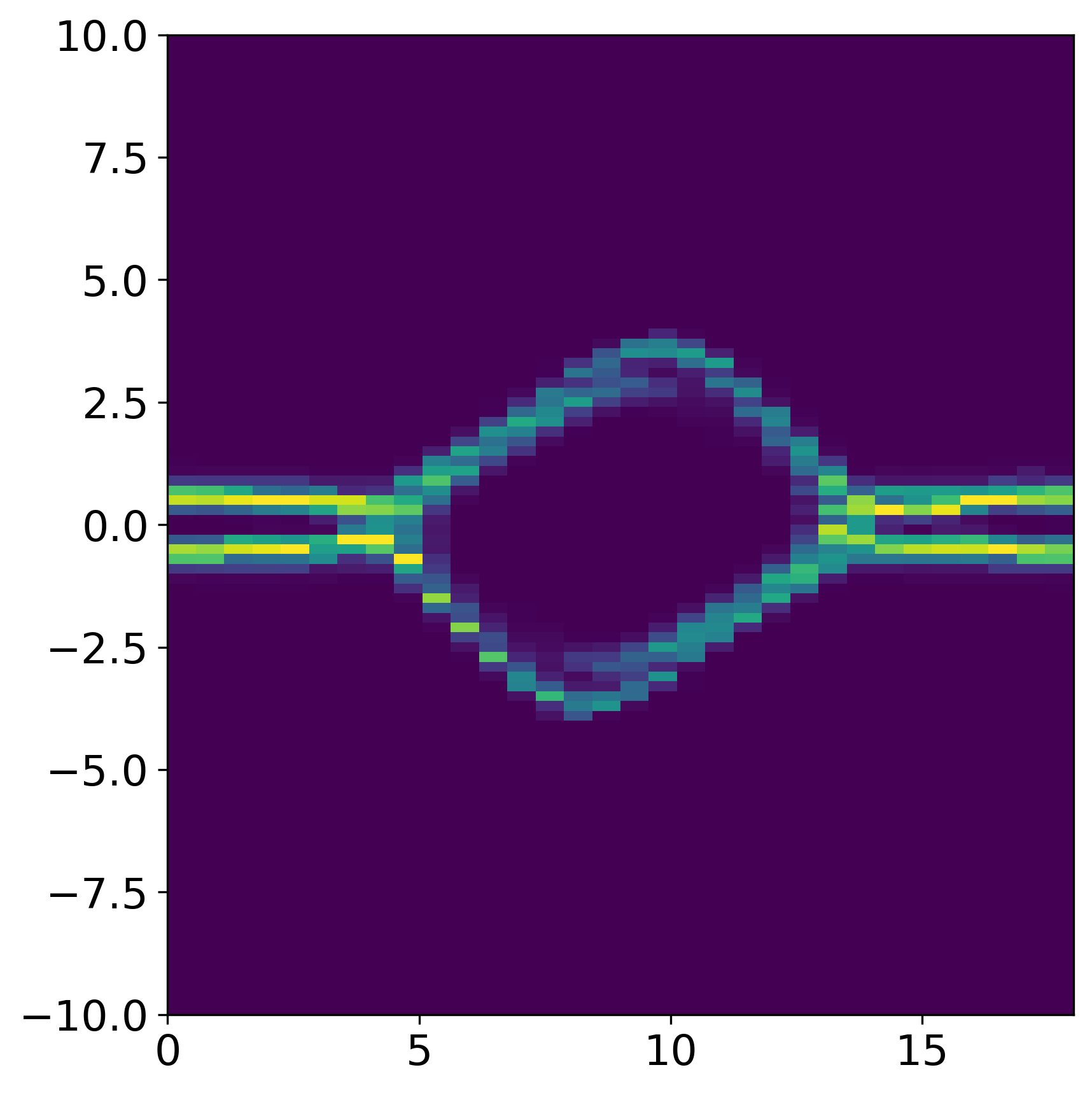}
    \caption{$8$ ppc}
    \label{fig:s1000_28}
\end{subfigure}
\hfill
\begin{subfigure}{0.24\textwidth}
    \includegraphics[width=\textwidth]{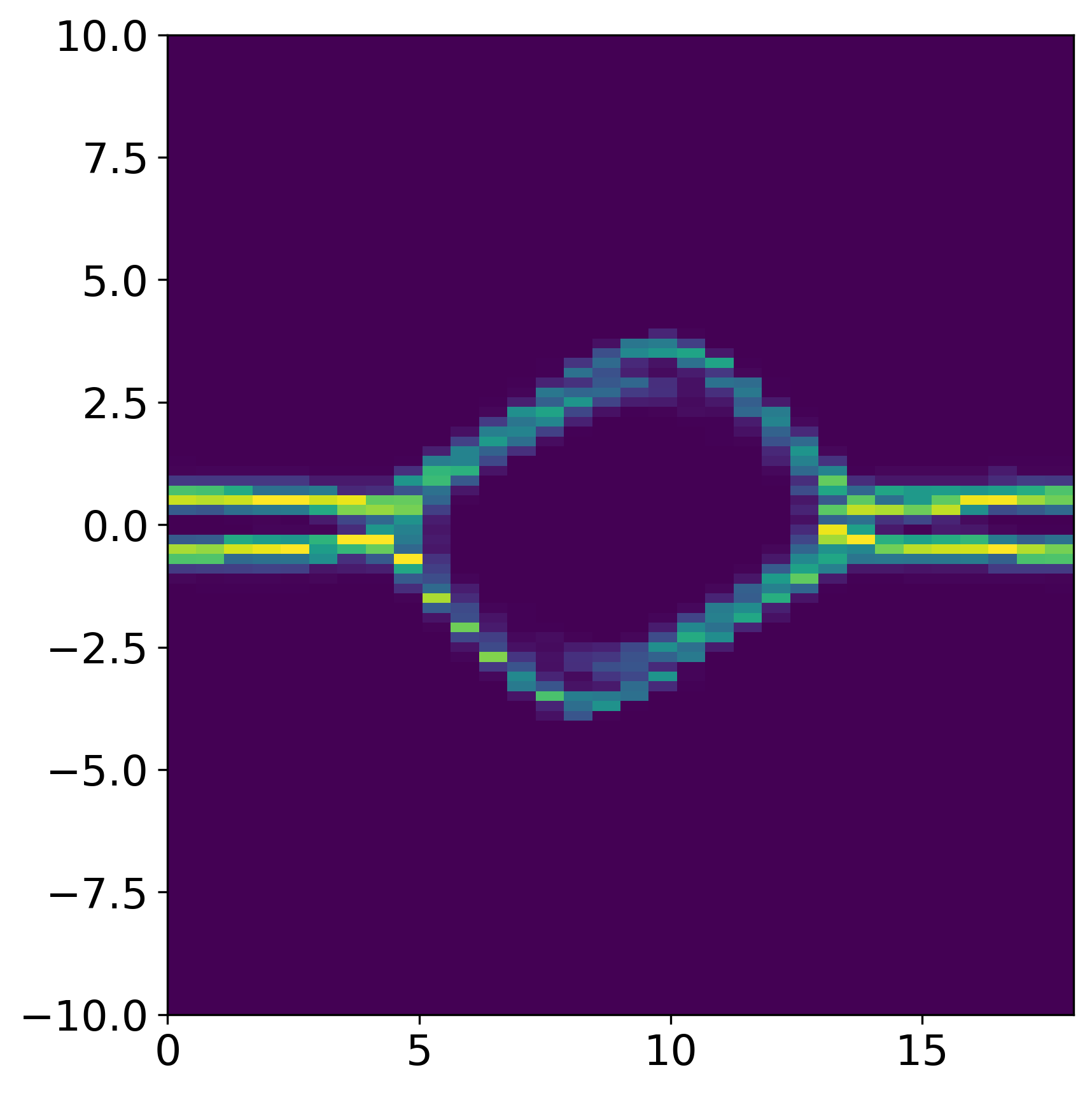}
    \caption{$4$ ppc}
    \label{fig:s1000_28}
\end{subfigure}
\hfill
\begin{subfigure}{0.24\textwidth}
    \includegraphics[width=\textwidth]{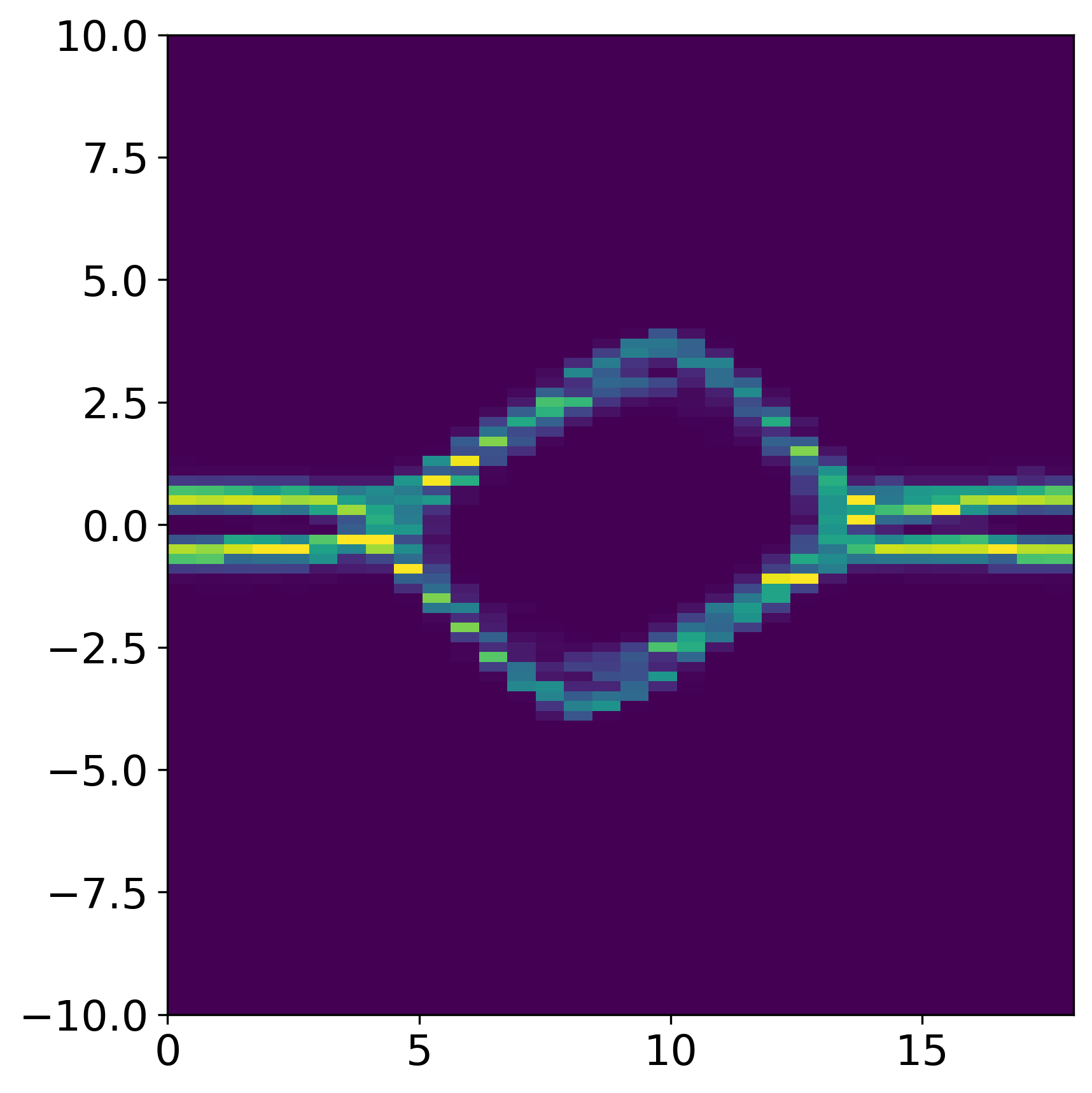}
    \caption{$2$ ppc}
    \label{fig:rho_six}
\end{subfigure}
\hfill
\begin{subfigure}{0.24\textwidth}
    \includegraphics[width=\textwidth]{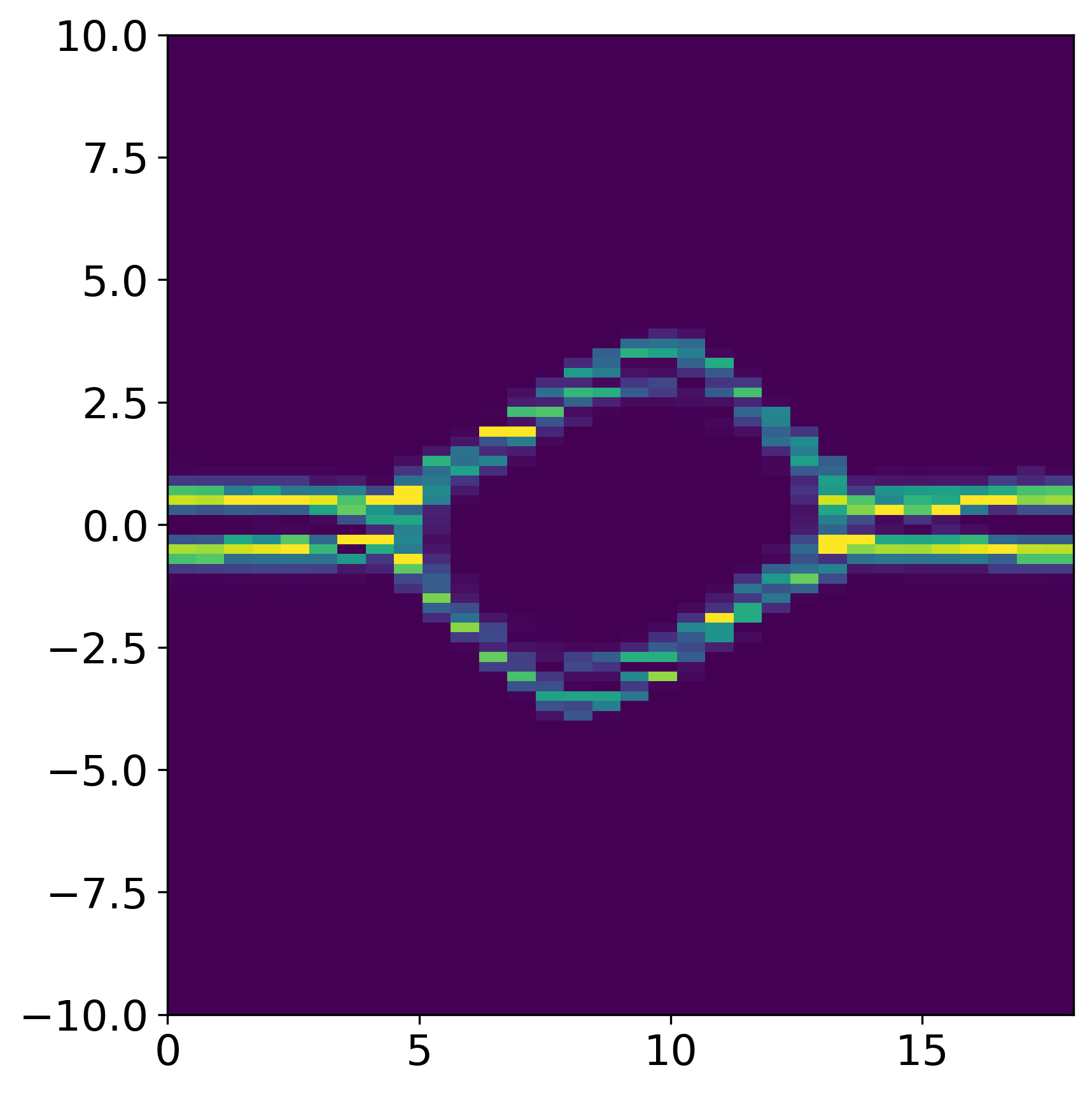}
    \caption{$1$ ppc}
    \label{fig:rho_six}
\end{subfigure}

  \caption{Relativistic two stream instability experiment in 1D with fluid to particle transformation. Snapshots at time $t=50$ with $L=100$ slices, $2^5$ cells, and the number of particles per cell per fluid node $\{1,2,4,8\}$.}
    \label{fig:rel_tis_ppf}
\end{figure}

\section{Summary and Outlook}

We have derived exact moment models and a slicing moment method based on these models. Our novel method is a new type of particle-in-cell method where we combine fluid moments in the spatial coordinate and a particle-like treatment in the velocity domain. The major advantage of the method over a standard particle-in-cell method is the fact that the solution does not suffer from numerical noise in physical space. The main limitation of the method is that the method fails to accurately represent the phase-space dynamics when the mapping $\bm v_{s,p}$ is no longer one-to-one. Future work shall extend the numerical experiments to higher dimensions and a comparison of the novel method to fully kinetic solvers for more realistic setups. 

 Other issues are also to be addressed such as the inclusion of collision operators, robust and accurate fluid to particle transformation as well as reverse particle to fluid transformation to enable efficient simulation of different scenarios. To incorporate higher degree fluid discretizations, center manifold reduction of the fluid equations would be considered.


\appendix

\section{Appendix}
\subsection{Proof of Theorem \ref{prop:main_theorem} in 1u1x} \label{sec:proof_main_theorem_1u1x}

 Before proving the theorem in its general form, we provide a proof to the system in one spatial dimension. The proof in one dimension outlines the general idea of the proof and provides a more detailed and clear view on the steps. 

Let us first collect some useful properties of the $\delta$-distribution, see \cite[Chapter III, Section 1]{gelfand}. We assume $\bm v(\bm x, t) \in C^{\infty}(\Omega_{x})^d$. We have the sifting property:
\begin{align} \label{eq:sifting}
\int_{\Omega} \delta(\bm u -  \bm v) \phi \, dV = \int_{\Omega_x} \int_{\Omega_t} \phi( \bm v) \, dt \, dx,  \hspace{1cm} \forall \phi \in C^{\infty}_0(\Omega).
\end{align}
We note that there exist an invertible transformation with non-vanishing Jacobian given by $( \bm \theta(\bm u, \bm x, t)$, $\bm \chi(\bm u, \bm x, t)$, $\tau(\bm u, \bm x, t) )= \left( \bm u - \bm v(\bm x, t), \bm x, t \right)$. The existence of this transformation is sufficient for the chain rule to hold \cite[Chapter III, Section 1]{gelfand}:
\begin{align} \label{eq:chain_x} 
\nabla \delta(\bm u -  \bm v) &= - \, \nabla \bm v \cdot \nabla_u \delta(\bm u - \bm v) ,\\ \label{eq:chain_t}
\partial_t \delta(\bm u -  \bm v) &= - \, \partial_t \bm v \cdot \nabla_u \delta(\bm u -  \bm v).
\end{align} 

Additionally, for a smooth function $G(\bm u, \bm x, t)$ the following product rule formulas hold.
\begin{align} 
\partial_{x_i} \left(G \, \delta(\bm u-  \bm v) \right) & = \partial_{x_i} G \, \delta(\bm u - \bm v)  + G \, \partial_{x_i} (\delta(\bm u -  \bm v)), \\ 
\partial_{u_i} \left(G \, \delta(\bm u-  \bm v) \right) & = \partial_{u_i} G \, \delta(\bm u -  \bm v) + G \, \partial_{u_i}\delta(\bm u -  \bm v).
\end{align}

In one-dimension, we denote derivatives with respect to $x$, $u$, $v$ by $\partial_x$, $\partial_u$, and $\partial_{v}$, respectively.
\begin{proof}[Proof (1u1x)]
We want to show that 
\begin{align} \label{eq:general_advec_1d}
\partial_t g_N + \partial_x \left( G(u) g_N \right) + \partial_u \left( F(u) g_N \right) + R(u) g_N = 0,
\end{align}
with $g_N = \sum_{i=0}^N \frac{(-1)^i}{i!} \, C^i \, \partial_u^{i} \delta(u- v)$, i.e. the ansatz \eqref{eq:burby_ansatz} in one dimension.
First we observe that from the definition of the ansatz \eqref{eq:burby_ansatz}, $\partial_t g_N$ is additive in $\partial_t C^k$ and, by the chain rule \eqref{eq:chain_t}, is additive in $\partial_t v$. From \eqref{eq:centered_hierarchy} and \eqref{eq:velocity} we see that $\partial_t C^k$ and $\partial_t v$ are additive in $G, F,$ and $R$. It follows that $\partial_t g_N$, and consequently \eqref{eq:general_advec_1d}, is additive in $G, F$ and $R$.
Due to the additive property, to prove the theorem it suffices to prove the following three distinct cases: (1) $F(u, x, t)=0, R(u, x, t)=0$; (2)  $G(u, x, t)= 0, R(u, x, t)=0$; (3) $G(u, x, t) = 0, F(u, x, t)=0$.
The first case is most complex, so we present it in detail proceeding in three steps. We consider
\begin{align} \label{eq:G_1d}
\partial_t g_N + \partial_x \left( G(u) g_N \right) = 0.
\end{align}
Thus, it suffices to show that
\begin{align} \label{eq:main_eq_G}
\sum_{k=0}^N \frac{(-1)^{k}}{k!} \, \partial_t C^k \,\, \partial_u^{k} \delta + \sum_{k=0}^N \frac{(-1)^{k}}{k!} \, C^k \,\, \partial_u^{k+1} \delta \,\, (-  \, \partial_t v) + \partial_x (G(u)g_N)  = 0,
\end{align}
where $\delta := \delta(u- v)$.

\textit{Step 1: Find expression for $\partial_t C^k$} 

\noindent Multiplying \eqref{eq:G_1d} with $(u- v)^k$ and integrating over $\Omega_u$, we obtain
\begin{align}
\int_{\Omega_u} \partial_t \left( g_N \, (u- v)^k \right) \, du - \int_{\Omega_u} k \, g_N \, (u- v)^{k-1} ( -  \, \partial_t v) \, du \\ \notag + \partial_x \int_{\Omega_u} g_N \, G(u) \, (u- v)^k du - \int_{\Omega_u} g_N \, G(u) \, k \, (u- v)^{k-1} \, ( -  \partial_x v) \, du = 0,
\end{align}
where we used product rules with respect to $\partial_t$ and $\partial_x$. We note that the second and the last terms vanish for $k=0$.

Next, we substitute the equation for the center, $\partial_t v = -G( v) \, \partial_x v$, and the ansatz for $g_N$
\begin{align} \label{eq:Ck_1d}
&\partial_t C^k - \int_{\Omega_u} \sum_{i=0}^{N} \frac{(-1)^{i}}{i!} C^i \, \partial_u^i \delta \, k \,   \, (u- v)^{k-1} \, \left( G( v) \partial_x v \right) \, du \\ \notag
&+ \partial_x \int_{\Omega_u} \sum_{i=0}^{N} \frac{(-1)^{i}}{i!} \, C^i \, \partial_u^i \delta \, G(u) \, (u- v)^k \, du + \int_{\Omega_u} \sum_{i=0}^N \frac{(-1)^i}{i!} \, C^i \partial_u^i \delta \, G(u) \, k \,   \, (u- v)^{k-1} \, \partial_x v \, du .
\end{align}
Next, all terms with $\partial_u^i \delta$ are integrated by parts with respect to $u$ variable to obtain expressions containing only $\delta$. We illustrate this explicitly for the first term on the second line
\begin{align}
\int_{\Omega_u} \sum_{i=0}^{N} \frac{(-1)^{i}}{i!} \, C^i \, \partial_u^i \delta \, G(u) \, (u- v)^k \, du = d \int_{\Omega_u} \sum_{i=0}^{N} \frac{(-1)^{2i}}{i!} \, C^i \,\delta \,  \partial_u^i \, \left( G(u) \, (u- v)^k \right) \, du \\ \notag
=  \int_{\Omega_u} \sum_{i=0}^{N} \frac{1}{i!} \, C^i \,\delta \, \sum_{z=0}^{i}  \binom{i}{z} \, \partial_u^{i-z} G(u) \, \partial_u^{z} (u- v)^k  \, du \\ \notag
= \int_{\Omega_u} \sum_{i=0}^{N} \frac{1}{i!} \, C^i \,\delta \, \sum_{z=0}^{i}  \binom{i}{z} \, \partial_u^{i-z} G(u) \, \frac{k!}{(k-z)!} (u- v)^{k-z}  \, du \\ \notag
=\sum_{i=0}^{N} \sum_{z=0}^{i}  \frac{1}{i!} \, C^i  \,  \binom{i}{z} \partial_v^{i-z} G( v) \, \frac{k!}{(k-z)!} (0)^{k-z}  ,
\end{align}
Here we used the general Leibniz rule to express the derivative of product as product of derivatives. In the last step, we applied the sifting property of $\delta$. This expression is non-vanishing only for $z=k$. Since $z$ runs up to $i$, the term vanishes for $i<k$ while for $i \geq k$, we set $z=k$ and obtain
\begin{align}
=\sum_{i=k}^{N}  \frac{1}{i!} \, C^i  \,  \binom{i}{k} \partial_v^{i-k} G( v) \, \frac{k!}{(k-k)!}  = \partial_x \sum_{i=k}^N \frac{1}{(i-k)!} \, C^i \,  \partial_v^{i-k} G( v).
\end{align}
After simplifying the remaining terms in a similar way, we obtain
\begin{align} \label{eq:dtCk_1d}
&\partial_t C^k - k \,   \, C^{k-1} \, G( v) d v \\ \notag
&+ \partial_x \sum_{i=k}^{N} \frac{1}{(i-k)!} C^i \,\, \partial_v^{i-k} G( v) + k\,   \, \sum_{i=k-1}^{N} \frac{1}{(i-(k-1))!} C^i \,\, \partial_v^{i-(k-1)} G( v) \, d v = 0.
\end{align}
As before, the second and the last terms are vanishing for $k=0$.

\textit{Step 2: Simplify the source term $\partial_x \left( G(u) g_N \right)$ }

We simplify in the sense of distributions and use tensor product test functions to ease the derivation. This is justified since the span of tensor product test functions $\phi_u \, \phi_x \, \phi_t$ is dense in the space of test functions \( \phi \in C^{\infty}_0(\Omega)\).

\begin{align}
\int \partial_x \left( G(u) g_N \right) \, \phi_u \, \phi_x \, \phi_t \, dV = -\int G(u) \, g_N \, \phi_u \, \partial_x \, \phi_x \, \phi_t \, dV \\ \notag
= - \int \sum_{i=0}^N \frac{(-1)^i}{i!} C^i \, \partial_u^i \delta \, G(u) \phi_u \, \partial_x \phi_x \, \phi_t \, dV \\ \notag
= - \int \sum_{i=0}^N \frac{(-1)^{2i}}{i!} C^i \, \delta \, \sum_{z=0}^i \binom{i}{z} \partial_u^{i-z} G(u) \, \partial_u^{z} \phi_u \, \partial_x \phi_x \, \phi_t \, dV \\ \notag
= - \int \sum_{i=0}^N \sum_{z=0}^i  \frac{(-1)^z}{i!} C^i \, \partial_u^{z}  \delta \, \binom{i}{z} \partial_v^{i-z} G( v) \, \phi_u \, \partial_x \phi_x \, \phi_t \, dV \\ \notag 
= - \int \sum_{i=0}^N \sum_{z=0}^i \, \frac{(-1)^z}{z!} \frac{1}{(i-z)!} C^i \, \partial_u^{z}  \delta \, \partial_v^{i-z} G( v) \, \phi_u \, \partial_x \phi_x \, \phi_t \, dV \\ \notag
=  \int \sum_{i=0}^N \sum_{z=0}^i \, \frac{(-1)^z}{z!} \frac{1}{(i-z)!} \partial_x \left( C^i \, \partial_v^{i-z} G( v) \, \partial_u^{z}  \delta  \right)  \, \phi_u \, \phi_x \, \phi_t \, dV \\ \notag 
=  \int \sum_{i=0}^N \sum_{z=0}^i \, \frac{(-1)^z}{z!} \frac{1}{(i-z)!} \partial_x \left( C^i \, \partial_v^{i-z} G( v) \right) \, \partial_u^{z}  \delta  \, \phi_u \, \phi_x \, \phi_t \, dV \\ \notag
-  \int \sum_{i=0}^N \sum_{z=0}^i \, \frac{(-1)^z}{z!} \frac{1}{(i-z)!} C^i \, \partial_v^{i-z} G( v) \, \partial_u^{z+1}  \delta \,   \, \partial_x v  \, \phi_u \, \phi_x \, \phi_t \, dV . \\ \notag 
\end{align}

\textit{Step 3:} Substituting $\partial_t C^k$ from \textit{Step 1}, $\partial_x \left( G(u) g_N\right)$ from \textit{Step 2}, and also $\partial_t v=-G( v) \partial_x v$ into \eqref{eq:main_eq_G}, we obtain
\begin{align} \label{eq:G_line_1}
\eqref{eq:main_eq_G} =  \sum_{k=1}^N \frac{(-1)^k}{k!} \left[ k \,   \, C^{k-1} \, G( v) \partial_x v \right] \partial_u^k \delta + \sum_{k=0}^N \frac{(-1)^k}{k!} \left[   \, C^k \, G( v) \partial_x v \right] \partial_u^{k+1} \, \delta \\ \label{eq:G_line_2}
- \sum_{k=0}^N \sum_{i=k}^N \frac{(-1)^k}{k!(i-k)!}  \partial_x \left( C^i \, \, \partial_v^{i-k} G( v) \right) \partial_u^k \delta + \sum_{i=0}^N \sum_{k=0}^i \frac{(-1)^k}{k!(i-k)!} \partial_x \left( C^i \, \, \partial_v^{i-k} G( v) \right) \partial_u^k \delta \\ \label{eq:G_line_3}
- \sum_{k=1}^N \sum_{i=k-1}^N \frac{(-1)^k}{k!(i-(k-1))!} \left[ k \,   \, C^i \,\,  \partial_v^{i-(k-1)}G( v) \partial_x v \right] \partial_u^k \delta \\ \label{eq:G_line_4}
- \sum_{i=0}^N \sum_{k=0}^i \frac{(-1)^k}{k!(i-k)!} \left[   \, \, C^i \,\, \partial_v^{i-k} G( v) \partial_x v \right] \partial_u^{k+1}\delta,
\end{align}
here in some terms the summation runs from $k=1$, as these terms are vanishing for $k=0$.

The terms in \eqref{eq:G_line_2} vanish due to the identity $\sum_{k=0}^N \sum_{i=k}^N a_{ik} = \sum_{i=0}^N \sum_{k=0}^i a_{ik}$. To simplify \eqref{eq:G_line_1}, we shift the summation index in the second term by one; with this shift the sum simplifies to \\ $\frac{(-1)^N}{N!}   \, C^N \, G( v) \, \, \partial_x v \, \, \partial_u^{N+1} \delta$. The sum of \eqref{eq:G_line_3} and \eqref{eq:G_line_4} gives $-\frac{(-1)^N}{N!}   C^N \, G( v) \, \partial_x v \, \partial_u^{N+1} \delta$, which cancels with the remaining term from \eqref{eq:G_line_1}. This completes the proof of $\partial_t g_N + \partial_x ( G(u) g_N)=0$.

The second case, when $G$ and $R$ are zero, is handled similarly. We skip the derivation and write down the final expression only:
\begin{align} \label{eq:F_line_1}
\partial_t g_N + \partial_u (F(u)g_N) =-\sum_{k=1}^N \frac{(-1)^k}{k!} k\, C^{k-1} \, F( v) \, \partial_u^k \delta - \sum_{k=0}^N \frac{(-1)^k}{k!} \, C^{k} \, F( v) \, \partial_u^{k+1} \delta \\ \label{eq:F_line_2}
+\sum_{k=1}^N \sum_{i=k-1}^N \frac{(-1)^k}{k!} \frac{k}{(i-(k-1))!} C^i \, \,  \partial_v^{i-(k-1)} F( v) \partial_u^k \delta \\ \label{eq:F_line_3}
- \sum_{i=0}^N \sum_{z=0}^i \frac{(-1)^{z+1}}{z!} \frac{1}{(i-z)!} C^i \partial_v^{i-z} F( v) \partial_u^{z+1} \delta .
\end{align}  
To show that this sum cancels, we note that \eqref{eq:F_line_1} simplifies to $-\frac{(-1)^N}{N!} C^N F( v) \partial_u^{N+1} \delta$ while the sum of \eqref{eq:F_line_2} and \eqref{eq:F_line_3} gives $-\frac{(-1)^{N+1}}{N!} C^N F( v) \partial_u^{N+1} \delta$. So that the whole expression vanishes.

The third case, when $G$ and $F$ are zero, yields an expression:
\begin{align}
\partial_t g_N+ R(u)g_N &= \\ \notag
&-\sum_{k=0}^N \sum_{i=k}^N \frac{(-1)^k}{k!}\frac{1}{(i-k)!} C^i \partial_v^{i-k} R( v) \partial_u^k \delta + \sum_{i=0}^N \sum_{k=0}^i \frac{(-1)^k}{k!} \frac{1}{(i-k)!} C^i  \partial_v^{i-k} R( v) \partial_u^k \delta,
\end{align}
that also vanishes. This completes the proof that $\partial_t g_N + \partial_x (G(u)g_N) + \partial_u (F(u)g_N) + Rg_N = 0$.
\end{proof}

\subsection{Proof of Theorem \ref{prop:main_theorem}} \label{sec:proof_main_theorem}

The proof of Theorem \ref{prop:main_theorem} in an arbitrary number of dimensions resembles closely the one-dimensional case. The main difference is that we have to use multi-index notation $\alpha = (\alpha_1,..., \alpha_d)$ with $|\alpha|=\sum_{i}^d \alpha_i $ to express tensors and derivatives. First, let us rewrite the tensor contraction. Recall that a component of a rank-$a$ tensor is denoted by $C^a_{k_1...k_a}$. The indices $(k_1...k_a)$ can be represented by multi-index $\alpha=(\alpha_1...\alpha_d)$, where $\alpha_i$ denotes the number of repeated indices in the direction $i$. For example, a fourth-order tensor component $(1,1,2,1)$ in three dimensions can be represented by the multi-index $\alpha=(3,1,0)$. However, this representation is unique only up to the permutation of indices, for instance, $(1,1,2,1)$ and $(1,2,1,1)$, along with any other permutation, correspond to the same multi-index $(3,1,0)$. This is not a problem, since the definition of moment tensors ensures their symmetry under the index permutations, with the total number of distinct permutations given by $\frac{a!}{\alpha!}$, where $\alpha!=\prod_i^d \alpha_i!$. Consequently, the contraction of moment tensors with another tensor that is also symmetric under permutations, such as gradients of $\delta$, can still be expressed using multi-indices as follows $\sum_{k_1=1}^d...\sum_{k_a=1}^d C^a_{k_1...k_a} \partial_u^{k_1} ... \partial_u^{k_a} \delta = \sum_{|\alpha|=a} \frac{a!}{\alpha!} C^a_{\alpha} \, \partial_u^{\alpha} \delta$ where the sum is over all $\alpha$ such that $|\alpha|=a$.
In this notation the ansatz \eqref{eq:burby_ansatz} reads
\begin{align} \label{eq:ansatz_mi}
g_N = \sum_{a=0}^N \sum_{|\alpha|=a} \frac{(-1)^a}{\alpha!} C^a_{\alpha} \partial_u^{\alpha} \delta.
\end{align}

\begin{proof}
As in the one-dimensional case, it suffices to prove the theorem for three distinct cases. We consider the case $\partial_t g_N + \nabla \cdot(\bm G(\bm u) g_N)=0$ which is the most intricate of the three cases. Substituting the ansatz \eqref{eq:ansatz_mi} the case reads
\begin{align} \label{eq:main_eq_G_dim}
 \sum_{k=0}^N \sum_{|\kappa|=k} \frac{(-1)^k}{\kappa!} \partial_t C^k_{\kappa} \partial_u^{\kappa} \delta +  \sum_{k=0}^N \sum_{|\kappa|=k} \sum_{|\theta|=1}  (-1)  \frac{(-1)^k}{\kappa!} C^k_{\kappa} \partial_u^{\kappa + \theta} \delta \, \partial_t v_{\theta} + \nabla \cdot(\bm G(\bm u) g_N) = 0
\end{align}

\textit{Step1: Find expression for $\partial_t C^k$}

From \eqref{eq:centered_hierarchy} we obtain
\begin{align} \notag
\partial_t C^k_{\kappa} =& \int g_N \partial_t \left[ (u- v)^{\kappa} \right] \, du - \int \sum_{|\omega|=1} \partial_x^{\omega}\left[ G_{\omega}(u) \, g_N \right] (u- v)^{\kappa} \, du \\ \notag
=&\int g_N \sum_{|\theta|=1}  (-1) \, \partial_u^{\theta} \left[ (u- v)^{\kappa} \right] \partial_t v_{\theta}  \, du - \int \sum_{|\omega|=1} \partial_x^{\omega}\left[ G_{\omega}(u) \, g_N  (u- v)^{\kappa} \right] \, du \\ \notag
 &-\int \sum_{|\omega|=1} (-1)^{|\omega|} \, G_{\omega}(u) \, g_N \, \partial_x^{\omega} \left[ (u- v)^{\kappa} \right] \, du \\ \label{eq:first_term}
 = & \int g_N \, \sum_{|\theta|=1} \sum_{|\omega|=1}   \frac{\kappa!}{(\kappa - \theta)!} \left( u- v\right)^{\kappa-\theta} G_{\omega}( v) \partial_x^{\omega} v_{\theta} \, du \\ \label{eq:third_term}
 &- \int \sum_{|\omega|=1} \partial_x^{\omega}\left[ G_{\omega}(u) \, g_N  (u- v)^{\kappa} \right] \, du 
 \\ \label{eq:second_term}
 &+ \int \sum_{|\omega|=1} \sum_{|\theta=1|} \,  (- ) \, G_{\omega}(u) g_N \frac{\kappa!}{(\kappa-\theta)!} \left( u- v \right)^{\kappa-\theta} \partial_x^{\omega} v_{\theta} \, du  .
\end{align}
We observe that \eqref{eq:first_term} and \eqref{eq:second_term} vanish when $|\kappa|=0$ or if $\theta_i > \kappa_i$ for some $i$. Substituting the ansatz \eqref{eq:ansatz_mi} and performing integration by parts over $u$-variable, we obtain expressions that contain only $\delta$. To illustrate this we consider the case \eqref{eq:third_term}.
\begin{align} 
 \eqref{eq:third_term} &=- \sum_{|\omega|=1} \partial_x^{\omega} \int \sum_{a=0}^N \sum_{|\alpha|=a} \frac{(-1)^a}{\alpha!} C^a_{\alpha} \, \partial_u^{\alpha} \delta \, G_{\omega}(u) \, (u- v)^{\kappa} \, du \\ \notag
 &=- \sum_{|\omega|=1} \partial_x^{\omega} \int \sum_{a=0}^N \sum_{|\alpha|=a} \frac{1}{\alpha!} C^a_{\alpha} \,  \delta \, \sum_{\beta \leq \alpha} \binom{\alpha}{\beta} \partial_u^{\alpha-\beta} G_{\omega}(u) \, \partial_u^{\beta} \left[ (u- v)^{\kappa} \right] \, du \\ \notag
  &=- \sum_{|\omega|=1} \partial_x^{\omega} \int \sum_{a=0}^N \sum_{|\alpha|=a}  \sum_{\beta \leq \alpha} \frac{1}{\alpha!} \binom{\alpha}{\beta} C^a_{\alpha}  \, \delta \, \partial_u^{\alpha-\beta} G_{\omega}(u) \, \frac{\kappa!}{(\kappa-\beta)!}  (u- v)^{\kappa-\beta} \, du \\ \label{eq:45_last}
   &=- \sum_{|\omega|=1} \partial_x^{\omega} \left[ \sum_{a=0}^N \sum_{|\alpha|=a}  \sum_{\beta \leq \alpha} \frac{1}{\alpha!} \binom{\alpha}{\beta} C^a_{\alpha}   \, \partial_v^{\alpha-\beta} G_{\omega}( v) \, \frac{\kappa!}{(\kappa-\beta)!}  (0)^{\kappa-\beta} \right].
\end{align}
Here we used the general Leibniz rule in multi-dimensions to express the derivative of product as product of derivatives: $\partial^{\alpha} \left( a^{\mu} b^{\nu} \right) = \sum_{\beta \leq \alpha} \binom{\alpha}{\beta} \partial^{\alpha-\beta} \, a^{\mu} \, \partial^{\beta} \, b^{\nu}$.

The term in \eqref{eq:45_last} is non-zero only when $\beta=\kappa$. Since $\beta \leq \alpha$, it follows that $\alpha \geq \kappa$ and also $a \geq k$, where recall $a=|\alpha|$ and $k=|\kappa|$. We therefore get
\begin{align}
 \eqref{eq:third_term}  &=- \sum_{|\omega|=1}\sum_{a=k}^N \sum_{\substack{|\alpha|=a \\ \alpha \geq \kappa}}  \frac{1}{(\alpha-\kappa)!}  \partial_x^{\omega} \left[  C^a_{\alpha}   \, \partial_u^{\alpha-\kappa} G_{\omega}(u) \, \right],
\end{align}
simplifying the remaining terms in the same way, we obtain
\begin{align} \notag
\partial_t C^k = \sum_{|\theta|=1} \sum_{|\omega|=1}   \frac{\kappa!}{(\kappa-\theta)!} \, C^{k-1}_{\kappa-\theta} \, G_{\omega}( v) \, \partial_x^{\omega}v_{\theta} - \sum_{|\omega|=1}\sum_{a=k}^N \sum_{\substack{|\alpha|=a \\ \alpha \geq \kappa}}  \frac{1}{(\alpha-\kappa)!}  \partial_x^{\omega} \left[  C^a_{\alpha}   \, \partial_v^{\alpha-\kappa} G_{\omega}(u) \, \right] \\ \label{eq:dtCk}
\sum_{|\theta|=1} \sum_{|\omega|=1} \sum_{a=k-1}^N \sum_{\substack{|\alpha|=a \\ \alpha \geq \kappa - \theta}} (-1) \frac{\kappa!}{(\alpha - (\kappa-\theta))!(\kappa-\theta)!} C^{a}_{\alpha} \, \partial_v^{\alpha-(\kappa-\theta)}G_{\omega}( v) \, \partial_x^w v_{\theta}.
\end{align}
As before, the first and the last terms vanish when $|\kappa|=k=0$ or when there exists an $i$ such that $\theta_i > \kappa_i$. We also observe the similarity of the expression to the one-dimensional case in \eqref{eq:dtCk_1d}.

\textit{Step 2: Simplifying the source term $\nabla \cdot (\bm G(\bm u) g_N)$}

The simplification follows the steps outlined in the one-dimensional case. We use the multi-index notation, and also the integration by parts and the Leibniz product in a very much the same way as when deriving (\ref{eq:first_term}- \ref{eq:second_term}) in Step 1. We state the result directly:
\begin{align}
\int &\sum_{|\omega|=1} \partial_x^{\omega} \left( G_{\omega} g_N \right) \phi \, dV \\ \notag
&= \int \sum_{|\omega|=1} \sum_{a=0}^N \sum_{|\alpha|=a} \sum_{\kappa \leq \alpha} \frac{(-1)^{\kappa}}{(\alpha-\kappa)! \, \kappa!} \, \partial_x^{\omega} \left[ C^{a}_{\alpha} \partial_u^{\alpha-\kappa} \, G_{\omega}( v) \right] \partial_u^{\kappa} \delta \, \phi \, dV
\\ \notag
&+\int \sum_{|\omega|=1} \sum_{a=0}^N \sum_{|\alpha|=a} \sum_{\kappa \leq \alpha} \sum_{|\theta|=1} (-1) \frac{(-1)^{\kappa}}{(\alpha-\kappa)!(\kappa)!} C^a_{\alpha} \, \partial_v^{\alpha-\kappa} G_{\omega}( v) \, \partial_x^{\omega} v_{\theta} \, \partial_u^{\kappa+\theta} \delta \, \phi \, dV . 
\end{align}

\textit{Step 3:} Substituting $\partial_t C^k_{\kappa}$ from \textit{Step 1}, $\nabla \cdot (\bm G(\bm u) g_N)$ from \textit{Step 2}, and $\partial_t v_\theta = - \sum_{|\omega|=1} G_{\omega}( v) \partial_x^{\omega} v_{\theta}$ into \eqref{eq:main_eq_G_dim} we obtain
\begin{align} \label{eq:tr_one}
\partial_t & g_N + \nabla \cdot \left( \bm G(\bm u) g_N \right) =  
\sum_{k=1}^N \sum_{|\theta|=1}  \sum_{\substack{|\kappa|=k \\ \kappa_i \geq \theta_i}}  \sum_{|\omega|=1}   \frac{(-1)^k}{(\kappa-\theta)!} \, C^{k-1}_{\kappa-\theta} \, G_{\omega}( v) \, \partial_x^{\omega}v_{\theta}   \partial_u^{\kappa} \delta \\\label{eq:tr_two}
&-\sum_{k=0}^N \sum_{|\kappa|=k}  \sum_{|\omega|=1}\sum_{a=k}^N \sum_{\substack{|\alpha|=a \\ \alpha \geq \kappa}}   \frac{(-1)^k}{\kappa!} \frac{1}{(\alpha-\kappa)!}  \partial_x^{\omega} \left[  C^a_{\alpha}   \, \partial_v^{\alpha-\kappa} G_{\omega}(u) \, \right]   \partial_u^{\kappa} \delta \\ \label{eq:tr_three}
& +\sum_{|\omega|=1} \sum_{|\theta|=1}  \sum_{k=1}^N \sum_{\substack{|\kappa|=k \\ \kappa_i \geq \theta_i}}  \sum_{a=k-1}^N \sum_{\substack{|\alpha|=a \\ \alpha \geq \kappa - \theta}} (-1) \frac{(-1)^{k}}{(\alpha - (\kappa-\theta))!(\kappa-\theta)!} C^{a}_{\alpha} \, \partial_v^{\alpha-(\kappa-\theta)}G_{\omega}( v) \, \partial_x^w v_{\theta} \, \, \partial_u^{\kappa} \delta \\ \label{eq:tr_four}
& +\sum_{k=0}^N \sum_{|\kappa|=k} \sum_{|\theta|=1} \sum_{|\omega|=1} \frac{(-1)^k}{\kappa!} \,   \, C^k_{\kappa} \, G_{\omega}( v) \, \partial_x^{\omega} v_{\theta} \, \partial_u^{\kappa+\theta} \delta \\ \label{eq:tr_five}
&+\sum_{|\omega|=1} \sum_{a=0}^N \sum_{|\alpha|=a} \sum_{\kappa \leq \alpha} \frac{(-1)^k}{(\alpha-\kappa)! \, \kappa!} \, \partial_x^{\omega} \left[ C^{a}_{\alpha} \partial_v^{\alpha-\kappa} \, G_{\omega}( v) \right] \partial_u^{\kappa} \delta  \\ \label{eq:tr_six}
&+\sum_{|\omega|=1} \sum_{a=0}^N \sum_{|\alpha|=a} \sum_{\kappa \leq \alpha} \sum_{|\theta|=1} (-m) \frac{(-1)^k}{(\alpha-\kappa)!(\kappa)!} C^a_{\alpha} \, \partial_v^{\alpha-\kappa} G_{\omega}( v) \, \partial_x^{\omega} v_{\theta} \, \partial_u^{\kappa+\theta} \delta .
\end{align}
In view of the note under \eqref{eq:dtCk}, the summation in \eqref{eq:tr_one} and \eqref{eq:tr_three} runs from $k=1$ since for $k=0$ these terms vanish. Since they also vanish whenever $\theta_i > \kappa_i$, we introduced a condition $\kappa_i \geq \theta_i$ under the summation signs.

\begin{itemize}
\item As in the one-dimensional case, we identify terms that cancel out. We first note the equality
\begin{align} \label{eq:sum_identity}
\sum_{k=0}^a \sum_{|\kappa|=k} \sum_{\substack{|\alpha|=a \\ \alpha \geq \kappa}} f(\alpha, \kappa) = \sum_{|\alpha|=a} \sum_{\kappa \leq \alpha} f(\alpha, \kappa),
\end{align}
the equality holds since any pair $(\alpha, \kappa)$ of the sum satisfies the conditions
\begin{align}
|\alpha|=a \text{ and } \kappa \leq \alpha \text{ and } |\kappa|=k \text{ and } 0 \leq k \leq a  & \iff |\alpha|=a \text{ and } \kappa \leq \alpha \text{ and } |\kappa| \leq a=|\alpha| \\
& \iff |\alpha|=a \text{ and } \kappa \leq \alpha .
\end{align}
Applying the equality to the term in \eqref{eq:tr_five}, the term is neutralized with \eqref{eq:tr_two}.

\item Next, we simplify the sum of \eqref{eq:tr_three} and \eqref{eq:tr_six}. First, we use the identity \eqref{eq:sum_identity} and identity \\ $\sum_{a=0}^N \sum_{k=0}^a c_{ak}= \sum_{k=0}^N \sum_{a=k}^N c_{ak}$ to rewrite \eqref{eq:tr_six} as follows
\begin{align}
&\sum_{|\omega|=1} \sum_{k=0}^N \sum_{a=k}^N \sum_{|\kappa|=k} \sum_{\substack{|\alpha|=a \\ \alpha \geq \kappa}} \sum_{|\theta|=1} (-1) \frac{(-1)^k}{(\alpha-\kappa)!(\kappa)!} C^a_{\alpha} \, \partial_v^{\alpha-\kappa} G_{\omega}( v) \, \partial_x^{\omega} v_{\theta} \, \partial_u^{\kappa+\theta} \delta ,
\end{align} 
then we factor out the summation over $\theta$ and relabel $\kappa$ to $\xi-\theta$. Since by assumption $\kappa_i \geq 0$ for any $i$, we need to be careful and write this assumption explicitly, i.e., we require $\xi_i \geq \theta_i $. We also relabel $k$ to $z-1$ and obtain
\begin{align} \label{eq:simplified_tr_six}
\sum_{|\omega|=1} \sum_{|\theta|=1}  \sum_{z=1}^{N+1} \sum_{a=z-1}^N \sum_{\substack{|\xi-\theta|=z-1 \\ \xi_i \geq \theta_i }} \sum_{\substack{|\alpha|=a \\ \alpha \geq \xi-\theta}} (-1) \frac{(-1)^{z-1}}{(\alpha-(\xi-\theta))!(\xi-\theta)!} C^a_{\alpha} \, \partial_v^{\alpha-(\xi-\theta)} G_{\omega}( v) \, \partial_x^{\omega} v_{\theta} \, \partial_u^{\xi} \delta .
\end{align}  
Since $|\theta|=1$ and $\xi_i \geq \theta_i$ we have $|\xi-\theta|=z-1 \iff |\xi|=z$. We observe that \eqref{eq:simplified_tr_six} cancels with \eqref{eq:tr_three} except for $z=N+1$. For $z=N+1$, \eqref{eq:simplified_tr_six} has a term with $|\xi-\theta|=N$ and $|\alpha|=N$. But $\alpha \geq \xi-\theta$. Thus $\alpha=\xi-\theta$. The extra term reads
\begin{align} \label{eq:remaining_term}
\sum_{|\omega|=1} \sum_{|\theta|=1} \sum_{\substack{|\xi-\theta|=N \\ \xi_i \geq \theta_i }} (-) \frac{(-1)^{N}}{(\xi-\theta)!} C^N_{\xi-\theta} \,  G_{\omega}( v) \, \partial_x^{\omega} v_{\theta} \, \partial_u^{\xi} \delta .
\end{align}

\item We now simplify the sum of \eqref{eq:tr_two} and \eqref{eq:tr_five}. Their sum leaves the following extra term:
\begin{align}
\sum_{|\omega|=1} \sum_{|\theta|=1} \sum_{\substack{|\xi-\theta|=N \\ \xi_i \geq \theta_i }}   \frac{(-1)^{N}}{(\xi-\theta)!} C^N_{\xi-\theta} \,  G_{\omega}( v) \, \partial_x^{\omega} v_{\theta} \, \partial_u^{\xi} \delta ,
\end{align}
this term cancels out with the the extra term \eqref{eq:remaining_term}. This completes the proof of the first case.
\end{itemize}
The analysis of the second case, when $\bm G=0$ and $R=0$ with the advection equation in the form of $\partial_t g_N + \nabla_u \cdot (\bm F g_N) = 0$, is analogous and we omit it. The analysis of the third case with $\bm G=0$ and $\bm F=0$ and with the advection equation in the form of $\partial_t g_N + \nabla_u \cdot (R g_N) = 0$ leads to the expression
\begin{align}
\partial_t g_N + R \, g_N = &-\sum_{k=0}^N \sum_{|\kappa|=k} \sum_{a=k}^N \sum_{\substack{|\alpha|=a \\ \alpha \geq \kappa}} \frac{(-1)^k}{\kappa!(\alpha-\kappa)!} \, C^{a}_{\alpha} \partial_v^{\alpha-\kappa} R( v) \partial_u^{\kappa}\delta \\
&+\sum_{a=0}^N \sum_{|\alpha|=a} \sum_{\kappa \leq \alpha} \frac{(-1)^k}{\kappa!(\alpha-\kappa)!} \, C^{a}_{\alpha}  \partial_v^{\alpha-\kappa} R( v) \partial_u^{\kappa}\delta,
\end{align}
this expression vanishes if we consider the identity $\sum_{a=0}^N \sum_{k=0}^a c_{ak} = \sum_{k=0}^N \sum_{a=k}^N c_{ak}$ along with the identity \eqref{eq:sum_identity}.
\end{proof}

\subsection{Proof of Theorem \ref{prop:main_theorem_particles}} \label{sec:proof_main_theorem_particles}
In this section we prove Theorem \ref{prop:main_theorem_particles}.
We refer to Appendix \ref{sec:proof_main_theorem} for the use of the multi-index notation. We use the multi-index notation to rewrite the particle ansatz \eqref{eq:particle_ansatz}
\begin{align}
g_N = \sum_{a=0}^N \sum_{|\alpha|=a} \frac{(-1)^a}{\alpha!} W^a_{\alpha} \partial^{\alpha} \delta(\bm z-\bm Z_p),
\end{align}
where, with a slight abuse of notation, the subscript $N$ denotes the degree of the ansatz.
To ease the notation, we will simply write $\delta := \delta(\bm z- \bm Z_p)$.
 
To show exactness of the moment system we show that $g_N$ satisfies
\begin{align}
\int \partial_t g_N \phi(\bm z, t) \, dz  + \int \nabla_z \cdot \left(\bm A(\bm z, t) g_N \right) \, \phi(\bm z, t) \, dz = 0 ,\hspace{1cm} \forall \phi \in C^{\infty}_0(\Omega).
\end{align}

First, we take the time derivative of the ansatz
\begin{align}
\partial_t g_N = \sum_{k=0}^N \sum_{|\kappa|=k} \frac{(-1)^k}{\kappa!} \partial_t W^k_{\kappa} \partial^{\kappa} \delta + \sum_{k=0}^N \sum_{|\kappa|=k} \frac{(-1)^k}{\kappa!} W^k_{\kappa} \partial^{\kappa} \partial_t \delta ,
\end{align}
the time derivative of the delta function can be expressed with the multi-index notation as follows $\partial_t \delta = \sum_{|I|=1} \partial^{I} \delta \, (-\partial_t Z^I_p)$. Using $\partial_t Z^I_p = A_I(Z_p,t)$, we obtain
\begin{align} \label{eq:dt_g_N}
\partial_t g_N = \sum_{k=0}^N \sum_{|\kappa|=k} \frac{(-1)^k}{\kappa!} \partial_t W^k_{\kappa} \partial^{\kappa} \delta - \sum_{k=0}^N \sum_{|\kappa|=k} \sum_{|I|=1} \frac{(-1)^k}{\kappa!} W^k_{\kappa} \partial^{\kappa+I} \delta \, A_I(Z_p,t) ,
\end{align}

\textit{Step 1: Calculate $\partial_t \bm W^k$} \\
This step is analogous to Step 1 in Appendix \ref{sec:proof_main_theorem}. The only difference is that we consider the phase space integrals to derive the equations for the moments; we state the result directly

\begin{align} \label{eq:w_time_rate}
\partial_t W^k_{\kappa} = \sum_{|\theta|=1} -\frac{\kappa!}{(\kappa-\theta)!} W^{k-1}_{\kappa-\theta} A_{\theta}(Z_p) + \sum_{|\theta|=1} \sum_{a=k-1}^{N} \sum_{\substack{|\alpha|=a \\ \alpha \geq \kappa - \theta}} \frac{\kappa!}{(\kappa-\theta)!(\alpha-(\kappa-\theta))!} W^a_{\alpha} \partial^{\alpha-(\kappa-\theta)} A_{\theta}(Z_p),
\end{align}
where both terms are non-zero iff $\kappa_i \geq \theta_i$ for all $i$. 

Now substituting \eqref{eq:w_time_rate} into \eqref{eq:dt_g_N} we obtain
\begin{align} \label{eq:dt_g}
\partial_t g_N & = \sum_{|\theta|=1} \sum_{k=1}^N  \sum_{a=k-1}^{N} \sum_{|\kappa|=k}  \sum_{\substack{|\alpha|=a \\ \alpha \geq \kappa - \theta}} \frac{(-1)^k }{(\kappa-\theta)!(\alpha-(\kappa-\theta))!} W^a_{\alpha} \partial^{\alpha-(\kappa-\theta)} A_{\theta}(Z_p) \partial^{\kappa} \delta \\ \notag
&+ \sum_{|\theta|=1} \sum_{k=1}^N \sum_{|\kappa|=k}  \frac{(-1)^{k+1}}{(\kappa-\theta)!} W^{k-1}_{\kappa-\theta} A_{\theta}(Z_p)  \partial^{\kappa} \delta 
- \sum_{|\theta|=1} \sum_{k=0}^N \sum_{|\kappa|=k}  \frac{(-1)^k}{\kappa!} W^k_{\kappa}   \, A_{\theta}(Z_p,t) \partial^{\kappa+\theta} \delta,
\end{align}
where the index in the first two terms begins at $k=1$ since these terms vanish for $k=0$. 
Next, we shift the range of indices in the last term to get
\begin{align}
\sum_{|\theta|=1} \sum_{k=0}^N \sum_{|\kappa|=k}  \frac{(-1)^k}{\kappa!} W^k_{\kappa}   \, A_{\theta}(Z_p,t) \partial^{\kappa+\theta} \delta = \sum_{|\theta|=1} \sum_{k=1}^{N+1} \sum_{|\kappa|=k}  \frac{(-1)^{k-1}}{(\kappa-\theta)!} W^{k-1}_{\kappa-\theta}   \, A_{\theta}(Z_p,t) \partial^{\kappa} \delta ,
\end{align}
the last two terms simplify so that \eqref{eq:dt_g} reads
\begin{align} \notag
\partial_t g_N & = \sum_{|\theta|=1} \sum_{k=1}^N  \sum_{a=k-1}^{N} \sum_{|\kappa|=k}  \sum_{\substack{|\alpha|=a \\ \alpha \geq \kappa - \theta}} \frac{(-1)^k }{(\kappa-\theta)!(\alpha-(\kappa-\theta))!} W^a_{\alpha} \partial^{\alpha-(\kappa-\theta)} A_{\theta}(Z_p) \partial^{\kappa} \delta \\ \label{eq:final_expression_dtg}
& - \sum_{|\theta|=1} \sum_{|\kappa|=N+1}  \frac{(-1)^N}{(\kappa-\theta)!} W^N_{\kappa-\theta}   \, A_{\theta}(Z_p,t) \partial^{\kappa} \delta,
\end{align}

\textit{Step 2: Simplify $\nabla_z \cdot \left( \bm A g_N \right)$} \\

Let us now consider the advection term. The simplification here is similar to Step 2 in Appendix \ref{sec:proof_main_theorem}. We recall that $\phi = \phi(\bm z, t) \in C^{\infty}_0(\Omega)$.
\begin{align} \notag
\int \nabla_z \cdot \left( \bm A(\bm z, t) g_N \right) \, \phi \, dz &= \int \sum_{|I|=1} \partial^I \left( A_I \, \sum_{a=0}^N \sum_{|\alpha|=a} \frac{(-1)^a}{\alpha!} W^a_{\alpha} \partial^{\alpha} \delta \right) \, \phi \, dz \\ \notag
& =-\int \sum_{|I|=1} A_I \, \sum_{a=0}^N \sum_{|\alpha|=a} \frac{(-1)^a}{\alpha!} W^a_{\alpha} \partial^{\alpha} \delta \, \partial^I  \phi \, dz \\ \notag
& =-\int \sum_{|I|=1} \, \sum_{a=0}^N \sum_{|\alpha|=a} \frac{(-1)^{2a}}{\alpha!} W^a_{\alpha} \delta \, \partial^{\alpha} \left( A_I \, \partial^I  \phi \right) \, dz \\ \notag
& =-\int \sum_{|I|=1} \, \sum_{a=0}^N \sum_{|\alpha|=a} \frac{1}{\alpha!} W^a_{\alpha} \delta \, \sum_{\beta \leq \alpha} \binom{\alpha}{\beta} \partial^{\alpha-\beta} A_I \, \partial^{\beta+I}   \phi  \, dz \\ \label{eq:advection_int}
& =\int \sum_{|I|=1} \, \sum_{a=0}^N \sum_{|\alpha|=a} \sum_{\beta \leq \alpha} \frac{(-1)^{b}}{(\alpha-\beta)! \beta!} W^a_{\alpha} \,  \partial^{\beta+I} \delta \,  \partial^{\alpha-\beta} A_I(Z_p)  \,   \phi  \, dz,
\end{align}
here in the last step we used the general Leibniz product rule $\partial^{\alpha} \left( a^{\mu} b^{\nu} \right) = \sum_{\beta \leq \alpha} \binom{\alpha}{\beta} \partial^{\alpha-\beta} \, a^{\mu} \, \partial^{\beta} \, b^{\nu}$.
Next, we invoke the identities $\sum_{|\alpha|=a} \sum_{\beta \leq \alpha} f(\alpha, \beta) = \sum_{b=0}^{a} \sum_{|\beta|=b} \sum_{\substack{ \alpha \geq \beta \\ |\alpha|=a} } f(\alpha, \beta)$ and $\sum_{a=0}^N \sum_{b=0}^a f(a,b)=\sum_{b=0}^N \sum_{a=b}^N f(a,b)$ 
so that the integrant of \eqref{eq:advection_int} reads
\begin{align}
\sum_{|\theta|=1} \, \sum_{b=0}^N  \sum_{a=b}^{N} \sum_{|\beta|=b} \sum_{\substack{ \alpha \geq \beta \\ |\alpha|=a} }  \frac{(-1)^{b}}{(\alpha-\beta)! \beta!} W^a_{\alpha} \,  \partial^{\alpha-\beta} A_{\theta}(Z_p)  \, \partial^{\beta+\theta} \delta ,
\end{align}
next we relabel $k=b+1$ and $\kappa = \beta + \theta$, and rewrite the sum
\begin{align} \notag
\sum_{|\theta|=1} \, \sum_{b=0}^N  \sum_{a=b}^{N} \sum_{|\beta|=b} \sum_{\substack{ \alpha \geq \beta \\ |\alpha|=a} } f(a, b, \alpha, \beta, \theta)  & = \sum_{|\theta|=1} \, \sum_{k-1=0}^N  \sum_{a=k-1}^{N} \sum_{|\kappa-\theta|=k-1} \sum_{\substack{ \alpha \geq \kappa-\theta \\ |\alpha|=a} } f(a, k-1, \alpha, \kappa-\theta, \theta) \\ \notag
&= \sum_{|\theta|=1} \, \sum_{k=1}^{N+1}  \sum_{a=k-1}^{N} \sum_{|\kappa|=k} \sum_{\substack{ \alpha \geq \kappa-\theta \\ |\alpha|=a} } f(a, k-1, \alpha, \kappa-\theta, \theta),
\end{align}
using this relation we obtain an expression
\begin{align}
\nabla_z \cdot \left( \bm A(\bm z, t) g_N \right) =  \sum_{|\theta|=1} \, \sum_{k=1}^{N+1}  \sum_{a=k-1}^{N} \sum_{|\kappa|=k} \sum_{\substack{ |\alpha|=a \\ \alpha \geq \kappa-\theta } }  \frac{(-1)^{k-1}}{(\alpha-(\kappa-\theta))! (\kappa-\theta)!} W^a_{\alpha} \,  \partial^{\alpha-(\kappa-\theta)} A_{\theta}(Z_p) \,  \partial^{\kappa} \delta ,
\end{align}
that we further split by splitting the sum $\sum_{k=1}^{N+1} = \sum_{k=1}^{N} + ...$
\begin{align} \notag
\nabla_z \cdot \left( \bm A(\bm z, t) g_N \right) &=  \sum_{|\theta|=1} \, \sum_{k=1}^{N}  \sum_{a=k-1}^{N} \sum_{|\kappa|=k} \sum_{\substack{ |\alpha|=a \\ \alpha \geq \kappa-\theta } }  \frac{(-1)^{k-1}}{(\alpha-(\kappa-\theta))! (\kappa-\theta)!} W^a_{\alpha} \,  \partial^{\alpha-(\kappa-\theta)} A_{\theta}(Z_p) \,  \partial^{\kappa} \delta ,\\ \label{eq:final_expression_div_A}
&+ \sum_{|\theta|=1} \sum_{\substack{ |\kappa|=N+1 }} \frac{(-1)^N}{(\kappa-\theta)!} W^N_{\kappa-\theta} A_{\theta}(\bm Z_p,t) \partial^{\kappa} \delta ,
\end{align}
summing up the terms in \eqref{eq:final_expression_dtg} and \eqref{eq:final_expression_div_A} we see that the sum vanishes and it follows that $\partial_t g_N + \nabla_z \cdot (\bm A \, g_N) = 0$.

\section{Acknowledgement}
We acknowledge financial support from the Deutsche Forschungsgemeinschaft (DFG) via the Collaborative Research Center SFB1491 Cosmic Interacting Matters—From Source to Signal.
\\ \\
\noindent \textbf{Author contributions} T.M. conceptualization, methodology, formal analysis, investigation, software, visualization, writing original draft.  K.K. conceptualization, methodology, funding,  supervision, editing the draft. 
\\ \\
\noindent \textbf{Data availibility} The datasets generated and analyzed during the current study are available from the corresponding author on reasonable request.
\\ \\ 
\noindent \textbf{Conflict of interest } The authors declare no Conflict of interest.

\bibliographystyle{plain}

{\footnotesize
\bibliography{paper}}


\end{document}